\documentclass[a4paper,UKenglish,cleveref,autoref]{lipics-v2019}
\pdfoutput=1

\usepackage{tikz}
\usetikzlibrary{intersections, calc, arrows, automata, positioning}
\newcommand{\fin}{{\mathit{fi}}}
\newcommand{\inn}{{\mathit{in}}}
\RequirePackage{tikz}
\usetikzlibrary{arrows,snakes}

\usepackage{pgf, verbatim}
\usetikzlibrary{arrows,automata}

\newcommand{\m}{{\sf m}}

\usepackage{amsmath}
\usepackage{mathtools}

\usepackage{amssymb}
\usepackage{listings}
\usepackage{rotating}
\usepackage{graphicx}
\usepackage{graphics}
\usepackage{subcaption}
\usepackage{changepage}
\allowdisplaybreaks

\usepackage{stmaryrd}
\usepackage{enumitem}
\usepackage[linesnumbered,ruled]{algorithm2e}
\usepackage{setspace}

\makeatletter
\def\hlinewd#1{%
	\noalign{\ifnum0=`}\fi\hrule \@height #1 \futurelet
	\reserved@a\@xhline}
\makeatother

\usepackage{thmtools,thm-restate}

\mathcode`<="4268
\mathcode`>="5269
\mathchardef\gr="213E
\mathchardef\ls="213C

\newcommand{\commenteq}[1]{\hspace{2em} [\mbox{#1}]}

\usepackage{bm}

\newcommand{\suchthat}{\;\ifnum\currentgrouptype=16 \middle\fi|\;}

\newcommand{\Dist}{\mathrm{Distr}}
\newcommand{\support}{\mathrm{support}}

\newcommand{\nat}{\mathbb{N}}
\newcommand{\A}{\mathcal{A}}
\newcommand{\B}{\mathcal{B}}
\newcommand{\C}{\mathcal{C}}
\newcommand{\D}{\mathcal{D}}

\newcommand{\M}{\mathcal{M}}
\newcommand{\V}{\mathcal{V}}
\newcommand{\Reals}{\mathbb{R}}

\newcommand{\vecb}{\mathbf{b}}
\newcommand{\vecc}{\mathbf{c}}
\newcommand{\vecv}{\mathbf{v}}
\newcommand{\vecx}{\mathbf{x}}
\newcommand{\vech}{\mathbf{h}}
\newcommand{\zero}{\mathbf{0}}
\newcommand{\one}{\mathbf{1}}
\newcommand{\size}{\text{size}}
\newcommand{\TVeqZERO}{\mathrm{TV}=0}
\newcommand{\TVneqZERO}{\mathrm{TV}\gr0}
\newcommand{\TVeqONE}{\mathrm{TV}=1}
\newcommand{\TVneqONE}{\mathrm{TV}\ls1}
\newcommand{\PBeqZERO}{\mathrm{PB}=0}
\newcommand{\PBneqZERO}{\mathrm{PB}\gr0}
\newcommand{\PBeqONE}{\mathrm{PB}=1}
\newcommand{\PBneqONE}{\mathrm{PB}\ls1}

\newcommand{\Run}{\mathit{Run}}
\newcommand{\tv}{\mathit{tv}}
\newcommand{\pb}{\mathit{pb}}

\newcommand{\ETR}{\sf \displaystyle \exists \mathbb{R}}

\newcommand{\pfun}{\mathrel{\ooalign{\hfil$\mapstochar\mkern5mu$\hfil\cr$\to$\cr}}}

\tikzstyle{vertex} = [draw,circle,minimum size=0.5cm,inner sep=0pt]

\definecolor{OliveGreen}{rgb}{0,0.6,0}

\title{Comparing Labelled Markov Decision Processes} 

\titlerunning{}%optional, please use if title is longer than one line

\author{Stefan Kiefer}{Department of Computer Science, University of Oxford, UK}{stekie@cs.ox.ac.uk}{https://orcid.org/0000-0003-4173-6877}{Supported by a Royal Society University Fellowship.}

\author{Qiyi Tang}{Department of Computer Science, University of Oxford, UK}{qiyi.tang@cs.ox.ac.uk}{https://orcid.org/0000-0002-9265-3011}{}

\authorrunning{S.\ Kiefer and Q.\ Tang}%TODO mandatory. First: Use abbreviated first/middle names. Second (only in severe cases): Use first author plus 'et al.'

\Copyright{Stefan Kiefer and Qiyi Tang}%  , TODO mandatory, please use full first names. LIPIcs license is "CC-BY";  http://creativecommons.org/licenses/by/3.0/

%TODO mandatory: Please choose ACM 2012 classifications from https://dl.acm.org/ccs/ccs_flat.cfm 

%TODO mandatory; please add comma-separated list of keywords

\category{}%optional, e.g. invited paper

\relatedversion{The paper is accepted to FSTTCS 2020.}%optional, e.g. full version hosted on arXiv, HAL, or other respository/website
\funding{}%optional, to capture a funding statement, which applies to all authors. Please enter author specific funding statements as fifth argument of the \author macro.

\acknowledgements{We thank the anonymous reviewers of this paper for their constructive feedback.}%optional

\nolinenumbers %uncomment to disable line numbering

\EventEditors{}
\EventNoEds{1}
\EventLongTitle{}
\EventShortTitle{}
\EventAcronym{}
\EventYear{}
\EventDate{}
\EventLocation{}
\EventLogo{}
\SeriesVolume{}
\ArticleNo{50}
\begin{document}\sloppy

\maketitle

\begin{abstract}
	A labelled Markov decision process is a labelled Markov chain with nondeterminism, i.e., together with a strategy a labelled MDP induces a labelled Markov chain.
	The model is related to interval Markov chains.
	Motivated by applications of equivalence checking for the verification of anonymity, we study the algorithmic comparison of two labelled MDPs, in particular, whether there exist strategies such that the MDPs become equivalent/inequivalent, both in terms of trace equivalence and in terms of probabilistic bisimilarity.
	We provide the first polynomial-time algorithms for computing memoryless strategies to make the two labelled MDPs inequivalent if such strategies exist.
	We also study the computational complexity of qualitative problems about making the total variation distance and the probabilistic bisimilarity distance less than one or equal to one.
\end{abstract} 

\section{Introduction} \label{section:introduction}
%\input{sections/introduction.tex}
%motivation
Given a model of computation (e.g., finite automata), and two instances of it, are they semantically equivalent (i.e., do they accept the same language)?
Such \emph{equivalence} problems can be viewed as a fundamental question for almost any model of computation.
As such, they permeate computer science, in particular, theoretical computer science.

In \emph{labelled Markov chains (LMCs)}, which are Markov chains whose states (or, equivalently, transitions) are labelled with an observable letter, there are two natural and very well-studied versions of equivalence, namely \emph{trace (or language) equivalence} and \emph{probabilistic bisimilarity}.

The \emph{trace equivalence} problem has a long history, going back to Sch\"utzenberger~\cite{Schutzenberger} and Paz~\cite{Paz71}
who studied \emph{weighted} and \emph{probabilistic} automata, respectively.
Those models generalize LMCs, but the respective equivalence problems are essentially the same.
It can be extracted from~\cite{Schutzenberger} that equivalence is decidable in polynomial time, using a technique based on linear algebra.
Variants of this technique were developed in \cite{Tzeng,DoyenHR08}.
%Tzeng~\cite{Tzeng96} considered the path-equivalence problem for nondeterministic
%automata which asks, given nondeterministic automata $A$~and~$B$, whether each
%word has the same number of accepting paths in~$A$ as in~$B$. He gives an {\sf NC} 
%algorithm\footnote{%
%The complexity class~{\sf NC} is the subclass of~{\sf P} containing those problems that can be solved in polylogarithmic parallel time (see, e.g., \cite{GHR95}).}
%for deciding path equivalence which can be straightforwardly adapted to yield 
%an {\sf NC} algorithm for equivalence of MCs.
%
More recently, the efficient decidability of the equivalence problem was exploited, both theoretically and practically, for the verification of probabilistic systems, see, e.g., \cite{KMOWW:CAV11,12KMOWW:CAV,Peyronnet12,Ngo13,Li15}.
In those works, equivalence naturally expresses properties such as obliviousness and anonymity, which are difficult to formalize in temporal logic.
In a similar vein, inequivalence can mean detectibility and the lack of anonymity.

\emph{Probabilistic bisimilarity} is an equivalence that was introduced by Larsen and Skou~\cite{LarsenS91}.
It is finer than trace equivalence, i.e., probabilistic bisimilarity implies trace equivalence.
A similar notion for Markov chains, called \emph{lumpability}, can be traced back at least to the classical text by Kemeny and Snell~\cite{KemenySnell76}.
Probabilistic bisimilarity can also be computed in polynomial time \cite{Baier1996,DerisaviHS03,ValmariF10}.
Indeed, in practice, computing the bisimilarity quotient is fast and has become a backbone for highly efficient tools for probabilistic verification such as \textsc{Prism}~\cite{Prism} and \textsc{Storm}~\cite{Storm}.

In this paper, we study equivalence problems for \emph{(labelled) Markov decision processes (MDPs)}, which are LMCs plus nondeterminism, i.e., each state may have several \emph{actions} (or ``moves'') one of which is chosen by a controller, potentially randomly.
An MDP and a controller \emph{strategy} together induce an LMC (potentially with infinite state space, depending on the complexity of the strategy).
The nondeterminism in MDPs gives rise to a spectrum of equivalence queries: one may ask about the existence of strategies for two given MDPs such that the induced LMCs become trace/bisimulation equivalent, or such that they become trace/bisimulation \emph{in}equivalent.
Another potential dimension of this spectrum is whether to consider general strategies or more restricted ones, such as memoryless or even memoryless deterministic (MD) ones.

In this paper, we focus on \emph{memoryless} strategies, for several reasons.
First, these questions for unrestricted strategies quickly lead to undecidability.
For example, in~\cite[Theorem~3.1]{FKS2020} it was shown that whether there exists a general strategy such that a given MDP becomes trace equivalent with a given LMC is undecidable.
Second, memoryless strategies are sufficient for a wide range of objectives in MDPs, and their simplicity means that even if it was known that a general strategy exists to accomplish (in)equivalence one might still wonder if there \emph{also} exists a memoryless strategy.
Third, probabilistic bisimilarity is a less natural notion for LMCs induced by general strategies: such LMCs will in general have an infinite state space, even when the MDP is finite.
Fourth, applying a memoryless strategy in an MDP is related to choosing an instance of an \emph{interval Markov chain (IMC)}.
IMCs are like Markov chains, but the transitions are labelled not with probabilities but with probability intervals.
IMCs were introduced by Jonsson and Larsen~\cite{JonssonL91} and have been well studied in verification-related domains~\cite{SenVA06,ChatterjeeSH08,DelahayeLLPW11,BenediktLW13,ChakrabortyKatoen15}, but also in areas such as systems biology, security or communication protocols, see, e.g.,~\cite{Delahaye15}.
Selecting a memoryless strategy in an MDP corresponds to selecting a probability from each interval (one out of generally uncountably many).
\emph{Parametric Markov chains} and \emph{parametric MDPs} are further related models, see, e.g., \cite{HahnHZ11,WinklerJPK19} and the references therein.

LMCs can also be compared in terms of their \emph{distance}.
We consider two natural distance functions between two LMCs: the \emph{total variation} distance (between the two trace distributions) and the \emph{probabilistic bisimilarity} distance~\cite{DesharnaisGJP04}.
Both distances can be at most~$1$.
The total variation (resp.\ probabilistic bisimilarity) distance is~$0$ if and only if the LMCs are trace equivalent (resp.\ probabilistic bisimilar).
Further, the probabilistic bisimilarity distance is an upper bound on the total variation distance~\cite{CBW2012}.
It was shown in~\cite{CK2014} (resp.~\cite{TangBreugel20}) that whether the total variation (resp.\ probabilistic bisimilarity) distance of two LMCs equals~$1$ can be decided in polynomial time.
This raises the question whether these results can be extended to MDPs, i.e., what is the complexity of deciding whether there exists a memoryless strategy to make the distance less than~$1$ or equal to~$1$, respectively.
It turns out that some of these problems are closely related to the corresponding (in)equivalence problem.

%notations of the problems.
Instead of comparing two MDPs with initial distributions/states, one may equivalently compare two initial distributions/states in a single MDP (by taking a disjoint union of the states).
In this paper we study the computational complexity of the following problems:
\begin{itemize}
	\item
	$\TVeqZERO$ ($\TVneqZERO$), which asks whether there is a memoryless strategy such that the two initial distributions are (not) trace equivalent in the induced labelled Markov chain; %ETR-complete (polynomial)
	\item
	$\TVeqONE$ ($\TVneqONE$),   which asks whether there is a memoryless strategy such that the two initial distributions (do not) have total variation distance one; %NP-hard NP?  (ETR-complete)
	\item
	$\PBeqZERO$ ($\PBneqZERO$),  which asks whether there is a memoryless strategy such that the two initial states are (not) probabilistic bisimilar; %NP-complete (polynomial)
	\item
	$\PBeqONE$ ($\PBneqONE$),    which asks whether there is a memoryless strategy such that the two initial states (do not) have probabilistic bisimilarity distance one. %NP-complete (NP-complete)
\end{itemize}
In \cref{section:tvInequivalence,section:pbInequivalence} we provide the first polynomial-time algorithms for $\TVneqZERO$ and $\PBneqZERO$, respectively.
We also show how to compute memoryless strategies that witness trace and probabilistic bisimulation inequivalence, respectively.
In \cref{section:DistanceOne} we discuss $\TVeqONE$ and $\PBeqONE$, and in \cref{section:summaryDistanceZeroAndNeqOne} we establish the complexity of the remaining four problems, which are about making the distance small ($=0$ or $\ls 1$). We conclude in \cref{section:conclusion}. \cref{tab:summary} summarises the results in the paper. Missing proofs can be found in the Appendix.
%related work 

\begin{table}[h]
	\centering
	\begin{tabular}{|c|c|}
		\hline
		Problem &  Complexity	\\
		\hline \hline
		$\TVeqZERO$ &  $\ETR$-complete\\
		$\TVneqZERO$ &  in ${\sf P}$ \\
		$\TVeqONE$ &  $\sf NP$-hard and in $\ETR$\\
		$\TVneqONE$ &   $\ETR$-complete\\
		$\PBeqZERO$ &  $\sf NP$-complete\\
		$\PBneqZERO$ & in   ${\sf P}$ \\
		$\PBeqONE$ &  $\sf NP$-complete\\
		$\PBneqONE$ &   $\sf NP$-complete\\
		\hline
	\end{tabular}	

	\caption{Summary of the results. These results also imply results for the problems which state ``for all memoryless strategies''. For example, $\TVneqZERO$ is the complement of the decision problem whether for all memoryless strategies the two initial distributions are  trace equivalent in the induced labelled Markov chains.} 
	\label{tab:summary}
\end{table}

\section{Preliminaries} \label{section:preliminaries}
 We write $\Reals$ for the set of real numbers and $\nat$ the set of nonnegative integers. Let $S$ be a finite set. We denote by $\Dist(S)$ the set of probability distributions on~$S$. By default we view vectors, i.e., elements of $\mathbb{R}^{S}$, as row vectors. For a vector $\mu \in [0, 1]^{S}$ we write $|\mu| := \sum_{s \in S} \mu(s)$ for its $L_1$-norm. A vector $\mu \in [0, 1]^{S}$ is a distribution (resp. subdistribution) over $S$ if $|\mu| = 1$ (resp.  $0 \ls |\mu| \le 1$). We denote column vectors by boldface letters; in particular, $\one \in \{1\}^{S}$ and  $\zero \in \{0\}^{S}$ are column vectors all whose entries are $1$ and $0$, respectively. For $s \in S$ we write $\delta_s$ for the (Dirac) distribution over $S$ with $\delta_{s}(s) = 1$ and $\delta_{s}(r) = 0$ for $r \in S \setminus \{s\}$. For a (sub)distribution $\mu$ we write $\support(\mu) = \{s \in S \mid \mu(s) \gr 0 \}$ for its support.

A \emph{labelled Markov chain} (LMC) is a quadruple $<S, L, \tau, \ell>$ consisting of a nonempty finite set $S$ of states, a nonempty finite set $L$ of labels, a transition function $\tau : S \to \Dist(S)$, and a labelling function $\ell: S \to L$.

We denote by $\tau(s)(t)$ the transition probability from $s$ to $t$. Similarly, we denote by $\tau(s)(E) = \sum_{t \in E} \tau(s)(t)$ the transition probability from $s$ to $E \subseteq S$. A trace in a LMC $\mathcal{M}$ is a sequence of labels $w = a_1a_2 \cdots a_n$ where $a_i \in L$. We denote by $L^{\le n}$ the set of traces of length at most $n$. Let $M: L \to [0, 1]^{S \times S}$ specify the transitions, so that $\sum_{a \in L}M(a)$ is a stochastic matrix, $M(a)(s, t) = \tau(s)(t)$ if $\ell(s) = a$ and $M(a)(s, t) = 0$ otherwise. We extend $M$ to the mapping $M : L^{*} \to [0, 1]^{S \times S}$ with $M(w) = M(a_1)\cdots M(a_n)$ for a trace $w = a_1 \cdots a_n$. If the LMC is in state $s$, then with probability $M(w)(s, s')$ it emits a trace $w$ and moves to state $s'$ in $|w|$ steps. 
%Fix an LMC $\mathcal{M} = <S,L,\tau,\ell>$. 
For a trace $w \in L^{*}$, we define $\Run(w) := \{w\}L^{\omega}$; i.e., $\Run(w)$ is the set of traces starting with $w$. To an initial distribution $\pi$ on $S$, we associate the probability space $(L^{\omega}, \mathcal{F}, \mathrm{Pr_{\mathcal{M}, \pi}})$, where $\mathcal{F}$ is the $\sigma$-field generated by all basic cylinders $\Run(w)$ with $w \in L^{*}$ and  $\mathrm{Pr_{\mathcal{M}, \pi}}: \mathcal{F} \to [0,1]$ is the unique probability measure such that $\mathrm{Pr_{\mathcal{M}, \pi}}(\Run(w)) = |\pi M(w)|$. We generalize the definition of $\mathrm{Pr_{\mathcal{M}, \pi}}$ to subdistributions $\pi$ in the obvious way, yielding sub-probability measures. We may drop the subscript $\mathcal{M}$ if it is clear from the context. %An event is a measurable set $E \subseteq L^{\omega}$. 

Given two initial distributions $\mu$ and $\nu$, the \emph{total variation distance} between $\mu$ and $\nu$ is defined as follows:
$$d_{\tv}(\mu, \nu) = \sup_{E \in \mathcal{F}} |\mathrm{Pr_{\mu}}(E) - \mathrm{Pr_{\nu}}(E)|.$$ 

We write $\mu \equiv \nu$ to denote that $\mu$ and $\nu$ are trace equivalent, i.e., $|\mathrm{Pr}_{\mu}(\Run(w))| = |\mathrm{Pr}_{\nu}(\Run(w))|$ holds for all $w \in L^{*}$. We have that trace equivalence and the total variation distance being zero are equivalent \cite[Proposition~3(a)]{CK2014}. 

The pseudometric \emph{probabilistic bisimilarity distance} of Desharnais et al. \cite{DGJP1999} , which we denote by $d_{\pb}$, is a function from $S \times S$ to $[0, 1]$, that is, an element of $[0, 1]^{S \times S}$. It can be defined as the least fixed point of the following function:
\[
\Delta(d)(s, t) = \left \{
\begin{array}{ll}
1 & \mbox{if $\ell(s) \not= \ell(t)$}\\
\displaystyle \min_{\omega \in \Omega(\tau(s), \tau(t))} \sum_{u, v \in S} \omega(u, v) \; d(u, v) & \mbox{otherwise}
\end{array}
\right .
\]
where the set $\Omega(\mu, \nu)$ of \emph{couplings} of $\mu,\nu \in \Dist(S)$ is defined as $\Omega(\mu, \nu)=\left \{ \, \omega \in \Dist(S \times S) \suchthat \sum_{t \in S} \omega(s, t) = \mu(s) \wedge \sum_{s \in S} \omega(s, t) = \nu(t) \, \right \}$. Note that a coupling $\omega \in \Omega$ is a joint probability distribution with marginals $\mu$ and $\nu$ (see, e.g., \cite[page 260-262]{billingsley1995}).

An equivalence relation $R \subseteq S \times S$ is a \emph{probabilistic bisimulation} if for all $(s, t) \in R$, $\ell(s) = \ell(t)$ and 
$\tau(s)(E) = \tau(t)(E)$ for each $R$-equivalence class $E$. \emph{Probabilistic bisimilarity}, denoted by $\mathord{\sim_{\M}}$ (or $\mathord{\sim}$ when $\M$ is clear), is the largest probabilistic bisimulation. For all $s,t \in S$, $s \sim t$ if and only if $d_{\pb}(s,t) = 0$ \cite[Theorem~1]{DGJP1999}.

A \emph{(labelled) Markov decision process} (MDP) is a tuple $<S, \mathcal{A}, L, \varphi, \ell>$ consisting of a finite set $S$ of states, a  finite set $\mathcal{A}$ of actions, a  finite set $L$ of labels, a partial function $\varphi: S \times \mathcal{A} \pfun \Dist(S)$ denoting the probabilistic transition, and a labelling function $\ell: S \to L$. The set of available actions in a state $s$ is $\mathcal{A}(s) = \{\m \in \mathcal{A} \mid \varphi(s, \m) \text{ is defined}\}$. 
%\SK{Not clear what object $\varphi$ is: a set or a function?}
A \emph{memoryless} strategy for an MDP is a  function $\alpha: S \to \Dist(\mathcal{A})$ that given a state $s$, returns a probability distribution on all the available actions at that state. Such strategies are also known as positional, as they do not depend on the history of past states. A strategy $\alpha$ is \emph{memoryless deterministic} (MD) if for all states $s$ there exists an action $\m \in \A(s)$ such that $\alpha(s)(\m) = 1$; we thus view an MD strategy as a function $\alpha: S \to \mathcal{A}$.

For the remainder of the paper,  we fix an MDP $\mathcal{D} = <S, \mathcal{A}, L, \varphi, \ell>$. Given a memoryless strategy $\alpha$ for $\mathcal{D}$, an LMC $\mathcal{D}(\alpha) = <S, L, \tau,\ell>$ is induced, where $\tau(s)(t) = \sum_{\m \in \A(s)}\alpha(s)(\m) \cdot \varphi(s, \m)(t)$. The matrix $M_{\alpha}$ specifies the transitions of the LMC $\mathcal{D}(\alpha)$ as is defined previously. 

We fix two initial distributions $\mu$ and $\nu$ on $S$ (resp.~two initial states $s$ and $t$) for problems related to total variation distance (resp.~probabilistic bisimilarity distance). 

\section{Trace Inequivalence} \label{section:tvInequivalence}
%\input{sections/tvInequivalence.tex}
%trace inequivalence
In this section we show that one can decide in polynomial time whether there exists a memoryless strategy~$\alpha$ so that $\mu \not\equiv \nu$ in $\mathcal{D}(\alpha)$.
In terms of the notation from the introduction, we show that $\TVneqZERO$ is in {\sf P}. Define the following column-vector spaces.
\begin{align*}
\V_1 &= <M_{\alpha_{1}}(a_1)M_{\alpha_{2}}(a_2)\cdots M_{\alpha_m}(a_m)\one:  \alpha_i \text{ is a memoryless strategy}; a_i \in L>  \text{ and} \\
\V_2 &= <M_{\alpha}(w)\one:  \alpha  \text{ is a memoryless strategy}; w \in L^*> \text{ and} \\
\V_3 &= <M_{\alpha}(w)\one:  \alpha  \text{ is an MD strategy}; w \in L^*>.
\end{align*}
Here and later we use the notation $<\cdot>$ to denote the span of (i.e., the vector space spanned by) a set of vectors.
%%Vector space equivalence
By the definitions, we have that $\mu \equiv \nu$ in all LMCs induced by all memoryless strategies~$\alpha$ if and only if 
$\mu M_{\alpha}(w) \one= \nu M_{\alpha}(w) \one$ holds for all memoryless strategies $\alpha$ and all $w \in L^{*}$. It follows:

\begin{proposition} \label{proposition:TVneqZERO-and-vector-space}
	For all distributions $\mu, \nu$ over $S$ we have: %\SK{Do we need subdistributions?}
	\[
	\exists \text{ a memoryless strategy } \alpha \text{ such that } \mu \not\equiv \nu \text{ in }  \D(\alpha)  \; \Longleftrightarrow\; \mu\vecv \not= \nu \vecv \ \text{ for some } \vecv \in \V_2.
	\]
\end{proposition}

To decide $\TVneqZERO$ and to compute the ``witness'' memoryless strategy such that $\mu \not\equiv \nu$ in the induced LMC, it suffices to compute a basis for $\V_2$; more precisely, a set of $\alpha$ and~$w$ such that the vectors $M_\alpha(w) \one$ span~$\V_2$.
As the set of memoryless strategies is uncountable, this is not straightforward.
From the definitions, we know $\V_3 \subseteq \V_2 \subseteq \V_1$.
We will show $\V_1 \subseteq \V_3$ and thus establish the equality of these three vector spaces. It follows from \cite[Theorem~5.12]{FKS2020} that computing a basis for $\V_1$ is in {\sf P}. It follows that our problem $\TVneqZERO$ is also in {\sf P}, but this does not explicitly give the witnessing memoryless strategy. Since $\V_2 = \V_3$, there must exist an MD strategy that witnesses $\mu \not\equiv \nu$. To find this MD strategy, one can go through all MD strategies (potentially exponentially many).  In the following, by considering the vector spaces while restricting the word length, we show that a witness MD strategy can also be computed in polynomial time.

We define the following column-vector spaces. For each $j \in \nat$,
\begin{align*}
\V_{1}^{j} &= <M_{\alpha_{1}}(a_1)M_{\alpha_{2}}(a_2)\cdots M_{\alpha_k}(a_k)\one:  \alpha_i \text{ is a memoryless strategy}; a_i \in L; k \le j>  \text{ and} \\
\V_{2}^{j} &= <M_{\alpha}(w)\one:  \alpha  \text{ is a memoryless strategy}; w \in L^{\le j} > \text{ and} \\
\V_{3}^{j} &= <M_{\alpha}(w)\one:  \alpha  \text{ is an MD strategy}; w \in L^{\le j} >. 
\end{align*}

%The vector spaces $\V_{1}^{j}$, $\V_{2}^{j}$ and $\V_{3}^{j}$ are similar to $\V_1$, $\V_2$ and $\V_3$ but the word length is restricted to be at most $j$ for some $j \in \nat$. 

Let $\alpha$ be an MD strategy and $\m$ be an action available at state $i$. Recall that an MD strategy can be viewed as a function $\alpha: S \to \mathcal{A}$. We define $\alpha^{i \to \m}$ to be the MD strategy such that  $\alpha^{i \to \m}(i)=  \m$ and $\alpha^{i \to \m}(s) = \alpha(s)$ for all $s \in S\setminus\{i\}$. Let $\vecc_i \in \{0, 1\}^{S}$ be the column bit vector whose only non-zero entry is the $i$th one. For a set $B \subseteq \Reals^S$, we define $<B>$ to be the vector space spanned by $B$.

We call a column vector an \emph{MD vector} if it is of the form $M_\alpha(w) \one$ for an MD strategy $\alpha$ and $w \in L^*$. Let $P$ be a set of MD strategy and word pairs, i.e., $P = \{(\alpha_1, w_1), (\alpha_2, w_2), \cdots, (\alpha_m, w_m) \}$ where $\alpha_i$ is an MD strategy and $w_i \in L^{*}$. We define a function $\B$ transforming such a set $P$ to the set of corresponding MD vectors, i.e., $\B(P) = \{M_{\alpha_1}(w_1)\one, M_{\alpha_2}(w_2)\one, \cdots , M_{\alpha_n}(w_n)\one \}$. 

\begin{restatable}{lemma}{lemmaStrategyComposition}\label{lemma:two-strategy-composition-new-strategy-construction}
	Let $j \in \nat$. For all MD strategies $\alpha_1$ and $\alpha_2$, $a \in L$ and $w \in L^{\le j}$, we have $M_{\alpha_{1}}(a)M_{\alpha_{2}}(w)\one \in <\V_{1}^{j} \cup \B(\{(\alpha, aw)\})>$ where $\alpha$ is the MD strategy defined by 	\[
	\alpha(i) = \left \{
	\begin{array}{ll}
	\alpha_1(i) & \mbox{if %$\alpha_1(i) \not= \alpha_2(i) \land
		$\vecc_{i} \not\in \V_{1}^{j}$}\\
	\alpha_2(i) & \mbox{otherwise}
	\end{array}
	\right .
	\]
\end{restatable}

The next lemma shows that a basis for $\V_{1}^{j}$ for some $j \ls |S|$ consisting only of MD vectors can be computed in polynomial time. 

\begin{restatable}{lemma}{lemmaPolyTime}\label{lemma:V1j-basis-polynomial-time}
	Let $j \in \nat$ with $j \ls |S|$. One can compute in polynomial time a set $P_j = \{(\alpha_0, w_0), \cdots, (\alpha_k, w_k)\}$ in which all $\alpha_i$ are MD strategies and all $w_i$ are in $L^{\le j}$ such that $\B(P_j)$ is a basis of~$\V_{1}^{j}$.
\end{restatable}
\begin{proof}[Proof sketch]
	We prove this lemma by induction on $j$. The base case where $j = 0$ is vacuously true with $P_0 = \{(\alpha_0, w_0)\}$ where $\alpha_0$ is an arbitrary MD strategy, $w_0 = \varepsilon$ and $\B(P_0) =\{ \one \}$. For the induction step, assume that we can compute in polynomial time a set $P_j  = \{(\alpha_0, w_0), \cdots, (\alpha_k, w_k)\}$ where all the strategies are MD strategies and all the words are in $L^{\le j}$ such that $\B(P_j)$ is a basis for~$\V_{1}^{j}$. We show that the statement holds for $j+1$.  Define
	\[
	\Sigma = \{ \alpha_0 \} \cup \{ \alpha_0^{s \to \m}: s \in S, \ \m \in \A(s) \} \qquad \text{and} \qquad 	\mathbb{M} = \{ M_{\alpha}(a) \in \Reals^{S \times S}: \alpha \in \Sigma,\ a \in L\}. \]
	%\SK{Unify whether we use $\mid$ or $:$ in a set comprehension. For $\V_1$ etc.\ we used $:$.}
	Next, we present Algorithm~\ref{alg:polynomial-basis} which computes a set $P_{j+1}$ in polynomial time such that 
	\begin{equation}\label{eqn:M-b-product-in-B-j+1}
	\text{for all } M \in \mathbb{M} \text{ and all } \vecb \in \B(P_j): M \cdot \vecb \in <\B(P_{j+1})>
	\end{equation}
	
	\begin{algorithm}[H]
		\DontPrintSemicolon
		$P_{j+1} := P_{j}$\; \label{alg:polynomial-basis-l1}
		\ForEach{$\alpha_1 \in \Sigma,\ a\in L \text{ and } (\alpha_2,w) \in P_j$}{
			\If{$M_{\alpha_1}(a)M_{\alpha_2}(w)\one \not\in <\B(P_{j+1})>$ \label{alg:polynomial-basis-if-condition}} 
			{
				add $(\alpha, aw)$ to $P_{j+1}$ where $\alpha$ is the MD strategy defined as \label{alg:polynomial-basis-line-add-pair}
				\[
				\alpha(i) = \left \{
				\begin{array}{ll}
				\alpha_1(i) & \mbox{if $\vecc_{i} \not\in \V_{1}^{j}$}\\
				\alpha_2(i) & \mbox{otherwise.}
				\end{array}
				\right .
				\]
			}
		}
		\caption{Polynomial-time algorithm computing $P_{j+1}$.}
		\label{alg:polynomial-basis}
	\end{algorithm}
	
	All the vectors in $\B(P_{j+1})$ are linearly independent, as we only add a pair if the corresponding vector is linearly independent to the existing vectors in $\B(P_{j+1})$ (lines~\ref{alg:polynomial-basis-if-condition}-\ref{alg:polynomial-basis-line-add-pair}). Since $\B(P_j)$ is a basis for $\V_{1}^{j}$, we can decide whether $\vecc_i \in \V_{1}^{j}$ for $i \in S$ in polynomial time, and thus compute a pair $(\alpha, aw)$ on line~\ref{alg:polynomial-basis-line-add-pair} in polynomial time. Since $|\Sigma|$ and $|L|$ are polynomial in the size of the MDP, $|P_j| \ls |S|$,  the number of iterations is polynomial in the size of the MDP. The construction of $P_{j+1}$ is then in polynomial time. It remains to show that after adding $(\alpha, aw)$ to $P_{j+1}$ (line~\ref{alg:polynomial-basis-line-add-pair}), we have $M \cdot \vecb = M_{\alpha_1}(a) M_{\alpha_2}(w) \one \in <\B(P_{j+1})>$ . Since the pair $(\alpha_2, w)$ is in $P_j$, we have $w \in L^{\le j}$. Then,
	\begin{eqnarray*}
		&&M \cdot \vecb \\
		&=&M_{\alpha_1}(a) M_{\alpha_2}(w) \one \\
		&\in& <\V_{1}^{j} \cup \B(\{(\alpha, aw)\})> \commenteq{Lemma~\ref{lemma:two-strategy-composition-new-strategy-construction}}\\
		&=& <\B(P_j) \cup \B(\{(\alpha, aw)\})> \commenteq{$\B(P_j)$ is a basis for $\V_{1}^{j}$ by induction hypothesis}\\
		&=& <\B\big(P_j \cup \{(\alpha, aw)\}\big)> 
	\end{eqnarray*}
	Since $P_j \subseteq P_{j+1}$ (line~\ref{alg:polynomial-basis-l1}), we have $\B(P_j) \subseteq \B(P_{j+1})$. By adding the pair $(\alpha, aw)$ to $P_{j+1}$, we have $<\B\big(P_j \cup \{(\alpha, aw)\}\big)> \subseteq <\B(P_{j+1})>$, and thus $M \cdot \vecb \in <\B(P_{j+1})>$.
	
	Finally, we show that the set $P_{j+1}$ satisfies $\V_{1}^{j+1} = <\B(P_{j+1})>$. We have
	\begin{align*}
	<\B(P_{j+1})>
	&\subseteq \V_{3}^{j+1} && \text{for all $(\alpha, w) \in P_{j+1}: \alpha$ is an MD strategy and $w \in L^{\le j+1}$}\\
	&\subseteq \V_{1}^{j+1} && \text{from the definitions}
	\end{align*}
	
	We prove the other direction $\V_{1}^{j+1} \subseteq <\B(P_{j+1})>$ in \cref{appendix:tvInequivalence}. 
\end{proof}

Combining  classical linear algebra arguments about equivalence checking (see, e.g., \cite{Tzeng}) with Lemma~\ref{lemma:V1j-basis-polynomial-time}, we obtain:
\begin{restatable}{lemma}{lemmaEqualityVSpaces}\label{lemma:equality-of-v-spaces}\ 
	\begin{enumerate}
		\item For all $j \ls |S|$ we have $\V_{1}^{j} = \V_{2}^{j} = \V_{3}^{j}$.
		\item We have $\V_1 = \V_2 = \V_3 = \V_{1}^{|S|-1} = \V_{2}^{|S|-1} = \V_{3}^{|S|-1}$.
	\end{enumerate}
\end{restatable}

Thus we obtain:
\begin{proposition}\label{proposition:V1-basis-polynomial-time}
	One can compute in polynomial time a set $P = \{(\alpha_0, w_0), \cdots, (\alpha_k, w_k)\}$ of MD strategy and word pairs such that $\B(P)$ is a basis of $\V_{2}$.  
\end{proposition}
\begin{proof}
	By Lemma~\ref{lemma:equality-of-v-spaces} it suffices to invoke Lemma~\ref{lemma:V1j-basis-polynomial-time} for $j = |S|-1$.
\end{proof}

Now we can prove the main theorem of this section.
\begin{theorem}\label{theorem:trace-inequivalence-polynomial}\
	The problem $\TVneqZERO$ is in {\sf P}.
	Further, for any positive instance of the problem $\TVneqZERO$, we can compute in polynomial time an MD strategy~$\alpha$ and a word~$w$ that witness $\mu \not\equiv \nu$, i.e., $\Pr_{\mu, \D(\alpha)} (\Run(w)) \not= \Pr_{\nu, \D(\alpha)} (\Run(w))$.
\end{theorem}
%\SK{Perhaps rewrite the theorem/proof by invoking \cref{proposition:TVneqZERO-and-vector-space}, which is not used currently.}
\begin{proof}
	A polynomial algorithm follows naturally from \cref{proposition:V1-basis-polynomial-time} and  \cref{proposition:TVneqZERO-and-vector-space}.  We first compute a set $P$ of MD strategy and word pairs such that $\B(P)$ is a basis for $\V_2$. For each $\vecb \in \B(P)$, we check whether $\mu \vecb \not= \nu\vecb$ and output ``yes'' indicating a positive instance if the inequality holds. Otherwise, we have $\mu \vecb = \nu\vecb$ for all $\vecb \in \B(P)$, and the algorithm outputs ``no'' indicating that $\mu \equiv \nu$ holds for all memoryless strategies.
	
	If the instance is positive, there exists a vector $\vecb \in \B(P)$ such that $\mu\vecb \not= \nu\vecb$. Since $\vecb$ is an MD vector which corresponds to a pair $(\alpha, w) \in P$, we have $\mu M_{\alpha}(w)\one \not= \nu M_{\alpha}(w)\one$, equivalently $\Pr_{\mu, \D(\alpha)} (\Run(w)) \not= \Pr_{\nu, \D(\alpha)} (\Run(w))$. 
\end{proof} 

\section{Probabilistic Bisimulation Inequivalence} \label{section:pbInequivalence}
In this section we show that one can decide in polynomial time whether there exists a memoryless strategy~$\alpha$ so that $s \not\sim t$ in $\mathcal{D}(\alpha)$, i.e., we show that $\PBneqZERO$ is in {\sf P}.

\begin{figure}[t]
	\centering
	\tikzstyle{BoxStyle} = [draw, circle, fill=black, scale=0.4,minimum width = 1pt, minimum height = 1pt]
	
	\begin{tikzpicture}[xscale=.6,>=latex',shorten >=1pt,node distance=3cm,on grid,auto]
	
	%\node[label]  at (2,2.7) {the MC~$\C$};
	
	\node[state] (s) at (-1,0) {$s$};
	\node[BoxStyle](bs) at (0.5, 0){};

	\node[state] (t) at (-1, -2.8) {$t$};
	\node[BoxStyle] (bt) at (0.5,-2.8){};
	
	\node[state] (sa) at (4,0.7) {$s_a$};
	\node[BoxStyle] (bsa) at (5.5, 0.08){};
	\node[state] (sb) at (4,-0.7) {$s_b$};
	\node[BoxStyle] (bsb) at (5.5, -0.7){};
	
	\node[state] (ta) at (4,-2.1) {$t_a$};
	\node[BoxStyle] (bta) at (5.5, -2.1){};
	\node[state] (tb) at (4,-3.5) {$t_b$};
	\node[BoxStyle] (btb) at (5.5, -3.5){};
	
	\node[state] (q1) at (9,0.7) {$q_1$};
	\node[BoxStyle]  at (10.5, 0.7){};
	\node[state] (q2) at (9,-1.4) {$q_2$};
	\node[BoxStyle] at (10.5, -0.6){};
	\node at (10.8, -0.9){$\m_1$};
	\node[BoxStyle] at (10.5, -2.2){};
	\node at (10.8, -1.9){$\m_2$};
	
	\node[state] (q3) at (9,-3.5) {$q_3$};
	\node[BoxStyle]  at (10.5, -3.5){};
	
	\node[state] (u) at (13,0.7) {$u$};
	\node[BoxStyle]  at (14.25, 0.7){};
	\node[state, fill=green] (v) at (13,-3.5) {$v$};
	\node[BoxStyle]  at (14.25, -3.5){};
	
	\path[-] (s) edge node [midway, right] {} (bs);
	\path[->] (bs) edge node [near start, above] {} (sa);%$\frac{1}{2}$
	\path[->] (bs) edge node [near start, below] {} (sb);%$\frac{1}{2}$
	
	\path[-] (t) edge node [midway, right] {} (bt);
	\path[->] (bt) edge node [near start, above] {} (ta);%$\frac{1}{2}$
	\path[->] (bt) edge node [near start, below] {} (tb);%$\frac{1}{2}$
	
	\path[->] (sa) edge node [near start, right] {} (q2); %$1$
	\path[-] (sb) edge node [midway, right] {} (bsb);
	\path[-] (ta) edge node [midway, right] {} (bta);
	\path[-] (tb) edge node [midway, right] {} (btb);
	
	\path[->] (bsb) edge node [near start, above] {} (q1); %$\frac{1}{2}$
	\path[->] (bsb) edge node [near start, below] {} (q3); %$\frac{1}{2}$
	
	\path[->] (bta) edge node [near start, above] {} (q1); %$\frac{1}{2}$
	\path[->] (bta) edge node [near start, below] {} (q2); %$\frac{1}{2}$
	
	\path[->] (btb) edge node [near start, above] {} (q2); %$\frac{1}{2}$
	\path[->] (btb) edge node [near start, below] {} (q3); %$\frac{1}{2}$
	
	\path[->] (q1) edge node [near start, above] {} (u); %$\frac{1}{2}$
	\path[->] (q2) edge node [midway, right] {} (u);
	\path[->] (q2) edge node [midway, right] {} (v);
	\path[->] (q3) edge node [near start, above] {} (v); %$\frac{1}{2}$
	\path[->] (u) edge [out=15,in=-15,looseness=8] node [midway, right] {} (u);
	\path[->] (v) edge [out=15,in=-15,looseness=8] node [midway, right] {} (v);
	%\path[->] (sa) edge [out=180,in=30,looseness=1.2] node [midway, above] {$1$} (s);
	%\path[->] (sb) edge [out=-180,in=-30,looseness=1.2] node [midway, below] {$1$} (s);
	%\path[->] (ta) edge  [out=0,in=150,looseness=1.2] node [midway, above] {$1$} (t);
	%\path[->] (tb) edge[out=0,in=-150,looseness=1.2] node [midway, below] {$1$} (t);
	
	\end{tikzpicture}
	%	}
	\caption{In this MDP no MD strategy witnesses $s \not\sim t$. All states have the same label except state $v$.
		By default the transition probabilities out of each action are uniformly distributed.
	}
	%We have $\varphi(s_b, m)(q_1) = \varphi(s_b, m)(q_2) = \frac{1}{2}$ where $m \in \A(s_b)$.\footnotemark
	\label{fig:noMDstrategyPBneqZero} 
\end{figure}
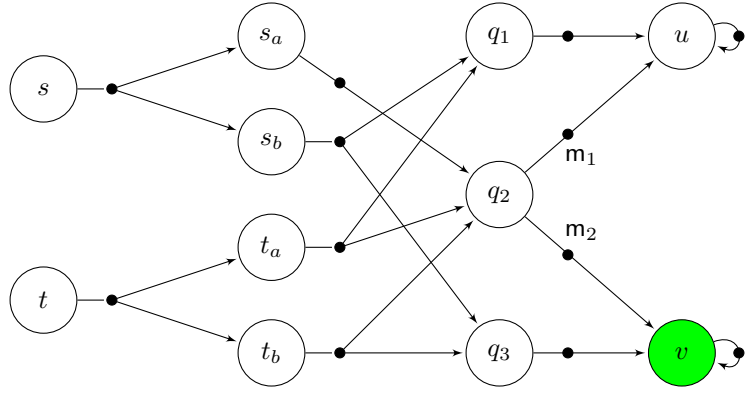

For some MDPs, there might be memoryless strategies such that $s \not\sim t$ in the induced LMC but no such strategy is MD. The MDP in \cref{fig:noMDstrategyPBneqZero} is such an example. Similar to the $or$-gate construction of \cite[Theorem~2]{CBW2012}, we have $s \sim t$ if and only if $q_1 \sim q_2$ or $q_2 \sim q_3$. We have $q_2 \sim q_1$ if the MD strategy maps $q_2$ to the action that goes to state $u$, otherwise $q_2 \sim q_3$ if the MD strategy maps $q_2$ to the action that goes to state $v$. This rules out the algorithm that goes through all the MD strategies.

%\footnotetext{The transition probability out of each action is a uniform distribution over its successors unless specified.}

We define an equivalence relation and run the classical polynomial-time partition refinement as shown in Algorithm~\ref{alg:polynomial-optimistic-partition-refinement}, with an equivalence relation $\mathord{\equiv_X}$ defined below. At the beginning, all states are in the same equivalence class. In a refinement step, a pair of states is split if there \emph{could} exist a memoryless strategy that makes them not probabilistic bisimilar.
Two states $s,t$ remain in the same equivalence class until the end if and only if they are probabilistic bisimilar under all memoryless strategies.

%\noindent \begin{minipage}{0.45\linewidth}
\begin{algorithm}[t]
	%		\setstretch{1.12}
	\DontPrintSemicolon
	$i = 0; X_0 := \{S\}$\; 
	\Repeat{$X_i = X_{i-1}$}{
		$i := i+1$\;
		$X_{i} := S/\mathord{\equiv_{X_{i-1}}}$ \;
	}
	\caption{\mbox{Partition Refinement}}
	\label{alg:polynomial-optimistic-partition-refinement}
\end{algorithm}
%\end{minipage}
%~
%\begin{minipage}{0.4\linewidth}
\begin{table}[b]
%	\centering
	\begin{tabular}{l}
		%\hlinewd{1pt}
		%\textbf{Example}: Algorithm~\ref{alg:polynomial-optimistic-partition-refinement} on the MDP in \cref{fig:noMDstrategyPBneqZero}\\
		\hlinewd{0.5pt}
		$X_0 = \{S\}$ \\
		$X_1 = \big\{  \{v\} , S\setminus\{v\} \big\}$\\
		$ X_2= \big\{ \{v\}, \{q_2\} , \{q_3\},  S\setminus\{v, q_2, q_3\} \big\}$\\
		$X_3= \big\{ \{v\}, \{q_2\} , \{q_3\}, \{s_a\} , \{s_b\}, \{t_a\} , \{t_b\},\{s, t, q_1, u\} \big\}$\\
		$X_4= \big\{ \{v\}, \{q_2\} , \{q_3\}, \{s_a\} , \{s_b\}, \{t_a\} , \{t_b\}, \{s\}, \{t\} , \{q_1, u\}$ \big\}\\
		\hlinewd{0.5pt}
	\end{tabular}
	\caption{Example of running Algorithm~\ref{alg:polynomial-optimistic-partition-refinement} on the MDP in \cref{fig:noMDstrategyPBneqZero}. } \label{tab:example-partition-refinement}
\end{table}
%\end{minipage}

The correctness of this approach is not obvious, as some splits that occurred in different iterations of the algorithm may have been due to different, potentially contradictory, memoryless strategies.
Furthermore, the algorithm does not compute a memoryless strategy that witnesses $s \not\sim t$.
%We will show, however, how to compute such a witness strategy in polynomial time.
The key to solving both problems will be \cref{lemma:partial-strategy-construction}.
%One direction of the proof is straightforward (\cref{lemma:partition-gets-finer} and \cref{lemma:X-is-a-bisimulation}), while the other direction is more involved (the rest of the section).  

A partition of the states $S$ is a set $X$ consisting of pairwise disjoint subsets $E$ of $S$ with $\bigcup_{E \in X} = S$. Recall that $\varphi(s, \m)(s')$ is the transition probability from $s$ to $s'$ when choosing action $\m$. Similarly, $\varphi(s, \m)(E)$ is the transition probability from $s$ to $E \subseteq S$ when choosing action $\m$. We write $\varphi(s, \m)(X)$ to denote the vector (probability distribution) $(\varphi(s, \m)(E))_{E \in X}$. We define $\varphi(s)(X) = \{\varphi(s, \m)(X): \m \in \A(s)\}$, which is a set of probabilistic distributions over the partition $X$ when choosing all available actions of $s$. Each partition is associated with an equivalence relation $\mathord{\equiv_{X}}$ on $S$: $s \equiv_{X} s'$ if and only if 
\begin{itemize}
	\item[-]$\ell(s) = \ell(s')$;
	\item[-] $s \not= s' \implies |\varphi(s)(X)| = 1$ and $\varphi(s)(X)= \varphi(s')(X)$.
\end{itemize} 

%We present below a method that works with refinement steps of partitions on the states. 
Let $S / \mathord{\equiv_{X}}$ denote the set of equivalence classes with respect to $\mathord{\equiv_{X}}$, which forms a partition of $S$. We present in \cref{tab:example-partition-refinement} the partitions of running the algorithm on the MDP in \cref{fig:noMDstrategyPBneqZero}. Notice that states $s$ and $t$ are no longer in the same equivalence class at the end. %$X_1=\{S\setminus\{v\}, \{v\}\}$.
The following lemma is standard, and claims that the partition gets finer.
\begin{restatable}{lemma}{lemmaPartitionInvariant}\label{lemma:partition-gets-finer}
	For all $i \in \nat$, we have $\mathord{\equiv_{X_{i+1}}} \subseteq \mathord{\equiv_{X_{i}}}$.
	%If $s \not\equiv_{X_{i}} t$ then $s\not\equiv_{X_{j}} t$ for all $j \ge i$.
\end{restatable}

%Let $n$ be the number of executions of the loop in Algorithm~\ref{alg:polynomial-optimistic-partition-refinement} such that the algorithm terminates with $X_{n} = X_{n-1}$. 
If the loop in Algorithm~\ref{alg:polynomial-optimistic-partition-refinement} is performed $|S|-1$ times then $X_{|S|-1}$ consists of $|S|$ one-element sets. Hence at most after $|S|-1$ refinement steps the partition $X_i$ cannot be refined.
We aim at proving that $s \equiv_{X_{|S|-1}} t$ if and only if $s \sim_{\D(\alpha)} t$ for all memoryless strategies $\alpha$. In the following lemma we show the forward direction:% show that if $s \equiv_{X_{|S|-1}} t$ then $s \sim_{\D(\alpha)} t$ for all memoryless strategies~$\alpha$. 

\begin{restatable}{lemma}{lemmaXBisimulation} \label{lemma:X-is-a-bisimulation}
	Let $X$ be a partition and $X = S/\mathord{\equiv_{X}}$. We have $\mathord{\equiv_{X}} \subseteq \mathord{\sim_{\D(\alpha)}}$ for all memoryless strategies $\alpha$. 
\end{restatable}

For the converse, to guarantee $\equiv_{X_{|S|-1}} $ is not too fine, it suffices to show that there exists a memoryless strategy $\alpha'$ such that $\mathord{\sim_{\D(\alpha')}} \subseteq \mathord{\equiv_{X}}$ where $X = S/\mathord{\equiv_{X}}$. To do that, we define the equivalence relations $\mathord{\sim_{\D(\alpha)}^{i}}$ with $0 \le i \le |S|$ for all memoryless strategies $\alpha$. 

Let $\alpha$ be a memoryless strategy. Let $\tau$ be the transition function for the LMC $\D(\alpha)$. Define the equivalence relation $\mathord{\sim_{\D(\alpha)}^{i}}$ with $0 \le i \le |S|$ on $S$: $s \sim_{\D(\alpha)}^{i} s'$ if and only if 
\begin{itemize}
	\item[-]$\ell(s) = \ell(s')$;
	\item[-] $i \gr 0 \implies \tau(s)(E)= \tau(s')(E)$ for all $E \in S / \mathord{\sim_{\D(\alpha)}^{i-1}}$.
\end{itemize}

Note that for the LMC $\D(\alpha)$, we have $\mathord{\sim_{\D(\alpha)}^{i+1}} \subseteq \mathord{\sim_{\D(\alpha)}^{i}}$ for all $i \in \nat$ and $\mathord{\sim_{\D(\alpha)}^{|S|-1}}$ is the probabilistic bisimilarity for the LMC $\D(\alpha)$ (see, e.g., \cite{Baier1996}).

Since the witness strategy might not be MD, we compute a set of prime numbers that can be used to form the weights of the actions. The prime numbers are used to rule out certain ``accidental'' bisimulations. We denote by $\size(\D)$ the size of the representation of an object $\D$. We represent
rational numbers as quotients of integers written in binary.
%For example, the size of
%a rational number is the sum of the bit lengths of its numerator and denominator. 

%We want to show that if $s \sim_{\D(\alpha)} t$ for all memoryless strategy $\alpha$ then $s \equiv_{X_n} t$. We consider the contrapositive of the statement, i.e., if $s \not\equiv_{X_n} t$ then there exists a memoryless strategy $\alpha$ such that $s \not\sim_{\D(\alpha)} t$. We show that such a strategy $\alpha$ can be computed in time polynomial in $\size(\D)$.

For $u \in S$, $\m \in \A(u)$ and $E \subseteq S$, we express $\varphi(u, \m)(E) $ as an irreducible fraction $\frac{a_{u, \m, E}}{b_{u, \m, E}}$ where $a_{u, \m, E}$ and $b_{u, \m, E}$ are coprime integers. Similarly, for $u \in S$, $\m_1, \m_2 \in \A(u)$ and $E \subseteq S$,  $\varphi(u, \m_1)(E)-\varphi(u, \m_2)(E)$ is expressed as an irreducible fraction $\frac{c_{u,\m_1,\m_2,E}}{d_{u,\m_1,\m_2,E}}$ that $c_{u,\m_1,\m_2,E}$ and $d_{u,\m_1,\m_2,E}$ are coprime integers. Let $N \subseteq \mathbb{N}$ be the following set: % of positive integers:
\begin{eqnarray*}N &= &\{b_{u, \m, E}: u \in S,\ \m \in \A(u) \text{ and } E \in \textstyle\bigcup_{i} X_i\} \cup \\
	&& \{c_{u,\m_1,\m_2,E}: u \in S,\ \m_1,\m_2 \in \A(u),\ E \in \textstyle\bigcup_{i} X_i \text{ and }  c_{u,\m_1,\m_2,E} \gr 0\}.
\end{eqnarray*}
%\change{Take the MDP in \cref{fig:noMDstrategyPBneqZero} as an example,  we have $N = \{1, 2\}$. We have $2 \in N$ since $\varphi(s, m_s)(\{s_a\}) = \frac{1}{2}$ where $m_s$ is the only available action at state $s$.}
%
We denote by $\theta(x)$ the number of different prime factors of a positive integer $x$, and by $\theta(N)$ the number of different prime factors in~$N$ where $N$ is a set of positive integers.

\begin{restatable}{lemma}{lemmaPolyPrimeFactorsN}\label{lemma:prime-factors-in-N-polynomial}
	$\theta(N)$ is polynomial in $\size(\D)$.
\end{restatable}

Using the prime number theorem, we obtain the following lemma which guarantees that one can find $|S|$ extra different prime numbers other than the prime factors in $N$ in time polynomial in $\size(\D)$. 

\begin{restatable}{lemma}{lemmaPolyPrimeFactors}\label{lemma:extra-prime-factors-polynomial}
	One can find $|S|$ different prime numbers in time polynomial in $\size(\D)$ such that any of them is coprime to all numbers in the set $N$. 
\end{restatable}

To each $u \in S$, we assign a different prime number $p_u$ that is coprime with all $b \in N$. This can be done in polynomial time by \cref{lemma:extra-prime-factors-polynomial}. We have 
\begin{equation}\label{eqn:prime-number-coprime-with-N}
p_u \nmid b \text{ for all } b\in N \quad \text{ and } \quad u\not=v \implies p_u \not= p_v \text{ for all } u, v \in S 
\end{equation} 

We define a partial memoryless strategy for $\D$ to be a partial function $\alpha': S \pfun \Dist(\A)$ that, given a state $s \in S$, returns $\alpha'(s) \in \Dist(\A(s))$ if $\alpha'(s)$ is defined. A memoryless strategy $\alpha$ is compatible with a partial memoryless strategy $\alpha'$, written as $\alpha \sqsupseteq \alpha'$, if and only if $\alpha(s) = \alpha'(s)$ for all $s$ such that $\alpha'(s)$ is defined. We construct the partial memoryless strategy iteratively.

\begin{restatable}{lemma}{lemmaPartialStrategyConstruction}\label{lemma:partial-strategy-construction}
	%	There exists a partial memoryless strategy $\alpha'$ such that if $s \not\equiv_{X_n} t$ then $s \not\sim t$ in $\D(\alpha)$ for all $\alpha \sqsupseteq \alpha'$. if $s \not\equiv_{X_i} t$ then $s \not\sim^{i}_{\D(\alpha)} t$ in $\D(\alpha)$ 
	Let $i \in \nat$ with $i \le |S|$. One can compute in polynomial time a partial strategy $\alpha'_i$ such that $\mathord{\sim^{i}_{\D(\alpha)}} \subseteq \mathord{\equiv_{X_i}}$ for all $\alpha \sqsupseteq \alpha'_i$.
\end{restatable}
\begin{proof}[Proof sketch]
	We prove the statement by induction on $i$. Let $s, t \in S$. The base case is $i=0$. By definition, we have if $s \not\equiv_{X_0} t$ then $\ell(s) \not= \ell(t)$. We also have if $\ell(s) \not= \ell(t)$, then $s \not\sim^{0}_{\D(\alpha)} t$ in $\D(\alpha)$ for all memoryless strategy $\alpha$. We simply let $\alpha_0'$ be the empty partial function such that $\alpha \sqsupseteq \alpha'_0$  holds for any memoryless strategy $\alpha$.
	
	For the induction step, assume that we can compute in polynomial time a partial strategy $\alpha_i'$ such that $\mathord{\sim^{i}_{\D(\alpha)}} \subseteq \mathord{\equiv_{X_i}}$ for all $\alpha \sqsupseteq \alpha'_i$, i.e., if $s \not\equiv_{X_i} t$ then $s \not\sim^{i}_{\D(\alpha)} t$ in $\D(\alpha)$. We show the statement holds for $i+1$. 
	
	\begin{algorithm}[H]
		\DontPrintSemicolon
		$\alpha_{i+1}' := \alpha_{i}'$\; \label{alg:polynomial-partial-strategy-basis}
		%\QT{$D:=\{u \in S : \alpha'_i(u) \text{ is defined}\}$}\;  \QT{$\setminus D$} 
		\ForEach{$u \in S$ such that $|\varphi(u)(X_i)| = 1$ and $|\varphi(u)(X_{i+1})| \not= 1$ \label{alg:partial-strategy-for-condition}}{
			pick $\m_1, \m_2 \in \A(u)$ such that for a set $E \in X_{i+1}: \varphi(u, \m_1)(E) \gr \varphi(u, \m_2)(E)$ \;
			$\alpha_{i+1}'(u)(\m_1): = \frac{1}{p_u}$\;
			$\alpha_{i+1}'(u)(\m_2): = 1- \frac{1}{p_u}$\;
			\label{alg:polynomial-partial-strategy-prime-number-selection}
		}
		\caption{Polynomial-time algorithm constructing $\alpha_{i+1}'$.}
		\label{alg:polynomial-construction-partial-strategy}
	\end{algorithm}

	Algorithm~\ref{alg:polynomial-construction-partial-strategy} computes the partial memoryless strategy $\alpha_{i+1}'$ in polynomial time. We show that $\alpha'_{j}$ does not overwrite $\alpha'_k$ for all $k \ls j$. It follows that for any $\alpha \sqsupseteq \alpha_{i+1}'$, it satisfies $\alpha \sqsupseteq \alpha_{i}'$. Let $\alpha \sqsupseteq \alpha_{i+1}'$. Assume $s \not\equiv_{X_{i+1}} t$. We distinguish the two cases: $s \not\sim^{i}_{\D(\alpha)} t$ and $s \sim^{i}_{\D(\alpha)} t$. For both cases we can derive $s \not\sim^{i+1}_{\D(\alpha)} t$, i.e., $\mathord{\sim^{i+1}_{\D(\alpha)}} \subseteq \mathord{\equiv_{X_{i+1}}}$ as desired. The details can be found in \cref{appendix:pbInequivalence}.
\end{proof}
% Let $\alpha_0'$ be the empty partial function.
For example, let $p_{q_2}$, the prime number assigned to state $q_2$ in \cref{fig:noMDstrategyPBneqZero}, be $3$ which is coprime with numbers in $N=\{1,2\}$.\footnote{We have $2 \in N$ since $\varphi(s, \m_s)(\{s_a\}) = \frac{1}{2}$ where $\m_s$ is the only available action at state $s$.} We show how the partial strategy $\alpha_1'$ is constructed. On line~1 of Algorithm~\ref{alg:polynomial-construction-partial-strategy}, $\alpha_1'$ is equal to $\alpha_0'$, the empty partial function. Since $|\varphi(q_2)(X_0)| = 1$ and $|\varphi(q_2)(X_{1})| = 2$, we enter the for loop. We can pick $\m_1, \m_2 \in \A(q_2)$ and $E=S\setminus\{v\} \in X_1$ on line~3, since $\varphi(q_2,\m_1)(E)= 1 \gr 0 = \varphi(q_2,\m_2)(E)$. We then define the strategy for $q_2$ (line~4): $\alpha_{1}'(q_2)(\m_1)= \frac{1}{3} $ and $\alpha_{1}'(q_2)(\m_2) = \frac{2}{3}$. We have completed the construction of $\alpha_1'$ as $|\varphi(u)(X_0)| =|\varphi(u)(X_1)| = 1$ for all other state $u$.

\begin{theorem}\label{theorem:polynomial-time-compute-strategy-PBneqZero}
	One can compute in polynomial time a memoryless strategy $\beta$ such that $\mathord{\sim_{\D(\beta)}} \subseteq \mathord{\sim_{\D(\alpha)}}$ for all memoryless strategies~$\alpha$. 
\end{theorem}
\begin{proof}	
	By invoking \cref{lemma:partial-strategy-construction} for $i = |S|-1$, a partial strategy $\alpha'_{|S|-1}$ can be computed in polynomial time such that $\mathord{\sim^{|S| - 1}_{\D(\alpha)}} \subseteq \mathord{\equiv_{X_{|S| - 1}}}$ for all $\alpha \sqsupseteq \alpha'_{|S| - 1}$. Since $\mathord{\sim^{|S| - 1}_{\D(\alpha)}} = \mathord{\sim_{\D(\alpha)}} $, we have $\mathord{\sim_{\D(\alpha)}} \subseteq \mathord{\equiv_{X_{|S| - 1}}}$ for all $\alpha \sqsupseteq \alpha'_{|S| - 1}$. Let $\beta$ be a memoryless strategy defined by
	\[
	\beta(u) = \left \{
	\begin{array}{ll}
	\alpha'_{|S| - 1}(u) & \mbox{if $\alpha'_{|S| - 1}(u)$ is defined}\\
	\delta_{\m_u} \text{ where } \m_u \in \A(u)& \mbox{otherwise}
	\end{array}
	\right .
	\]
	By definition the memoryless strategy $\beta$ is compatible with $\alpha'_{|S| - 1}$. We have:
	\begin{align*}
	\mathord{\sim_{\D(\beta)}} 
	&\subseteq \mathord{\equiv_{X_{|S| - 1}}} && \beta \sqsupseteq \alpha'_{|S| - 1}\\
	&\subseteq \mathord{\sim_{\D(\alpha)}} \text{ for all strategy } \alpha &&  X_{|S| - 1} = S/\mathord{\equiv_{X_{|S| - 1}}} \text{ and \cref{lemma:X-is-a-bisimulation}}\qedhere
	\end{align*}
	
\end{proof}

\begin{corollary}\label{corollary:PBneqZero-in-P}
	The problem $\PBneqZERO$ is in $\sf P$. Further, for any positive instance of the problem $\PBneqZERO$, we can compute in polynomial time a memoryless strategy that witnesses $s \not\sim t$.
\end{corollary} 

\section{The Distance One Problems} \label{section:DistanceOne}
In this section, we summarise the results for the two distance one problems, namely $\TVeqONE$ and $\PBeqONE$.
The \emph{existential theory of the reals, {\sf ETR},} is the set of valid formulas of the form
$$\exists x_1 \dots \exists x_n~R(x_1, \dots, x_n),$$ 
where $R$ is a boolean combination of comparisons of the form
$p(x_1, \dots, x_n) \sim 0$, in which $p(x_1, \dots, x_n)$ is a multivariate polynomial (with rational coefficients) and
$\mathord{\sim} \in \{ \ls, \gr, \mathord{\le}, \mathord{\ge}, \mathord{=}, \mathord{\ne} \}$.
The complexity class $\ETR$~\cite{SchaeferS17} consists of those problems that are many-one reducible to {\sf ETR} in polynomial time.
Since {\sf ETR} is {\sf NP}-hard and in {\sf PSPACE}~\cite{Can88,Renegar92}, we have ${\sf NP} \subseteq \ETR \subseteq {\sf PSPACE}$.
%It is known that $\ETR$ is closed under {\sf NP}-reductions~\cite{CateKO13} which is needed later for showing the membership of $\TVneqONE$ in $\ETR$. %(I THINK WE NEEDED THIS FOR SOMETHING?).

%TV=1
For some MDPs there exist memoryless strategies that make $d_{\tv}(\delta_s,\delta_t) = 1$ but no such strategy is MD. For example, consider the MDP in \cref{fig:noMDstrategyPBeqOne} which has two MD strategies. We have $d_{\tv}(\delta_s,\delta_t) = \frac{1}{2}$ which is less than $1$ in the LMC induced by any of the two MD strategies, and $d_{\tv}(\delta_s,\delta_t) = 1$ in the LMC induced by any other strategy. Thus, we cannot simply guess an MD strategy. We show that the problem $\TVeqONE$ is in~$\ETR$, using the characterization from~\cite[Theorem~21]{CK2014} of total variation distance~$1$ in LMCs and some reasoning on convex polyhedra:
\begin{restatable}{theorem}{theoryTVDistanceOneUB} \label{theorem:tvdistance-eq-one-reduce-to-ETR}
	The problem $\TVeqONE$ is in~$\ETR$.
\end{restatable}

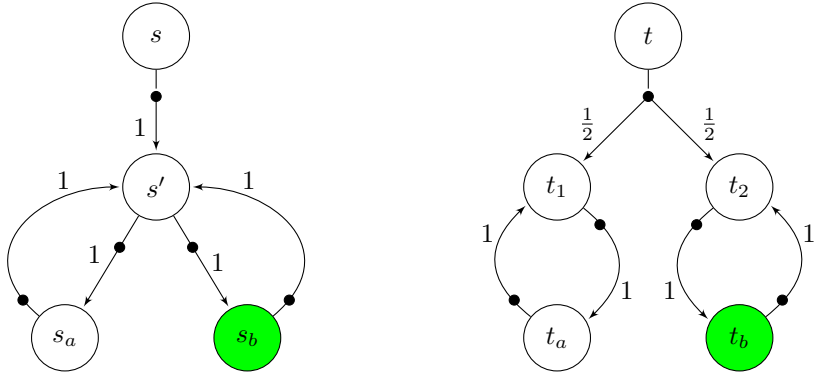
\begin{figure}[t]			
	%	\begin{minipage}[b]{0.53\columnwidth}
	%		\hspace{-2em}
	%		\resizebox{!}{4cm}{
	%\input{figures/exampleNoMDstrategyPBeqOne1.tex}
	\begin{subfigure}{0.45\textwidth}
		\centering
		\tikzstyle{BoxStyle} = [draw, circle, fill=black, scale=0.4,minimum width = 1pt, minimum height = 1pt]
		
		\begin{tikzpicture}[xscale=.6,>=latex',shorten >=1pt,node distance=3cm,on grid,auto]
		
		\node[state] (t) at (2,4) {$s$};
		\node[BoxStyle] (tt) at (2,3.2){};
		
		\node[state] (t1) at (2,2) {$s'$};
		
		\node[BoxStyle] (tat) at (1.2,1.2){};
		\node[BoxStyle] at (2.8,1.2){};
		
		\node[state] (qb) at (0,0) {$s_a$};
		\node[state, fill=green] (qc) at (4,0) {$s_b$};
		
		\node[BoxStyle]at (-.92,0.5){};
		\node[BoxStyle] at (4.92,0.5){};
		
		\path[-] (t) edge node [midway, right] {} (tt);
		\path[->] (tt) edge node [midway, left] {$1$} (t1);
		
		\path[->] (t1) edge node [midway, left, yshift=0.1cm] {$1$} (qb);
		
		\path[->] (t1) edge node [midway, right] {$1$} (qc);
		
		\draw [->,postaction={decorate}] (qb) to [out=150, in=180, looseness=1.8] node [near end, above] {$1$} (t1);
		\draw [->,postaction={decorate}] (qc) to [out=30, in=0, looseness=1.8] node [near end, above] {$1$} (t1);
		
		\end{tikzpicture}
		%\hspace{-2.5em}
		%\input{figures/exampleNoMDstrategyPBeqOne2.tex}
		\hfill
	\end{subfigure}
	\begin{subfigure}{0.45\textwidth}
		\centering
		\tikzstyle{BoxStyle} = [draw, circle, fill=black, scale=0.4,minimum width = 1pt, minimum height = 1pt]
		
		\begin{tikzpicture}[xscale=.6,>=latex',shorten >=1pt,node distance=3cm,on grid,auto]
		
		\node[state] (t) at (2,4) {$t$};
		\node[BoxStyle] (tt) at (2,3.2){};
		
		\node[state] (t1) at (0,2) {$t_1$};
		\node[state] (t2) at (4,2) {$t_2$};
		
		\node[BoxStyle] (tat) at (0.94,1.5){};%(0,1.2){};
		\node[BoxStyle] (tbt)at (3.06,1.5){};%(4,1.2){};
		
		\node[state] (qb) at (0,0) {$t_a$};
		\node[state, fill=green] (qc) at (4,0) {$t_b$};
		
		\node[BoxStyle] at (-.94,0.5){};
		\node[BoxStyle] at (4.94,0.5){};
		
		\path[-] (t) edge node [midway, right] {} (tt);
		\path[->] (tt) edge node [midway, left,  xshift=-0.1cm, yshift=0.1cm] {$\frac{1}{2}$} (t1);
		\path[->] (tt) edge node [midway, right,  xshift=0.1cm, yshift=0.1cm] {$\frac{1}{2}$} (t2);
		\draw [->,postaction={decorate}] (t1) to [out=330, in=30, looseness=1.8] node [near end, right] {$1$} (qb);
		
		\draw [->,postaction={decorate}] (t2) to [out=210, in=150, looseness=1.8] node [near end, left] {$1$} (qc);
		%\path[-] (t1) edge node [midway, left] {} (tat);
		%\path[->] (tat) edge node [midway, left] {$1$} (qb);
		
		%\path[-] (t2) edge node [midway, left] {} (tbt);
		%\path[->] (tbt) edge node [midway, right] {$1$} (qc);
		
		\draw [->,postaction={decorate}] (qb) to [out=150, in=210, looseness=1.8] node [near end, left] {$1$} (t1);
		\draw [->,postaction={decorate}] (qc) to [out=30, in=-30, looseness=1.8] node [near end, right] {$1$} (t2);
		
		\end{tikzpicture}
	\end{subfigure}
	%	}
	\caption{In this MDP, no MD strategy witnesses $d_{\tv}(\delta_s, \delta_t)=1$ (nor $d_{\pb}(s,t) = 1$). States $s_b$ and $t_b$ have label $b$ while all other states have label~$a$.}\label{fig:noMDstrategyPBeqOne} %\SK{Does this also work for $d_{tv}$?}
	%	\end{minipage}
\end{figure}

The problem $\TVeqONE$ is {\sf NP}-hard, and $\PBeqONE$ is {\sf NP}-complete.
The hardness results for both problems are by reductions from the Set Splitting problem.
Given a finite set $S$ and a collection $\C$ of subsets of $S$, Set Splitting asks whether there is a partition of $S$ into disjoint sets $S_1$ and $S_2$ such that no set in $\C$ is a subset of $S_1$ or $S_2$. 

\begin{figure}[t]
	\centering
	
	\tikzstyle{BoxStyle} = [draw, circle, fill=black, scale=0.3,minimum width = 1pt, minimum height = 1pt]
	\begin{tikzpicture}[xscale=.6,>=latex',shorten >=1pt,node distance=3cm,on grid,auto]
	
	%\node[label]  at (2.8,2.7) {the MDP~$\D$};
	%\draw (-1.6,-.6) [dashed] rectangle (7.7,2.7);
	\node[state] (s) at (3,3.5) {$s$};
	\node[BoxStyle] (sdot) at (3,2.8){};
	
	\node[state] (t) at (11,3.5) {$t$};
	\node[BoxStyle] (tdot) at (11,2.8){};
	
	\node[state] (c1) at (1,2) {$C_{1}$};
	%\node[state] (c2) at (2.9,2) {$C_2$};
	\node[state] (cm) at (5,2) {$C_{2}$};
	
	\node[state] (c1p) at (9,2) {$C_{1}'$};
	%\node[state] at (10.9,2) {$C_{2}'$};
	\node[state] (cmp) at (13,2) {$C_{2}'$};
	%label=left:{{$\m_{1}$}}
	
	\node[BoxStyle] at (1.8,1.2){};
	\node[BoxStyle] at (3.3,1.2){};
	\node[BoxStyle] at (5.75,1.2){};
	\node[BoxStyle] at (6.65,1.2){};
	\node[BoxStyle] at (7.35,1.2){};
	\node[BoxStyle] at (8.15,1.2){};
	\node[BoxStyle] at (10.6,1.2){};
	\node[BoxStyle] at (12.2,1.2){};

	\node[state] (s1) at (3,0) {$e_{1}$};
	\node[state] (s2) at (6.9,0) {$e_{2}$};
	% \node[state] at (8.6,0) {$e_{3}$};
	\node[state] (sn) at (11,0) {$e_{3}$};
	
	\node[BoxStyle] at (3,-0.9){};
	\node[BoxStyle] at (5.1,-0.9){};
	\node[BoxStyle] at (6.5,-0.9){};
	\node[BoxStyle] at (7.4,-0.9){};
	\node[BoxStyle] at (8.7,-0.9){};
	\node[BoxStyle] at (11,-0.9){};
	
	\node[state] (u) at (3,-2) {$u$};
	\node[state, fill=green] (v) at (11,-2) {$v$};
	
	\node[BoxStyle] at (1.75,-2){};
	\node[BoxStyle] at (12.25,-2){};
	
	%\path[->] (qc) edge [loop right] node [right, midway] {$c,1$}  (qc);
	%\path[->] (qb) edge [loop left] node [left, midway] {$b,1$}  (qb);
	\path[-] (s) edge node  [near end,left] {} (sdot);
	\path[->] (sdot) edge node  [xshift=-2mm, near start, left] {$\frac{1}{2}$} (c1);
	\path[->] (sdot) edge node  [xshift=2mm, near start,right] {$\frac{1}{2}$} (cm);

	\path[-] (t) edge node  [near end,left] {} (tdot);
	\path[->] (tdot) edge node  [xshift=-2mm, near start, left] {$\frac{1}{2}$} (c1p);
	\path[->] (tdot) edge node  [xshift=2mm, near start, right] {$\frac{1}{2}$} (cmp);
	
	\path[->] (v) edge  [out=15,in=-15,looseness=8] node [right] {$1$} (v);
	\path[->] (u) edge [out=165,in=195,looseness=8] node [left]{$1$} (u);
	
	\path[->, draw=red, thick] (c1) edge node  [xshift= -1mm,very near end,left] {} (s1);
	\path[->] (c1p) edge node  [xshift= 1mm,very near end,right] {} (s1);
	\path[->] (c1) edge node  [xshift= -1mm,very near end,left] {} (s2);
	\path[->, draw=red, thick] (c1p) edge node  [xshift= 1mm,very near end,right] {} (s2);
	\path[->] (cm) edge node [xshift= -1mm,very near end,left] {} (s2);
	\path[->, draw=red, thick] (cmp) edge node [xshift= 1mm,very near end,right] {} (s2);
	\path[->, draw=red, thick] (cm) edge node [xshift= -1mm,very near end,left] {} (sn);
	\path[->] (cmp) edge node [xshift= 1mm,very near end,right] {} (sn);
	
	\path[->, draw=red, thick] (s1) edge node  [very near end,left] {} (u);
	\path[->] (s1) edge node  [xshift=-2mm, very near end,left] {} (v);
	
	\path[->] (s2) edge node  [very near end,left] {} (u);
	\path[->, draw=red, thick] (s2) edge node  [xshift=-2mm, very near end,left] {} (v);
	
	\path[->, draw=red, thick] (sn) edge node  [xshift=2mm, very near end,right] {} (u);
	\path[->] (sn) edge node  [very near end,right] {} (v);
	\end{tikzpicture}
	%		}
	\caption{The MDP in the reduction from Set Splitting for {\sf NP}-hardness of $\TVeqONE$ (or $\PBeqONE$). }
	\label{fig:reductionfromSetSplit}
\end{figure}
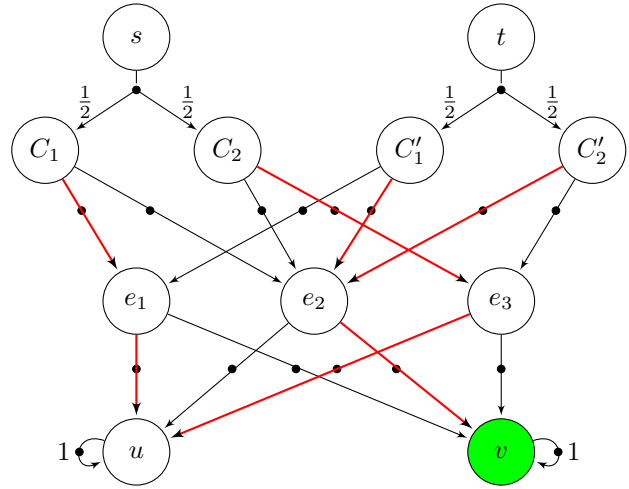

Let $<S, \C>$ be an instance of Set Splitting where $S = \{e_1, \cdots, e_n\}$ and $\C = \{C_1, \cdots, C_m\}$ is a collection of subsets of $S$. We construct an MDP $\D$ consisting of the following states: two states $s$ and $t$, a state $e_i$ for each element in $S$, twin states $C_j$ and $C_j'$ for each element in $\C$, two sink states $u$ and $v$. State $v$ has label $b$ while all other states have label $a$. State $s$ ($t$) has a single action which goes with uniform probability $\frac{1}{m}$ to states $C_i$ ($C_i'$) for $1 \le i \le m$.  For each $e_i \in C_j$, there is an action from state $C_j$ and $C_j'$ leading to state $e_i$ with probability one. Each state $e_i$ has two actions going to the sink states $u$ and $v$ with probability one, respectively. We have: $<S, \C> \in {\mbox{Set Splitting}} \iff \exists \, \mbox{memoryless strategy $\alpha$ such that} \; d_{\tv}(\delta_s, \delta_t) = 1 \, \mbox{in} \; \D(\alpha).$ 

For example, let $S = \{e_1, e_2, e_3\}$ and $\C = \{C_1, C_2\}$ with $C_1 = \{e_1, e_2\}$ and $C_2 = \{e_2, e_3\}$. \figurename~\ref{fig:reductionfromSetSplit} shows the corresponding MDP. The MD strategy highlighted, corresponding to the partition of $S_1 = \{e_1, e_3\}$ and $S_2 = \{e_2\}$, witnesses $d_{\tv}(\delta_s, \delta_t) = 1$.

\begin{restatable}{theorem}{theoryTVDistanceOneLB}\label{theorem:tvdistance-one-NP-hardness}
	The Set Splitting problem is polynomial-time many-one reducible to $\TVeqONE$, hence $\TVeqONE$ is {\sf NP}-hard. 
\end{restatable}

%PB=1

%For some MDPs there exist memoryless strategies that make $d_{\pb}(s,t) = 1$ but no such strategy is MD. For example, consider the MDP in \cref{fig:noMDstrategyPBeqOne} which has two MD strategies. We have $d_{\pb}(s,t) = \frac{1}{2}$ which is less than one in the LMC induced by any of the two MD strategies, and $d_{\pb}(s,t) = 1$ in the LMC induced by any other strategy.

% Similarly, in the MDP of \cref{fig:noMDstrategyPBeqOne}, $d_{\pb}(s,t) = 1$ cannot be witnessed by an MD strategy, which rules out the algorithm The problem $\PBeqONE$ is {\sf NP}-complete. 

The problem $\PBeqONE$ is {\sf NP}-complete. The MDP in \cref{fig:noMDstrategyPBeqOne} is also an example of no MD strategy witnessing  $d_{\pb}(s,t) = 1$, which rules out the algorithm of simply guessing an MD strategy. By \cite{TvB2018}, deciding whether $d_{\pb}(s,t) = 1$ in an LMC can be formulated as a reachability problem on a directed graph induced by the LMC. One can nondeterministically guess the graph induced by the LMC and use Algorithm~\ref{alg:polynomial-construction-partial-strategy}  to construct a memoryless strategy that witnesses $d_{\pb}(s,t) = 1$. % Thus, to show the {\sf NP} upper bound, we cannot simply guess an MD strategy. Similarly to Theorem~\ref{theorem:tvdistance-one-NP-hardness}, the polynomial-time many-one reduction for $\PBeqONE$ is from Set Splitting. %there is a memoryless strategy $\alpha$ that witnesses $d_{pb}(s,t) = 1$ if and only if the memoryless strategy $\alpha'$ for the LMC $\D(\alpha)$ constructed using Algorithm~\ref{alg:polynomial-construction-partial-strategy} also witnesses $d_{pb}(s,t) = 1$. This suggests a nondeterministic algorithm for checking whether $d_{pb}(s,t) = 1$ holds: construct a memoryless strategy $\alpha$ using Algorithm~\ref{alg:polynomial-construction-partial-strategy} for the guessed states which are reachable from $s$ and $t$ and check if $d_{pb}^{\D(\alpha)}(s,t) = 1$. 

\begin{restatable}{theorem}{theoremPBDistanceOneUBLB}
	The problem $\PBeqONE$ is {\sf NP}-complete.
\end{restatable}

\section{Making Distances Small} \label{section:summaryDistanceZeroAndNeqOne}
In this section, we summarise the results for the remaining problems, which are all about making the distance small (equal to~$0$ or less than~$1$).

We show that $\TVeqZERO$ and $\TVneqONE$ are $\ETR$-complete. The proof for the membership of $\TVeqZERO$ in $\ETR$ is similar to \cite[Theorem~4.3]{FKS2020}. For both hardness results we provide reductions from the \emph{Nonnegative Matrix Factorization (NMF)} problem, which asks, given a nonnegative matrix $J \in \mathbb{Q}^{n \times m}$ and a number $r \in \nat$,
whether there exists a factorization $J = A \cdot W$ with nonnegative matrices $A \in \mathbb{R}^{n \times r}$ and $W \in \mathbb{R}^{r \times m}$.
The NMF problem is $\ETR$-complete by \cite[Theorem~2]{Shitov2016}, see also \cite{CohenR93,Vavasis09,AroraGKM12} for more details on the NMF problem.  The reduction is similar to \cite[Theorem~4.5]{FKS2020}.

%TV=0
%For hardness, fix an instance of the NMF problem, a nonnegative matrix $J \in \mathbb{R}^{n \times m}$ and a number $r \in \mathbb{N}$. %Similar to \cite[Theorem~4.5]{FKS2020}, we assume, without loss of generality, that~$J$ is a stochastic matrix. 

%TODO say something about the upperbound (\cite[Proposition~10]{14KW-ICALP})
\begin{restatable}{theorem}{theoremTVDistanceZeroUBLB}
	The problem $\TVeqZERO$ is $\ETR$-complete.
\end{restatable}

%TV<1
\begin{restatable}{theorem}{theoremTVDistanceNeqOneUBLB}
	The problem $\TVneqONE$ is $\ETR$-complete.
\end{restatable}

%PB=0
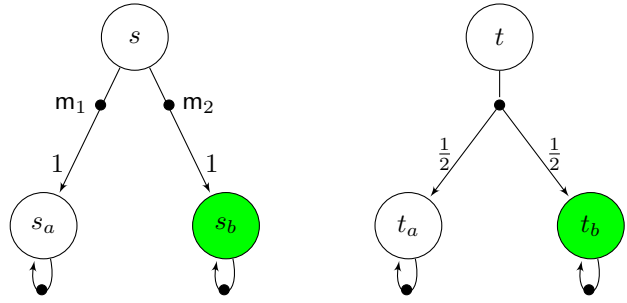
\begin{figure}[t]
%	\begin{subfigure}{0.45\textwidth}
		%\centering
		%\captionsetup{justification=centering}
%		\resizebox{0.72\columnwidth}{!}{
			%\input{figures/exampleNoMDstrategyPBeqZero.tex}
			
			\centering
			\tikzstyle{BoxStyle} = [draw, circle, fill=black, scale=0.4,minimum width = 1pt, minimum height = 1pt]
			
			\begin{tikzpicture}[xscale=.6,>=latex',shorten >=1pt,node distance=3cm,on grid,auto]
			
			%\node[label]  at (2,2.7) {the MC~$\C$};
			
			\node[state] (s) at (0,0) {$s$};
			\node[BoxStyle] at (-0.75, -0.9){};
			\node at (-1.4, -0.9){$\m_1$};
			\node[BoxStyle] at (0.75, -0.9){};
			\node at (1.4, -0.9){$\m_2$};
			
			\node[state] (sa) at (-2, -2.5) {$s_a$};
			\node[BoxStyle] at (-2.04, -3.35){};
			\node[state, fill=green] (sb) at (2,-2.5) {$s_b$};
			\node[BoxStyle] at (1.96, -3.35){};
			
			\node[state] (t) at (8,0) {$t$}; %5
			\node[BoxStyle] (bt) at (8, -0.9){};
			
			\node[state] (ta) at (6,-2.5) {$t_a$};
			\node[BoxStyle] at (5.96,-3.35){};
			\node[state, fill=green] (tb) at (10,-2.5) {$t_b$};
			\node[BoxStyle] at (9.96,-3.35){};
			
			\path[->] (s) edge node [near end, left] {$1$} (sa);
			\path[->] (s) edge node [near end, right] {$1$} (sb);
			
			\path[-] (t) edge node [midway, right] {} (bt);
			\path[->] (bt) edge node [midway, left] {$\frac{1}{2}$} (ta);
			\path[->] (bt) edge node [midway, right] {$\frac{1}{2}$} (tb);
			
			\path[->] (sa) edge [out=-75,in=-105,looseness=4] node [midway, below] {} (sa);
			\path[->] (sb) edge [out=-75,in=-105,looseness=4] node [midway, right] {} (sb);
			\path[->] (ta) edge [out=-75,in=-105,looseness=4] node [midway, left] {} (ta);
			\path[->] (tb) edge [out=-75,in=-105,looseness=4] node [midway, left] {} (tb);
			
			\end{tikzpicture}
			
%		}
%		\vspace{-0.6em}
		\caption{In this MDP, no MD strategy witnesses $d_{\pb}(s,t) = 0$. States $s_b$ and $t_b$ have label $b$ while all other states have label $a$.}%The MDP~$\mathcal{D}$ that no MD strategy witnesses $d_{pb}(s,t) = 0$. 
		\label{fig:noMDstrategyPBeqZero}
\end{figure}
	%\hfill
\begin{figure}[t]
%	\hspace{2em}
%	\begin{subfigure}{0.45\textwidth}
		%\centering
	%	\captionsetup{justification=centering}
%		\resizebox{0.9\columnwidth}{!}{
			%\input{figures/exampleNoMDstrategyPBneqOne.tex}
			\centering
			\tikzstyle{BoxStyle} = [draw, circle, fill=black, scale=0.4,minimum width = 1pt, minimum height = 1pt]
			
			\begin{tikzpicture}[xscale=.6,>=latex',shorten >=1pt,node distance=3cm,on grid,auto]
			
			%\node[label]  at (2,2.7) {the MC~$\C$};
			\node[state] (s) at (0,0) {$s$};
			\node[BoxStyle] at (-0.75, -0.9){};
			\node at (-1.4, -0.9){$\m_1$};
			\node[BoxStyle] at (0.75, -0.9){};
			\node at (1.4, -0.9){$\m_2$};
			
			\node[state] (sa) at (-2,-2.4) {$s_a$};
			\node[BoxStyle] at (-2.75, -1.7){};
			\node[state, fill=green] (sb) at (2,-2.4) {$s_b$};
			\node[BoxStyle] at (2.75, -1.7){};
			
			\node[state] (t) at (8,0) {$t$};
			\node[BoxStyle] (bt) at (8,-0.7){};
			
			\node[state] (ta) at (6,-2.4) {$t_a$};
			\node[BoxStyle] at (5.25, -1.7){};
			\node[state, fill=green] (tb) at (10,-2.4) {$t_b$};
			\node[BoxStyle] at (10.75, -1.7){};
			
			\path[->] (s) edge node [near end, left] {$1$} (sa);
			\path[->] (s) edge node [near end, right] {$1$} (sb);
			
			\path[-] (t) edge node [midway, right] {} (bt);
			\path[->] (bt) edge node [near start, left] {$\frac{1}{2}$} (ta);
			\path[->] (bt) edge node [near start, right] {$\frac{1}{2}$} (tb);
			
			\path[->] (sa) edge [out=135, in=180, looseness=1.4] node [near end, above] {$1$} (s);
			\path[->] (sb) edge [out=45,in=0,looseness=1.4] node [near end, above] {$1$} (s);
			\path[->] (ta) edge  [out=135, in=180, looseness=1.4] node [near end, above] {$1$} (t);
			\path[->] (tb) edge[out=45,in=0,looseness=1.4] node [near end, above] {$1$} (t);
			
			\end{tikzpicture}
%		}
%		\vspace{-0.1em}
%		\caption{}%The MDP~$\mathcal{D}$ that no MD strategy witnesses $d_{pb}(s,t) \ls 1$.
		
%	\end{subfigure}
%	\vspace{-1em}
	\caption{In this MDP, no MD strategy witnesses $d_{\pb}(s,t) \ls 1$. States $s_b$ and $t_b$ have label $b$ while all other states have label $a$.}\label{fig:noMDstrategyPBneqOne}
	%\label{fig:noMDstrategyPBeqZeroneqOne}
\end{figure}

%%TODO say something about the upperbound
Finally, we show that $\PBeqZERO$ and  $\PBneqONE$ are $\sf NP$-complete. For some MDPs there exist memoryless strategies that make $d_{\pb}(s,t) = 0$ (resp.~$d_{\pb}(s,t) \ls 1$) but no such strategy is MD. Indeed, for the MDP in \cref{fig:noMDstrategyPBeqZero} (resp.\ \cref{fig:noMDstrategyPBneqOne}), it is easy to check that the only strategy $\alpha$ which makes $d_{\pb}(s,t) = 0$ (resp.~$d_{\pb}(s,t) \ls 1$),  requires randomness, that is, $\alpha(s)(\m_1)=\alpha(s)(\m_2) = \frac{1}{2}$, where $\m_1$ and $\m_2$ are the two available actions of state $s$. Thus, to show the {\sf NP} upper bound, we cannot simply guess an MD strategy. Instead, one can nondeterministically guess a partition of the states and check in polynomial time if the partition is a probabilistic bisimulation.

The hardness results for both problems are by reductions from the Subset Sum problem. The reduction is similar to \cite[Theorem~4.1]{FKS2020}. 

\begin{restatable}{theorem}{theoremPBDistanceZeroUBLB}\label{theorem:pbdistance-zero-NP-complete}
	The problem $\PBeqZERO$ is $\sf NP$-complete.
\end{restatable}

%PB<1
By \cite{TvB2018}, deciding whether $d_{\pb}(s,t) \ls 1$ in an LMC can be formulated as a reachability problem on a directed graph induced by the LMC. In addition to a partition, our {\sf NP} algorithm also guesses the graph induced by the LMC.

\begin{restatable}{theorem}{theoremPBDistanceNeqOneUBLB}\label{theorem:pbdistance-neq-one-NP-complete}
	The problem $\PBneqONE$ is $\sf NP$-complete.
\end{restatable} 

\section{Conclusions} \label{section:conclusion}
We have studied the computational complexity of qualitative comparison problems in labelled MDPs.
Motivated by the connection between obliviousness/anonymity and equivalence, we have devised polynomial-time algorithms to decide the existence of strategies for trace and bisimulation \emph{in}equivalence.
In case of trace inequivalence, there always exists an MD witness strategy, and our algorithm computes it.
The trace inequivalence algorithm is based on linear-algebra arguments that are considerably more subtle than in the LMC case.
For bisimulation inequivalence, MD strategies may not exist, but we have devised a polynomial-time algorithm to compute a memoryless strategy witnessing inequivalence; here the randomization is based on prime numbers to rule out certain ``accidental'' bisimulations.
The other 6 problems do not have polynomial complexity (unless ${\sf P} = {\sf NP}$), and we have established completeness results for all of them except $\TVeqONE$, where a complexity gap between {\sf NP} and $\ETR$ remains.

Concerning the relationship to interval Markov chains and parametric Markov chains mentioned in the introduction, the lower complexity bounds that we have derived in this paper carry over to corresponding problems in these models.
Transferring the upper bounds requires additional work, as, e.g., even the consistency problem for IMCs (i.e., whether there \emph{exists} a Markov chain conforming to an IMC) is not obvious to solve.
Nevertheless, the algorithmic insights of this paper will be needed.

\bibliography{ref}

\newpage\appendix\label{section:appendix}
%%Language Inequivalence
%TV>0
\section{Proofs of \cref{section:tvInequivalence}}\label{appendix:tvInequivalence}
The following lemma is technical and is only used in the proof of ~\cref{lemma:two-strategy-composition-new-strategy-construction}.
\begin{lemma}\label{lemma:vector-space-equal-flipping-bit}
	Let $j \in \nat$ and $\alpha$ be an MD strategy. Let $i \in S$. If $\vecc_i \not\in \V_{1}^{j}$, then $M_{\alpha}(w)\one = M_{\alpha^{i \to \m}}(w)\one$  for all $\m \in \A(i)$ and $w \in L^{\le j}$.
\end{lemma}
\begin{proof}
	Let $j \in \nat$ and $\alpha$ be an MD strategy. Let $i \in S$ and $\m \in \A(i)$. Assume $\vecc_i \not\in \V_{1}^{j}$. We prove this lemma by induction on the length of the trace $w$. 
	The base case where $|w| = 0$ is vacuously true. For the induction step, assume $M_{\alpha}(w)\one = M_{\alpha^{i \to \m}}(w)\one$ holds for all~$w$ of length less than~$j$. Let $a \in L$. 
	\begin{eqnarray*}
		&&M_{\alpha}(aw)\one\\
		&=&M_{\alpha}(a)M_{\alpha}(w)\one\\
		&=& M_{\alpha}(a)M_{\alpha^{i \to \m}}(w)\one \commenteq{$M_{\alpha}(w)\one = M_{\alpha^{i \to \m}}(w)\one$ by induction hypothesis}\\
		&=& ( M_{\alpha}(a) - M_{\alpha^{i \to \m}}(a) ) M_{\alpha^{i \to \m}}(w)\one + M_{\alpha^{i \to \m}}(aw) \one\\
		&=& x \cdot \vecc_i+ M_{\alpha^{i \to \m}}(aw) \one \commenteq{for some $x \in \Reals$},
		%	&=& M_{\alpha_{b^{+i}}}(aw) \one 
	\end{eqnarray*}
	where the last equality follows from the fact that since the two matrices $M_{\alpha}(a)$ and $M_{\alpha^{i \to \m}}(a)$ differ only in the $i$th row, $M_{\alpha}(a) - M_{\alpha^{i \to \m}}(a)$ is a matrix all whose rows except possibly the $i$th row are zero vectors.  The product of such a matrix with a column vector %$(M_{\alpha}(a) - M_{\alpha^{i \to \m}}(a)) \vecv$ 
	is then a multiple of $\vecc_i$.
	
	Since both $M_{\alpha}(aw)\one$ and $M_{\alpha^{i \to \m}}(aw)\one$ are in $\V_{1}^{j}$, the difference of them, which is $x \cdot \vecc_i$, is in $\V_{1}^{j}$ as well. However, as $\vecc_i \not\in \V_{1}^{j}$ by assumption, we have $x = 0$. Hence $M_{\alpha}(aw)\one = M_{\alpha^{i \to \m}}(aw)\one$.
\end{proof}

\lemmaStrategyComposition*
\begin{proof}
	Let $j \in \nat$. Let $\alpha_1$ and $\alpha_2$ be two MD strategies, $a \in L$ and $w \in L^{\le j}$. 
	
	Since $|w| \le j$ and for all $i$ with $\alpha_2(i) \not= \alpha(i)$ we have $\vecc_i \not\in \V_{1}^{j}$, by Lemma~\ref{lemma:vector-space-equal-flipping-bit}, we have \begin{equation}\label{eqn:switch-actions-alpha2}
	M_{\alpha_2}(w)\one = M_{\alpha}(w)\one
	\end{equation} 
	Then, 
	\begin{eqnarray*}
		&&M_{\alpha_{1}}(a)M_{\alpha_{2}}(w)\one \\
		&=&M_{\alpha_{1}}(a)M_{\alpha}(w)\one \commenteq{\eqref{eqn:switch-actions-alpha2}}\\
		&=& ( M_{\alpha_1}(a) - M_{\alpha}(a) ) M_{\alpha}(w)\one + M_{\alpha}(aw) \one.
	\end{eqnarray*}
	The first summand in the previous line is in the vector space $\V_{1}^{j}$, since $ ( M_{\alpha_1}(a) - M_{\alpha}(a) ) M_{\alpha}(w)\one  \in <\vecc_i:  \vecc_i \in  \V_{1}^{j}> \subseteq \V_{1}^{j}$. Thus, $M_{\alpha_{1}}(a)M_{\alpha_{2}}(w)\one \in <\V_{1}^{j} \cup \{M_{\alpha}(aw)\one\}> = <\V_{1}^{j} \cup \B (\{(\alpha, aw)\} )> $.
\end{proof}

The next lemma shows that a basis for $\V_{1}^{j}$ for some $j \ls |S|$ consisting only of MD vectors can be computed in polynomial time. 

\lemmaPolyTime*
\begin{proof}
	We have shown $<\B(P_{j+1})> \subseteq \V_{1}^{j+1} $ in the proof sketch.	To show the other direction, $\V_{1}^{j+1} \subseteq <\B(P_{j+1})>$, it suffices to show for all memoryless strategies $\alpha'$, $a \in L$ and $\vecb \in <\B(P_{j})>$, we have $M_{\alpha'}(a) \cdot \vecb \in <\B(P_{j + 1})>$. Let $\alpha'$ be an arbitrary memoryless strategy and $a \in L$. The matrix $M_{\alpha'}(a)$ can be expressed as a linear combination of the matrices from $\mathbb{M}$: 
	\[
	M_{\alpha'}(a) = M_{\alpha_0}(a) 
	+ \sum_{s \in S} 
	\left( - M_{\alpha_0}(a) + 
	\sum_{\m \in \A(s)}
	\alpha'(s)(\m) \cdot M_{\alpha_0^{s \to \m}}(a)
	\right)
	\]
	That is, there are $y_\alpha \in \Reals$ for all $\alpha \in \Sigma$ such that
	\begin{equation}\label{eqn:linear-combination-sm-dm-strategy}
	M_{\alpha'}(a) = \sum_{ \alpha \in \Sigma} y_{\alpha} \cdot M_{\alpha}(a)\,.
	\end{equation}
	Let $\vecb \in \B(P_{j})$. Then,
	\begin{eqnarray*}
		M_{\alpha'}(a)  \cdot \vecb &=&  \sum_{ \alpha \in \Sigma} y_{\alpha} \cdot  M_{\alpha}(a) \cdot \vecb \commenteq{\eqref{eqn:linear-combination-sm-dm-strategy}}%\\
		%&\in& <\B(P_{j+1}> \commenteq{\eqref{eqn:M-b-product-in-B-j+1}}
	\end{eqnarray*}
	
	Since each term of the summation, $y_{\alpha} \cdot  M_{\alpha}(a) \cdot \vecb$, is in $<\B(P_{j+1})>$ by \eqref{eqn:M-b-product-in-B-j+1},  $M_{\alpha'}(a)  \cdot \vecb $ is also in $<\B(P_{j+1})>$.
\end{proof}

For the proof of the following lemma we combine classical linear algebra arguments about equivalence checking (see, e.g., \cite{Tzeng}) with Lemma~\ref{lemma:V1j-basis-polynomial-time}.
\lemmaEqualityVSpaces*
\begin{proof}
	We prove the items in turn.
	\begin{enumerate}
		\item
		Let $j \ls |S|$.
		From Lemma~\ref{lemma:V1j-basis-polynomial-time} it follows that $\V_{1}^{j} \subseteq \V_{3}^{j}$.
		From the definitions of $\V_{1}^{j}, \V_{2}^{j}, \V_{3}^{j}$ we have $\V_{3}^{j} \subseteq \V_{2}^{j} \subseteq \V_{1}^{j}$.
		\item
		We have $\V_{1}^{j} \subseteq \V_{1}^{j+1}$ for all $j \in \nat$.
		Further we have for all $j \in \nat$:
		\begin{equation}
		\V_{1}^{j+1} = <\mathbf{v},\ M_\alpha(a) \mathbf{v} : \alpha \text{ is a memoryless strategy, } a \in L,\ \mathbf{v} \in \V_{1}^{j}> \label{eq:equality-of-v-spaces}
		\end{equation}
		It follows that if $\V_{1}^{j} = \V_{1}^{j+1}$ then $\V_{1}^{j} = \V_{1}^{k}$ holds for all $k \ge j$.
		Since $\dim \V_{1}^{j} \le |S|$ for all $j \in \nat$, it follows that $\V_{1}^{|S|-1} = \V_{1}^{k}$ holds for all $k \ge |S|-1$.
		By \eqref{eq:equality-of-v-spaces} we see that $\V_{1}^{|S|-1}$ contains~$\one$ and is closed under pre-multiplication with $M_\alpha(a)$ for any memoryless strategy~$\alpha$ and for any $a \in L$.
		But from the definition of~$\V_1$ we can derive that $\V_1$ is the smallest vector space that contains~$\one$ and has that closure property.
		Thus $\V_1 \subseteq \V_{1}^{|S|-1}$.
		We have:
		\begin{align*}
			\V_1
			&\subseteq \V_{1}^{|S|-1} && \text{as just shown}\\
			&= \V_{2}^{|S|-1} = \V_{3}^{|S|-1} && \text{by item~1}\\
			&\subseteq \V_3 \subseteq \V_2 \subseteq \V_1 && \text{from the definitions}
		\end{align*}
		Hence these vector spaces are all equal.\qedhere
	\end{enumerate} 
\end{proof}

%PB>0
\section{Proofs of \cref{section:pbInequivalence}}\label{appendix:pbInequivalence}
The following lemma is quite standard, it shows that the partition gets finer after each iteration of the partition refinement algorithm.

\lemmaPartitionInvariant*
\begin{proof}
	We prove the statement by induction on $i$. The base case where $i=0$ is vacuously true. For the induction step, assume $\mathord{\equiv_{X_{i+1}}} \subseteq \mathord{\equiv_{X_{i}}}$. Then, for each $\mathord{\equiv_{X_{i}}}$-equivalence class $E$, we have 	\begin{equation}\label{eqn:partition-finer-IH}
	E = \textstyle\bigcup_j E_j \text{ where } E_j \in S/ \mathord{\equiv_{X_{i+1}}}.
	\end{equation}
	
	Next, we show $\mathord{\equiv_{X_{i+2}}} \subseteq \mathord{\equiv_{X_{i+1}}}$. %we show $X_{i+2} = S / \equiv_{X_{i+1}}$, i.e., %$\J(X_{i+1}) = S / \equiv_{X_{i+1}}$. 
	
	Let $s, t \in S$ and $s \equiv_{X_{i+2}} t$. If $s = t$, then $s \equiv_{X_{i+1}} t$. Otherwise, assume $s \not= t$. By the definition of $\mathord{\equiv_{X_{i+2}}}$, we have $\ell(s) = \ell(t)$,  $|\varphi(s)(X_{i+2})| = 1$ and $\varphi(s)(X_{i+2})= \varphi(t)(X_{i+2})$. For all $\m_1, \m_2 \in \A(s)$, $\m_1', \m_2' \in \A(t)$ and $E \in X_{i+2}$, we have
	$$\varphi(s, \m_1)(E) = \varphi(s, \m_2)(E)= \varphi(t, \m_1')(E) = \varphi(t, \m_2')(E).$$
	Since $X_{i+2} = S/ \mathord{\equiv_{X_{i+1}}}$ and $X_{i+1} = S/ \mathord{\equiv_{X_{i}}}$, by \eqref{eqn:partition-finer-IH}, we have for all $\m_1, \m_2 \in \A(s)$, $\m_1', \m_2' \in \A(t)$ and all $E' \in X_{i+1}$, 
	$$\varphi(s, \m_1)(E') = \varphi(s, \m_2)(E')= \varphi(t, \m_1')(E') = \varphi(t, \m_2')(E').$$
	
	Thus, $|\varphi(s)(X_{i+1})| =  |\varphi(t)(X_{i+1})| = 1$ and $\varphi(s)(X_{i+1})= \varphi(t)(X_{i+1})$. By the definition of $\mathord{\equiv_{X_{i+1}}}$, we have $s \equiv_{X_{i+1}} t$.
\end{proof}

The next lemma shows that if $s \equiv_{X_{|S|-1}} t$ then $s \sim_{\D(\alpha)} t$ for all memoryless strategy $\alpha$. 
\lemmaXBisimulation*
\begin{proof}
	Let $X$ be a partition and $X = S/\mathord{\equiv_{X}}$. Let $\alpha$ be a memoryless strategy. We show that $\mathord{\equiv_{X}}$ is a probabilistic bisimulation in the induced LMC $\D(\alpha)$. By the definition of probabilistic bisimulation, it suffices to show that for all $(u, v) \in \mathord{\equiv_{X}}$, we have $\ell(u) = \ell(v)$ and $\tau(u)(E) = \tau(v)(E)$ for each $\mathord{\equiv_{X}}$-equivalence class $E$.
	
	Since $X = S/\mathord{\equiv_{X}}$, each element $E \in X$ is an $\mathord{\equiv_{X}}$-equivalence class. Let $u \equiv_{X} v$. We distinguish the following two cases: $u = v$ and $u \not= v$. If $u = v$, then $u \sim_{\D(\alpha)} v$ is vacuously true. Assume $u \not= v$. Let $\pi$ be a probability distribution over $X$ and $\pi = \varphi(u, \m_u)(X)$ for some $\m_u \in \A(u)$. By definition of $\equiv_{X}$, we have $\ell(u) = \ell(v)$ and for all $\m_u \in \A(u)$ and $\m_v \in A(v)$, $\varphi(u, \m_u)(X) = \varphi(v, \m_v)(X) = \pi$. 
	
	In the LMC $\D(\alpha)$, the transition probability from $u$ to $E \in X$ is 
	\begin{eqnarray*}
		\tau(u)(E) &=& \sum_{\m_u \in \A(u)} \alpha(u)(\m_u) \cdot \varphi(u, \m_u)(E) \\
		&=& \sum_{\m_u \in \A(u)} \alpha(u)(\m_u) \cdot \pi(E) \commenteq{$\pi(E) = \varphi(u, \m_u)(E)$}\\
		&= & \pi(E) \commenteq{$\alpha(u)$ is a probability distribution over $\A(u)$}\\
	\end{eqnarray*}
	
	Similarly, the transition probability from $v$ to $E \in X$, $\tau(v)(E)$, is also equal to $\pi(E)$. Thus, for all $\mathord{\equiv_{X}}$-equivalence class $E$, we have $\tau(u)(E) = \tau(v)(E)$. This completes the proof.
\end{proof}

\lemmaPolyPrimeFactorsN*
\begin{proof}
	Since $X_i$ for some $i$ is a partition of $S$, we have $|X_i| \le S$. Together with the fact that Algorithm~\ref{alg:polynomial-optimistic-partition-refinement} runs for at most $|S|$ iterations, $|\bigcup_{i} X_i|$ is polynomial in $|S|$. Thus, $|N|$ is polynomial in $\size(\D)$.
	
	Let $b \in N$. Since the smallest prime number is $2$, we have $\theta(b) \ls \log{b}$. Furthermore, since $\log{b}$ is the bit size of $b$, $\theta(b)$ is then polynomial in $\size(\D)$.  
	
	Finally, since $\theta(N) \le  |N| \cdot \max_{b \in N} \theta(b)$, $\theta(N)$ is polynomial in $\size(\D)$.
\end{proof}

\lemmaPolyPrimeFactors*
\begin{proof}
	We denote by $p(x)$ the number of primes less than or equal to a positive integer $x$. We show that there exists $x \in \nat$ such that $p(x) \ge |S| + \theta(N)$ and $x$ is polynomial in $\size(\D)$.
	
	Let $x \gr 55$ and $x  \ge (|S| + \theta(N))^2$. Then, 
	
	\begin{eqnarray*}
		|S| + \theta(N) &\le& \sqrt{x}   \\
		&=& \frac{x}{\sqrt{x}}\\
		&\ls& \frac{x}{\log{x} +2} \commenteq{$\sqrt{x} \gr \log{x} +2$}\\
		&\ls& p(x) \commenteq{$\frac{x}{\log{x} +2} \ls p(x)$ for $x \gr 55$ by \cite{Rosser1941}} 
	\end{eqnarray*}
	
	It follows that $x$ is polynomial in $\size(\D)$, as $\theta(N)$ is polynomial in $\size(\D)$ by \cref{lemma:prime-factors-in-N-polynomial}. For each positive integer $ i \le x$, we can check whether it is prime using the algorithm in \cite{AKS2004} and is coprime to all number in $N$. Each check can be done in polynomial time as shown in \cite{AKS2004} and that $|N|$ is polynomial in $\size(\D)$.
\end{proof}

\lemmaPartialStrategyConstruction*
\begin{proof}
	Following the proof sketch, we show the rest of the proof in detail.
	
	In Algorithm~\ref{alg:polynomial-optimistic-partition-refinement}, once a state $u$ satisfies $|\varphi(u)(X_i)| \not= 1$ for some partition $X_i$, it satisfies $|\varphi(u)(X_j)| \not= 1$ for all $j \ge i$ since $\mathord{\equiv_{X_{j}}} \subseteq \mathord{\equiv_{X_{i}}}$ for all $j \ge i$ by \cref{lemma:partition-gets-finer}. A state $u$ is only added to the domain of the partial strategy once (in Algorithm~\ref{alg:polynomial-construction-partial-strategy} line~\ref{alg:partial-strategy-for-condition}), which guarantees that $\alpha'_{i+1}$ does not overwrite $\alpha'_i$. It follows that for any $\alpha \sqsupseteq \alpha_{i+1}'$, it satisfies $\alpha \sqsupseteq \alpha_{i}'$.
	
	Let $s \not\equiv_{X_{i+1}} t$. Let $\alpha \sqsupseteq \alpha_{i+1}'$. If $s \not\sim^{i}_{\D(\alpha)} t$, then $s \not\sim^{i+1}_{\D(\alpha)} t$, since $\mathord{\sim_{\D(\alpha)}^{i+1}} \subseteq \mathord{\sim_{\D(\alpha)}^{i}}$ for all $i \in \nat$. Otherwise, assume $s \sim^{i}_{\D(\alpha)} t$. From $\alpha \sqsupseteq \alpha_{i}'$ and the induction hypothesis, it follows that $s \equiv_{X_{i}} t$. Since $s \not\equiv_{X_{i+1}} t$, by the definition of $\mathord{\equiv_{X_i}}$ and $\mathord{\equiv_{X_{i+1}}}$, we have $\ell(s) = \ell(t)$, $s \not= t$ and $|\varphi(s)(X_{i})| = |\varphi(t)(X_{i})| = 1$. 
	
	Towards a contradiction, assume $s \sim^{i+1}_{\D(\alpha)} t$. We show that under this assumption, for all $E\in X_{i+1}$, $\tau(s)(E) = \tau(t)(E)$ should hold. Let $E \in X_{i+1}$. Since  $X_{i+1} = S / \mathord{\equiv_{X_{i}}}$, $E$ is an equivalence class with respect to $\mathord{\equiv_{X_i}}$. By the induction hypothesis, we have
	\begin{equation}\label{eqn:equiv-Xi-coarser-than-sim-i}
	E = \textstyle\bigcup_j E_j \text{ where } E_j \in S/ \mathord{\sim^{i}_{\D(\alpha)}}.
	\end{equation}
	
	Then, 
	\begin{eqnarray*}
		\tau(s)(E) &=&\textstyle \sum_{E_j} \tau(s)(E_j) \commenteq{\eqref{eqn:equiv-Xi-coarser-than-sim-i}}\\
		&=& \textstyle\sum_{E_j} \tau(t)(E_j) \commenteq{$s \sim^{i+1}_{\D(\alpha)} t$}\\
		&=& \tau(t)(E)
	\end{eqnarray*}	 
	
	The two different prime numbers $p_s$ and $p_t$ are associated with state $s$ and $t$, respectively. If $|\varphi(s)(X_{i+1})| \not= 1$, $\alpha_{i+1}'(s)$ is defined using $p_s$ on line~4. It follows that $\tau(s)(E) \not= \tau(t)(E)$, since $p_s$ can divide the denominator of $\tau(s)(E)$ but not the denominator of $\tau(t)(E)$. The case when $|\varphi(t)(X_{i+1})| \not= 1$ is symmetrical. Otherwise, $|\varphi(s)(X_{i+1})| = |\varphi(t)(X_{i+1})|  = 1$. By the definition of $\equiv_{X_{i+1}}$, there must exist a set $E$ that $\tau(s)(E) \not= \tau(t)(E)$. 
	
	We now show in detail that for all the following cases, we have the contradiction that $\tau(s)(E) \not=\tau(t)(E)$ for some $E \in X_{i+1}$.	
	
	\begin{itemize}
		\item[-] Assume $|\varphi(s)(X_{i+1})| \not= 1$.
		From the construction of $\alpha_{i+1}'$ on line~3, we pick $\m_1, \m_2 \in \A(s)$ and $E \in X_{i+1}$ such that $\varphi(s, \m_1)(E) \gr \varphi(s, \m_2)(E)$.  In the LMC $\D(\alpha)$, the probability from $s$ to $E$ is 
		\begin{eqnarray*}
			\tau(s)(E) &=& \sum_{\m \in \A(s)}\alpha(s)(\m) \cdot \varphi(s, \m)(E)\\
			&=& \sum_{\m \in \A(s)}\alpha_{i+1}'(s)(\m) \cdot \varphi(s, \m)(E) \commenteq{$\alpha \sqsupseteq \alpha_{i+1}'$}\\
			&=& \frac{1}{p_s}\cdot\varphi(s, \m_1)(E) + (1 - \frac{1}{p_s})\cdot\varphi(s, \m_2)(E) \\
			&=& \varphi(s, \m_2)(E) + \frac{1}{p_s}\cdot \big(\varphi(s, \m_1)(E)-\varphi(s, \m_2)(E)\big)  \\
			&=& \frac{a_{s, \m_2, E}}{b_{s, \m_2, E}} + \frac{1}{p_s}\cdot \frac{c_{s, \m_1, \m_2, E}}{d_{s, \m_1, \m_2, E}} \\
			&=& \frac{a_{s, \m_2, E}  d_{s, \m_1, \m_2, E} p_s + c_{s, \m_1, \m_2, E} b_{s, \m_2, E}}{b_{s, \m_2, E} d_{s, \m_1, \m_2, E} p_s},
		\end{eqnarray*}
		where $a_{s, \m_2, E}, b_{s, \m_2, E}, c_{s, \m_1, \m_2, E}$ and $d_{s, \m_1, \m_2, E}$ are defined before \cref{lemma:prime-factors-in-N-polynomial}.
		The first summand of the numerator in the previous line can be divided by $p_s$. By \eqref{eqn:prime-number-coprime-with-N}, we have $p_s \nmid b_{s, \m_2, E}$ and $p_s \nmid c_{s, \m_1, \m_2, E}$. Thus, $p_s$ can not divide the second summand, and hence, not the numerator. For $u \in  S$ and $E \subseteq S$, we express $\tau(u)(E)$ as an irreducible fraction $\frac{x_{\tau(u)(E)}}{y_{\tau(u)(E)}}$ where $x_{\tau(u)(E)}$ and $y_{\tau(u)(E)}$ are coprime integers. It follows that $p_s \mid y_{\tau(s)(E)}$.
		
		%Since $p_s \nmid b_{v, m, B}$ for all $v \equiv_{X_i} s$ and $p_s \not= p_v$ for all $v \equiv_{X_{i}} s$ such that $p_v$ is defined, we have $p_s \nmid y_{\tau(t)(B)}$, and thus the desired contradiction that $\tau(s)(B) \not= \tau(t)(B)$.
		
		For state $t$, we have either $|\varphi(t)(X_{i+1})| \not= 1$ or $|\varphi(t)(X_{i+1})| = 1$. Assume $|\varphi(t)(X_{i+1})| \not= 1$. Similar to $s$, we have $\m_1', \m_2' \in \A(t)$ and $E' \in X_{i+1}$ such that $\varphi(t, \m_1')(E') \gr \varphi(t, \m_2')(E')$. In the LMC $\D(\alpha)$, the probability from $t$ to $E$ is 
		\begin{eqnarray*}
			\tau(t)(E) &=& \sum_{\m \in \A(t)}\alpha(t)(\m) \cdot \varphi(t, \m)(E)\\
			&=& \sum_{\m \in \A(t)}\alpha_{i+1}'(t)(\m) \cdot \varphi(t, \m)(E) \commenteq{$\alpha \sqsupseteq \alpha_{i+1}'$}\\\
			&=& \frac{1}{p_t}\cdot\varphi(t, \m_1')(E) + (1 - \frac{1}{p_t})\cdot\varphi(t, \m_2')(E) \\
			&=& \frac{1}{p_t}\cdot\frac{a_{t, \m_1', E}}{b_{t, \m_1', E}} + \frac{p_t-1}{p_t}\cdot\frac{a_{t, \m_2', E}}{b_{t, \m_2', E}} \\
			&=& \frac{a_{t, \m_1', E} b_{t, \m_2', E} + (p_t -1 ) \cdot a_{t, \m_2', E} b_{t, \m_1', E}}{p_t b_{t, \m_1', E} b_{t, \m_2', E}} 
		\end{eqnarray*}
		By \eqref{eqn:prime-number-coprime-with-N},  the two prime numbers $p_s$ and $p_t$ are different and $p_s \nmid b_{t, \m_1', E}, b_{t, \m_2', E}$. It follows that $p_s \nmid y_{\tau(t)(E)}$, and thus $\tau(s)(E) \not= \tau(t)(E)$.
		
		We consider the other case where $|\varphi(t)(X_{i+1})| = 1$. It follows that $\varphi(t, \m)(E)$ is the same for all $\m \in \A(t)$. Let $\m' \in \A(t)$. In the LMC $\D(\alpha)$, the probability from $t$ to $E$ is 
		\begin{eqnarray*}
			\tau(t)(E) &=& \sum_{\m \in \A(t)}\alpha(t)(\m) \cdot \varphi(t, \m)(E)\\
			&=& \varphi(t, \m')(E)\\
			&=& \frac{a_{t, \m', E} }{b_{t, \m', E}} 
		\end{eqnarray*}
		
		By \eqref{eqn:prime-number-coprime-with-N}, $p_s \nmid b_{t, \m', E}$. It follows that $p_s \nmid y_{\tau(t)(E)}$, and thus $\tau(s)(E) \not= \tau(t)(E)$.
		
		\item[-] Assume $|\varphi(s)(X_{i+1})| = 1$ and $|\varphi(t)(X_{i+1})| \not= 1$.
		To avoid redundancy, we do not show the proof as this case is similar to the case $|\varphi(s)(X_{i+1})| \not= 1$ and $|\varphi(t)(X_{i+1})| = 1$.
		\item[-] Assume $|\varphi(s)(X_{i+1})| = |\varphi(t)(X_{i+1})| = 1$. Since $s \not\equiv_{X_{i+1}} t$, by definition of $\equiv_{X_{i+1}}$, we have $\varphi(s)(X_{i+1}) \not= \varphi(t)(X_{i+1})$. Let $\m_s \in \A(s)$ and $\m_t \in \A(t)$. There exists a set $E \in X_{i+1}$ such that $\varphi(s, \m_s)(E) \not= \varphi(t, \m_t)(E)$. In the LMC $\D(\alpha)$, we have 
		\begin{eqnarray*}
			\tau(s)(E) &=&\sum_{\m \in \A(s)}\alpha(s)(\m) \cdot \varphi(s, \m)(E)\\ 
			&=& \varphi(s, \m_s)(E)\\
			&\not=& \varphi(t, \m_t)(E)\\
			&=&\sum_{\m \in (t)}\alpha(t)(\m) \cdot \varphi(t, \m)(E)\\ 
			&=& \tau(t)(E).
		\end{eqnarray*}
	\end{itemize}
	This completes the proof.
\end{proof}

%%Distance One
\section{Proofs of \cref{section:DistanceOne}}\label{appendix:DistanceOne}
%TV=1
\subsection{Proofs of $\TVeqONE$}\label{appendix:tvDistanceOne}
In this section, we show that the problem $\TVeqONE$ is in $\ETR$ and is {\sf NP}-hard. Recall that $\TVeqONE$ is the problem asking whether there is a memoryless strategy $\alpha$ for $\D$ such that the total variation distance of the two initial distributions is one in the induced labelled Markov chain $\D(\alpha)$, i.e., $d_{\tv}(\mu, \nu) = 1$.

Define the set 
$
R^{\mu, \nu} := \{(r_1,r_2) \in S \times S: \exists w \in L^*:  r_1 \in \support(\mu M(w)) \text{ and }  r_2 \in \support(\nu M(w)) \},
$
which can be computed in polynomial time as shown in \cite[Lemma~20]{CK2014}. For each $r_1 \in S$, define the projection $R^{\mu, \nu}_{r_1} := \{r_2 \in S: (r_1, r_2) \in R^{\mu, \nu}\}$. According to \cite[Theorem~21]{CK2014}, the following proposition holds.

\begin{proposition}\label{proposition:tvdistance-one-theorem}
	We have $d_{\tv}(\mu , \nu) \ls 1$ if and only if there are $r_1 \in S$ and subdistributions $\mu_1$ and $\mu_2$ such that 
	\begin{equation}
	\label{equation:conditions-vardistance-less-than-one}
	\mu_1 \equiv \mu_2\quad \text{ and } \quad r_1 \in \support(\mu_1) \quad \text{ and }\quad \support(\mu_2)  \subseteq  R^{\mu, \nu}_{r_1}
	\end{equation}   
\end{proposition}

%proof of the existence of a hyperplane that strictly separates the two disjoint convex polyhedra: https://math.stackexchange.com/questions/2143236/separating-disjoint-polyhedra-with-lp-duality
It is known that $\ETR$ is closed under {\sf NP}-reductions~\cite{CateKO13} which is needed for showing the membership of $\TVeqONE$ in $\ETR$, and later the membership of $\TVneqONE$ in $\ETR$. 

\theoryTVDistanceOneUB*
\begin{proof}
%By \cite[Proposition~3(b)]{CK2014}, the following statement holds: 
Let $B_{\alpha} \in \Reals^{S \times r}$ be a matrix consisting of $r \le |S|$ linearly independent columns which we denote by $\vecb_0, \cdots, \vecb_{r-1}$. Furthermore, we have  
\begin{itemize}
\item[-]
$\vecb_0 = \one$;
\item[-]
 $\vecb_i = M_\alpha(w_i) \one$ where $w_i \in L^{\le|S|}$ for all $1 \le i \ls r$.
\end{itemize}

The columns of $B_{\alpha}$ are linearly independent, i.e., $B_{\alpha}$ has full rank $r$, if and only if there exists a \emph{reduced QR factorization} of $B_{\alpha}$, i.e., there exist a matrix $Q \in \Reals^{S \times r}$ with orthonormal columns and an upper triangular matrix $R \in \Reals^{r \times r}$ with all diagonal entries being nonzero such that $B_{\alpha} = Q R$. 

The matrix $B_{\alpha}$ is a basis for the vector space $$ \left < M_\alpha(w) \cdot \one:  w \in L^{*} \right >$$ if and only if $B_\alpha$ is closed under pre-multiplication with $M_{\alpha}(a)$ for any $a \in L$, i.e., for each label $a \in L$, there exists a matrix $F(a) \in \mathbb{R}^{r \times r}$ such that  
$$M_{\alpha}(a) \cdot B_{\alpha} = B_{\alpha}F(a).$$

%(this is called a reduced QR factorization). Then the diagonal entries of R are all nonzero. Conversely, if Q has orthonormal columns, and the diagonal entries of R are all nonzero, the product Q R has full rank. So B \in \R^{S \times m} has full rank m if and only if it has a reduced Q R factorization Q R with diagonal entries of R all nonzero.

Let $I_n \in \Reals^{n \times n}$ denote the identity matrix of size $n$. Let $H_{\alpha} \in \Reals^{S \times r'}$ be a matrix consisting of $r'$ columns which are denoted by $\vech_0, \cdots, \vech_{r'-1}$. Furthermore, we require that all of the columns have length one and they are mutually orthogonal, i.e., $H_{\alpha}^{T}H_{\alpha} = I_{r'}$. It is an orthonormal basis for the vector space $<\mathbf{x}: B_{\alpha} ^{T} \mathbf{x} = \zero >$, i.e., the orthogonal complement of~$<B_{\alpha}>$, if and only if $ B_{\alpha}^{T} \vech_i= \zero$ for all $0 \le i \ls r'$ and $\mathit{rank}(B_\alpha) + \mathit{rank}(H_\alpha) = r+r'=|S|$. %(\underbrace{0, \cdots, 0}_\text{$|S|$ times} )

Recall that $\vecc_i \in \{0, 1\}^{S}$ is the column bit vector whose only non-zero entry is the $i$th one. For each $s\in S$, define a convex polyhedron
$$\mathcal{P}_{s} =\Big\{ \vecc_s + \sum_{t \in S}\lambda_t \vecc_t + \sum_{t \in R_{s}^{\mu, \nu}}\lambda_t' (-\vecc_t) \; \bigr\vert \; \lambda_t \ge 0,\ \lambda_t '\ge 0 \Big\} .$$
We call $\vecc_t$ for $t \in S$ and $-\vecc_t$ for $t \in R_{s}^{\mu, \nu}$ the \emph{spanning} vectors of~$\mathcal{P}_s$.

Assume the matrix $B_{\alpha}$ is a basis for $\left < M_\alpha(w) \cdot \one:  w \in L^{*} \right >$ and $H_{\alpha}$ is an orthonormal basis for the orthogonal complement of $<B_{\alpha}>$. We show that the two convex polyhedra $<H_{\alpha}>$ and $\mathcal{P}_s$ intersect if and only if $d_{\tv}(\mu, \nu) \ls 1$ in $\D(\alpha)$. We distinguish the following two cases:
\begin{itemize}
	\item[-] Assume $s \in R_{s}^{\mu, \nu}$. It is easy to check that $\zero \in <H_{\alpha}> \cap \mathcal{P}_s$. Define the two subdistributions $\mu_1$ and $\mu_2$ as $\mu_1 =\mu_2 = \delta_s$.  By \cref{proposition:tvdistance-one-theorem}, $d_{\tv}(\mu, \nu) \ls 1$ holds since $\mu_1$ and $\mu_2$ satisfy~\eqref{equation:conditions-vardistance-less-than-one}.
	
	\item[-] Assume $s \not\in R_{s}^{\mu, \nu}$. We first show the backward implication. From \cref{proposition:tvdistance-one-theorem}, there exist subdistributions $\mu_1$ and $\mu_2$ satisfying~\eqref{equation:conditions-vardistance-less-than-one}. Let $N = \mu_1(s) -\mu_2(s)$. Since $s \in \support(\mu_1)$ and $s \not\in R_{s}^{\mu,\nu}$, we have $N = \mu_1(s) \gr 0$. Define the vector $\vecv =  \frac{(\mu_1 - \mu_2)^T}{N}$. We can easily verify that it is in both  $<H_{\alpha}>$ and $\mathcal{P}_s$, and hence, $<H_{\alpha}> \cap \mathcal{P}_s \not= \emptyset$.
	
	For the converse,  assume $<H_{\alpha}> \cap \mathcal{P}_s \not= \emptyset$. Let $\vecv$ be a column vector such that $\vecv \in <H_{\alpha}>$ and $\vecv \in \mathcal{P}_s$. 
	Since $\vecv \in \mathcal{P}_{s}$ and $s \not\in R_{s}^{\mu, \nu}$, we have $\vecv(s) \ge 1$ and $\vecv(t) \ge 0$ for all $t \in S \setminus R_{s}^{\mu, \nu}$. Since $B_{\alpha}^{T} \vecv= \zero$ and $\vecb_0 = \one$, we have $\one^{T}\vecv = 0$. It follows that $\{t: \vecv(t) \ls 0\}  \subseteq R_{s}^{\mu, \nu} $ and $\{t: \vecv(t) \ls 0\} \not= \emptyset$. Let $N=\sum_{u \in S}|\vecv(u)|$. Define the two subdistributions $\mu_1$  and $\mu_2$ as follows:
	
	$
	\mu_1(u) = \left \{
	\begin{array}{ll}
	\frac{\vecv(u)}{N} & \mbox{if $\vecv(u) \gr 0$,}\\
	0 & \mbox{otherwise;}
	\end{array}
	\right .
	$ and 
	$
	\mu_2(u) = \left \{
	\begin{array}{ll}
	-\frac{\vecv(u)}{N} & \mbox{if $\vecv(u) \ls 0$,}\\
	0 & \mbox{otherwise.}
	\end{array}
	\right .
	$ 
	
	Since $\mu_1 -\mu_2 = \frac{\vecv}{N}$ and $\vecv$ is orthogonal with $<B_{\alpha}>$,  $\mu_1 -\mu_2$ is also  orthogonal with $<B_{\alpha}>$, and thus $\mu_1 \equiv \mu_2$.  Furthermore, we have $\mu_1(s) = \frac{\vecv(s)}{N} \ge \frac{1}{N} \gr 0$, and $\support(\mu_2) = \{t: \vecv(t) \ls 0 \} \subseteq R_{s}^{\mu,\nu}$. From \cref{proposition:tvdistance-one-theorem}, it follows that $d_{\tv}(\mu, \nu) \ls 1$ in the LMC $\D(\alpha)$ since $\mu_1$ and $\mu_2$ satisfy \eqref{equation:conditions-vardistance-less-than-one}. 
\end{itemize}

From the analysis above, to show that there exists a memoryless strategy $\alpha$ such that $d_{\tv}(\mu, \nu) = 1$ in the LMC $\D(\alpha)$, it suffices to show that there exists a memoryless strategy $\alpha$ such that $<H_{\alpha}> \cap \mathcal{P}_s = \emptyset$ for all $s \in S$. By \cite[Theorem~5.5.1]{lindahl2016}, if the two convex polyhedra $<H_{\alpha}>$ and $\mathcal{P}_{s}$ are disjoint, then there exists a hyperplane that strictly separates them, i.e., there exist $a_s, b_s \in \mathbb{R}$ and a row vector $v_s \in \mathbb{R}^{S}$ such that $v_s \cdot \vecx \le a_s$ for all $\vecx \in <H_{\alpha}>$, $v_s \cdot \mathbf{x} \ge b_s$ for all $\mathbf{x} \in \mathcal{P}_s$ and $a_s \ls b_s$. Since $\zero \in <H_{\alpha}>$, we have $a_s \ge 0$ and $0 \le a_s  \ls b_s$. For any column vector $\vech$ of $H_{\alpha}$, $a\cdot \vech$ is also in $<H_{\alpha}>$ for any $a\in \mathbb{R}$. It follows that $v_s \cdot \vech \le \frac{a_s}{a}$ and $v_s \cdot \vech \ge -\frac{a_s}{a}$ for all $a \gr 0$, and hence,  $$-\lim\limits_{a\to\infty}\frac{a_s}{a} \le v_s \cdot \vech \le \lim\limits_{a\to\infty}\frac{a_s}{a}.$$ 

Since both the left and right limits exist and are equal to zero, we have  $ v_s \cdot \vech = 0$ for all column vectors $\vech$ of $H_{\alpha}$. It follows that $v_s \cdot \vecx = 0$ for all $\vecx \in <H_{\alpha}>$, $v_s \cdot \vecx \ge b_s$ for all $\mathbf{x} \in \mathcal{P}_s$ and $b_s \gr 0$.

A memoryless strategy $\alpha$ for $\D$ can be characterised by numbers $x_{s,\m} \in [0,1]$ where $s \in S$ and $\m \in \A$ such that $x_{s,\m} = \alpha(s)(\m)$. We write $\bar{x}$ for the collection $(x_{s,\m})_{s\in S, \m \in \A}$. Thus, to decide if there exists a memoryless strategy such that $d_{\tv}(\mu, \nu) = 1$, we nondeterministically guess a set of $r-1$ words $w_i \in L^{\le|S|}$ where $1 \le i \ls r$ and a nonnegative integer $r'$, then check the following decision problem, which is a closed formula in the existential theory of the reals:

$\exists \bar{x}$, a matrix $B_\alpha \in \mathbb{R}^{S \times r}$ the columns of which are denoted by $\vecb_0, \cdots, \vecb_{r-1}$, a matrix $Q \in \Reals^{S \times r}$, an upper triangular matrix $R \in \Reals^{r \times r}$, matrices $F(a) \in \mathbb{R}^{r \times r}$ for all $a \in L$, a matrix $H_\alpha \in \mathbb{R}^{S \times r'}$ the columns of which are denoted by $\vech_0, \cdots, \vech_{r'-1}$, row vectors $v_s \in \mathbb{R}^{S}$ and $b_s \in \mathbb{R}$ for all $s \in S$  such that
\begin{flalign*}
&\left.\begin{array}{l} -\text{for all $s \in S: \textstyle\sum_{\m \in \A(s)}x_{s, \m} = 1$;} \qquad \commenteq{$\bar{x}$ characterising a memoryless strategy}\\ \end{array}\right. \MoveEqLeft[4]\\
&\left. \begin{array}{l}
- \vecb_0 = \one;\\
-\text{for all $1 \le i \ls r: \vecb_i = M_\alpha(w_i) \one$};\\
- Q^{T}Q = I_{r};\\
- R[i, i] \not= 0 \text{ for all } i; \\
- B_{\alpha} = QR; \\
- \text{for all labels $a \in L: M_{\alpha}(a) \cdot B_{\alpha} = B_{\alpha}F(a)$};\\
\end{array}\right\} \text{$B_{\alpha} \text{ is a basis for } < M_\alpha(w) \cdot \one:  w \in L^{*}  >$} \\
&\left.\begin{array}{l} 
- H_{\alpha}^{T}H_{\alpha} = I_{r'};\\
- \text{for all } 0 \le i \ls r': B_{\alpha}^{T} \vech_i= \zero;\\
- r + r' = |S|; \\ 
\end{array}\right\} \parbox{3.1in}{$H_\alpha$ is an orthonormal basis for the orthogonal complement of $<B_{\alpha}>$} \\
&\left. \begin{array}{l}
-\text{for all $s \in S: v_s \cdot \vech_i = 0$ for all $0 \le i \ls r'$};\\
-\text{for all $s \in S: v_s \cdot \vecx \ge b_s$ for all $\vecx \in \mathcal{P}_s$, i.e.,} \\
  \text{\phantom{for all} $v_s \cdot \vecc_s = b_s$ and $v_s \cdot \vecc \ge 0$ for all spanning vectors $\vecc$ of~$\mathcal{P}_s$};\\
-\text{for all $s \in S: b_s \gr 0$}.\\
\end{array}\right\} \parbox{1.0in}{for all $s \in S$, $<H_{\alpha}>$ and $\mathcal{P}_s$ do not intersect}
\end{flalign*}

%TODO this might be the algebraic formulation of the proof: $d_{\tv}(\mu, \nu) =1 \iff \forall s \in S: \exists \vecb_s \in B_{\alpha}$ such that $\vecb_s(s) \gr 0, \vecb_s(t) = 0 \; \forall t \in R_{s}^{\mu, \nu}$ and $\vecb_s(t) \ge 0 \; \forall t \in S \setminus R_{s}^{\mu, \nu}$.

%\begin{itemize}
%	\item[-]
%	for all $s \in S: \sum_{m \in \A(s)}x_{s, m} = 1$;
%	\item[-]
%	 $\vecb_0 = \one$;
%	\item[-]
%	 for all $1 \le i \le |S|-1: \vecb_i = M_\alpha(w_i) \one$;
%	\item[-]
%	for all label $a \in L: M_{\alpha}(a) \cdot B_{\alpha} = B_{\alpha}F(a)$;
%	\item[-]
%	$\mathit{rank}(B_\alpha) + \mathit{dim}(H_\alpha) = |S|$;
%	\item[-]
%	for all $s \in S: v_s \cdot \vecx \ls a_s$ for each column vector $\vecx$ of $H_{\alpha}$; 
%	\item[-] 
%	for all $s \in S$: $v_s \cdot \vecx \ge b_s$ for all $\vecx$ in the cone~$P_s$, i.e., $v_s \cdot \vecc_s = b_s$ and $v_s \cdot \vecc \ge 0$ for all spanning vectors $\vecc$ of~$P_s$;
%	\item[-]
%	 for all $s \in S: a_s \ls b_s$.\qedhere
%\end{itemize} 
\end{proof}

%NP-hardness
%\begin{figure}
%	\centering
%	%\input{figures/reductionfromSetSplit.tex}
%	\resizebox{0.5\columnwidth}{!}{\input{figures/reductionfromSetSplit.tex}}
%	\caption{The MDP~$\mathcal{D}$ in the reduction from Set Splitting for {\sf NP}-hardness of $\TVeqONE$ (or $\PBeqONE$). All states have the same label except $u$ and $v$. $C_i$ and $C_i'$ have an action to each $e_j \in C_i$ with probability one. If $e_1 \in C_1$, then there is an action from state $C_1$ and $C_1'$ to $e_1$. Similarly, if $e_n \in C_m$, there is an action from $C_n$ and $C_n'$ to state $e_n$.}\label{fig:reductionfromSetSplit}
%\end{figure}
%Define the set 
%$
%R^{\mu, \nu} := \{(r_1,r_2) \in S \times S: \exists w \in L^*:  r_1 \in \support(\mu M(w)) \text{ and }  r_2 \in \support(\nu M(w)) \},
%$
%which can be computed in polynomial time as has been shown in \cite[Lemma~20]{CK2014}. For each $r_1 \in S$, define the projection $R^{\mu, \nu}_{r_1} := \{r_2 \in S: (r_1, r_2) \in R^{\mu, \nu}\}$. 

Let $\mu_1, \mu_2$ be two subdistributions on $S$. We write $\mu_1 \le \mu_2$ to say that $\mu_1(u) \le \mu_2(u)$ for all $u \in  S$. According to \cite[Proposition~17]{CK2014}, the following proposition holds.

\begin{proposition}\label{proposition:tvdistance-one}
	We have $d_{\tv}(\mu , \nu) \ls 1$ if and only if there are $w \in L^{*}$ and $\mu_1$ and $\mu_2$ with $\mu_1 \le \mu M(w)$ and $\mu_2 \le \nu M(w)$ and $\mu_1 \equiv \mu_2$ and $|\mu_1| = |\mu_2| \gr 0$.
\end{proposition}

\theoryTVDistanceOneLB*
\begin{proof}
	Let $<S, \C>$ be an instance of Set Splitting where $S = \{e_1, \cdots, e_n\}$ and $\C = \{C_1, \cdots, C_m\}$ is a collection of subsets of $S$. We construct an MDP $\D$, see \figurename~\ref{fig:reductionfromSetSplit} for example, consisting of the following states: two states $s$ and $t$, a state $e_i$ for each element in $S$, twin states $C_j$ and $C_j'$ for each element in $\C$, two sink states $u$ and $v$. State $v$ has label $b$ while all other states have label $a$. State $s$ ($t$) has a single action which goes with uniform probability $\frac{1}{m}$ to states $C_i$ ($C_i'$) for $1 \le i \le m$.  For each $e_i \in C_j$, there is an action from state $C_j$ and $C_j'$ leading to state $e_i$ with probability one. Each state $e_i$ has two actions going to the sink states $u$ and $v$ with probability one, respectively. We show that 
	$$<S, \C> \in {\mbox{Set Splitting}} \iff \exists \, \mbox{memoryless strategy $\alpha$ such that} \; d_{\tv}(\mu, \nu) = 1 \, \mbox{in} \; \D(\alpha).$$
	
	Intuitively, making $C_i$ (resp.~$C_i'$) select the transition to $e_j$ simulates the membership of $e_j$ in $S_1$ (resp.~$S_2$).
	
	($\implies$)
	Let $S_1$ and $S_2$ be the two disjoint sets that partition $S$ and split the elements in $\C$. For the MDP $\D$, we define an MD strategy $\alpha$ as follows: 
	let state $e_i \in S_1$ select the action transitioning to $u$ and state $e_i \in S_2$ the action to $v$; let state $C_i$ select an available action that goes to a state in $S_1$ and $C_i'$ an available action that goes to a state in $S_2$. 
	
	We show that $d_{\tv}(\mu, \nu) = 1$ in the LMC $\D(\alpha)$. Let $\mu_1$ and $\mu_2$ be subdistributions over the states that are reachable from $s$ and $t$, respectively. Let $E' \in \mathcal{F}$ be a set of words always ending with infinite number of $b$'s. Since a word emitted by running $\D(\alpha)$ from an arbitrary state in $\support(\mu_2)$ always ends with infinitely many $b$'s, we have $\Pr_{\mu_2}(E') \gr 0$. On the other hand, a word emitted by running $\D(\alpha)$ from an arbitrary state in $\support(\mu_1)$ always ends with infinitely many $a$'s, we have $\Pr_{\mu_1}(E') = 0$. Then, 
	\begin{align*}
		d_{\tv}(\mu_1, \mu_2)
		=&  \sup_{E \in \mathcal{F}} |\textstyle\Pr_{\mu_1}(E) - \textstyle\Pr_{\mu_2}(E)| \\
		 \ge & |\textstyle\Pr_{\mu_1}(E') - \textstyle\Pr_{\mu_2}(E')|  \gr 0
	\end{align*}	   
	By Proposition~\ref{proposition:tvdistance-one}, we have $d_{\tv}(\mu,\nu) = 1$ in the LMC $\D(\alpha)$.
	
	($\impliedby$)
	Let $\alpha$ be a memoryless strategy for $\D$ such that $d_{\tv}(\mu, \nu) = 1$. Let~$\tau$ be the transition function for the LMC $\D(\alpha)$. Let $S_1 = \bigcup_{C_i}\support(\tau(C_i))$ and $S_2 = S \setminus S_1$. Let $S_2' = \bigcup_{C_i'}\support(\tau(C_i'))$. It suffices to show that $S_2' \subseteq S_2$ and $S_1$ and $S_2'$ split the elements of $\C$.

	Since $S_2' \bigcap S_1 = \emptyset$, otherwise $d_{\tv} (\mu, \nu) \ls 1$ by Proposition~\ref{proposition:tvdistance-one}. We have $S_2' \subseteq S_2$. We prove by contradiction that $S_1$ and $S_2'$ split the elements of $\C$. Assume there is a set $C_i \in \C$ which is not split by $S_1$ and $S_2'$. Furthermore, without loss of generality, assume $C_i \subseteq S_1$, that is, for all states $e \in C_i: e \in \support(\tau(C_i))$. Since state $C_i$ and $C_i'$ have the same successors in the MDP $\D$, there must exist a state $e' \in C_i$ such that $e' \in \support(\tau(C_i'))$. Let $\mu_1 = \mu_2 =\delta_{e'}$. We have $d_{\tv}(\mu_1, \mu_2) = 0$, which leads to the desired contradiction $d_{\tv} (\mu, \nu) \ls 1$ in $\D (\alpha)$ by Proposition~\ref{proposition:tvdistance-one}.
\end{proof}

%PB=1
\subsection{Proofs of $\PBeqONE$}\label{appendix:pbDistanceOne}
Next, we show that the problem $\PBeqONE$ is {\sf NP}-complete. Recall that $\PBeqONE$ is the problem asking whether there is a memoryless strategy $\alpha$ for $\D$ such that the probabilistic bisimilarity distance of the two initial states is one in the induced labelled Markov chain $\D(\alpha)$, i.e., $d_{\pb}(s, t) = 1$.

\begin{definition}\label{definition:pb-distance-graph}
	The directed graph $G = (V, E)$ is defined by
	\[
	\begin{array}{rcl}
	V &=& \{ (u, v): \ell(u) = \ell(v) \}\\
	E &=& \{\, <(u, v), (s', t')>: \tau(s')(u) \gr 0 \wedge \tau(t')(v) \gr 0 \,\}
	\end{array}
	\]	
\end{definition}

By \cite[Theorem~4, Proposition~5]{TvB2018}, the following proposition holds.

\begin{proposition}\label{proposition:pbdistance-neq-one-graph-reachability}
	We have $d_{\pb}(s, t) \ls 1$ if and only if in the graph $G = (V, E)$ the vertex $(s, t)$ is reachable from some $(u, v)$ with $u \sim v$.
\end{proposition}

\begin{theorem}\label{theorem:pb-distanceone-ub}
	The problem $\PBeqONE$ is in {\sf NP}.
\end{theorem}
\begin{proof}
 Suppose there exists a memoryless strategy $\beta$ such that $d_{\pb}(s, t) = 1$ in $\D(\beta)$. Let $G$ be the graph of \cref{definition:pb-distance-graph} induced by the LMC $\D(\beta)$. Consider an MDP $\D' = <S,\A', L, \varphi', \ell>$, which is over the same state space as $\D$ but is restricted to choose actions that conform to the graph $G$. Thus, $\beta$ is also a strategy of $\D'$. Furthermore, we have $\D'(\alpha) = \D(\alpha)$ for all memoryless strategy $\alpha$ of~$\D'$.

According to \cref{theorem:polynomial-time-compute-strategy-PBneqZero}, a memoryless strategy $\alpha'$ of $\D'$ such that $\mathord{\sim_{\D'(\alpha')}} \subseteq \mathord{\sim_{\D'(\alpha)}}$ for all memoryless strategy~$\alpha$ can be computed in polynomial time. Thus, we have 
%\begin{equation}\label{equation:relation-alpha'-coarser-than-beta}
$\mathord{\sim_{\D'(\alpha')}} \subseteq \mathord{\sim_{\D'(\beta)}}$,
%\end{equation} 
that is, if $u \not\sim_{\D'(\beta)} v$ then $u \not\sim_{\D'(\alpha')} v$ for $u, v \in S$. Let $G'$ be the graph of \cref{definition:pb-distance-graph} for the LMC $\D'(\alpha')$. Since $\D'$ conforms to $G$, $G'$ is a subgraph of $G$. Let $R$ and $R'$ be the set of state pairs that can reach $(s, t)$ in $G$ and $G'$, respectively.  We have  $R' \subseteq R$.

According to \cref{proposition:pbdistance-neq-one-graph-reachability}, since $d_{\pb}(s, t) = 1$ in the LMC $\D'(\beta)$, we have $u \not\sim_{\D'(\beta)} v$ for all $(u, v) \in R$. By $\mathord{\sim_{\D'(\alpha')}} \subseteq \mathord{\sim_{\D'(\beta)}}$ and $R' \subseteq R$, we have $u \not\sim_{\D'(\alpha')} v$ for all $(u, v) \in R'$. By \cref{proposition:pbdistance-neq-one-graph-reachability}, we have $d_{\pb}(s, t) = 1$ in the LMC $\D'(\alpha')$, and hence, $\alpha'$ is a memoryless strategy that witnesses $d_{\pb}(s, t) = 1$. 

This induces the following nondeterministic algorithm: we guess the graph $G$ and check whether $d_{\pb}(s, t) = 1$ holds in $\D(\alpha')$, where both the construction of the memoryless strategy $\alpha'$ (using Algorithm~\ref{alg:polynomial-construction-partial-strategy}) and the checking of $d_{\pb}(s, t) = 1$ are in polynomial time.
\end{proof}

\begin{theorem}\label{theorem:pb-distanceone-lb}
	The Set Splitting problem is polynomial-time many-one reducible to $\PBeqONE$, hence $\PBeqONE$ is {\sf NP}-hard. 
\end{theorem}
\begin{proof}
	Given an instance of Set Splitting $<S, \C>$ where $S = \{e_1, \cdots, e_n\}$ and $\C = \{C_1, \cdots, C_m\}$ is a collection of subsets of $S$, we construct the same MDP $\D$ as shown in \cref{theorem:tvdistance-one-NP-hardness}, see \figurename~\ref{fig:reductionfromSetSplit} for example. We show that 
	$$<S, C> \in {\mbox{Set Splitting}} \iff \exists \, \alpha \;\mbox{for $\D$ such that} \; d_{\pb}(s, t) = 1 \, \mbox{in} \; \mathcal{D}(\alpha).$$
	
	($\implies$)
	Let $S_1$ and $S_2$ be the two disjoint sets that partition $S$ and split the elements of $C$. According to Theorem~\ref{theorem:tvdistance-one-NP-hardness}, there exists a memoryless strategy $\alpha$ such that $d_{\tv}(\delta_s, \delta_t) = 1$ in the induced LMC $\D(\alpha)$. Since probabilistic bisimilarity distance is an upper bound of the total variational distance \cite{CBW2012},  we have that $d_{\pb}(s, t) = 1$ in $\D(\alpha)$.
	
	($\impliedby$)
	Let $\alpha$ be a memoryless strategy for $\D$ such that $d_{\pb}(s, t) = 1$ in the LMC  $\D(\alpha)$. Let~$\tau$ be the transition function for the LMC $\D(\alpha)$. Let $S_1 = \bigcup_{C_i}\support(\tau(C_i))$ and $S_2 = S \setminus S_1$. Let $S_2' = \bigcup_{C_i'}\support(\tau(C_i'))$. It suffices to show that $S_1$ and $S_2'$ split the elements of $\C$ and $S_2' \subseteq S_2$.
	
	Since $d_{\pb}(s, t) = 1$, by definition of probabilistic bisimilarity distance, $d_{\pb}(C_i, C_j') = 1$ for any choice of $C_i$ and $C_j'$. We can obtain, by the same argument,  $d_{\pb}(e_k, e_l) = 1$ for any $e_k \in \support(\tau(C_i))$ and $e_l \in \support(\tau(C_j'))$. Thus, we have $\support(\tau(C_i)) \cap \,\support(\tau(C_j')) = \emptyset$ for any choice of $C_i$ and $C_j'$. It follows that $S_1 \cap S_2' = \emptyset$, that is, $S_2' \subseteq S_2$. Furthermore, for any set $C_i \in \C$, there are two states $e_k, e_l \in C_i$ such that $e_k \in \support(\tau(C_i))$ and $e_l \in \support(\tau(C_i'))$, that is, $e_k$ and $e_l$  split the set $C_i$.  
\end{proof}

\theoremPBDistanceOneUBLB*
\begin{proof}
	It follows from \cref{theorem:pb-distanceone-ub} and \cref{theorem:pb-distanceone-lb}.
\end{proof}

%%Remaining Problems
\section{Proofs of \cref{section:summaryDistanceZeroAndNeqOne}}
%TV=0 
\subsection{Proofs of $\TVeqZERO$}\label{appendix:tvDistanceZERO}
In this section we show that the problem $\TVeqZERO$ is $\ETR$-complete. Recall that $\TVeqZERO$ is the problem asking whether there is a memoryless strategy $\alpha$ for $\D$ such that the total variation distance of the two initial distributions is zero in the induced labelled Markov chain $\D(\alpha)$, i.e., $d_{\tv}(\mu, \nu) = 0$.

The following proposition is adapted from \cite[Proposition~10]{14KW-ICALP}, which will be used to prove Theorem~\ref{theorem:trace-equivalence-reduce-to-ETR}.
\begin{proposition}
	\label{prop:language-equivalence}
	Let $\mathcal{M} = <S,L,\tau,\ell>$ be an LMC and $\mu$ and $\nu$ be two (sub)distributions. We have that $\mu \equiv \nu$ if and only if there exists $F \in \mathbb{R}^{S \times S}$ such that
	\begin{itemize}
		\item[-]
		the first row of $F$ is $\mu -\nu$;
		\item[-]
		$F \mathbf{1} = \mathbf{0}$
		and for each label $a \in L$ there exists a matrix $B(a) \in \mathbb{R}^{S \times S}$ such that 
		$$F  M(a) = B(a) F.$$ 
	\end{itemize}
\end{proposition}

\begin{theorem}
	\label{theorem:trace-equivalence-reduce-to-ETR}
	The problem $\TVeqZERO$ is in $\ETR$.%the existential theory of the reals.
\end{theorem}
\begin{proof}
	The proof is very similar to the one of \cite[Theorem~4.3]{FKS2020}. 
	
	%Let $\mathcal{D} = <S, \mathcal{A}, L, \delta>$ be an MDP. Let $\mu$ and $\nu$ be two initial distributions on $S$. 
	A memoryless strategy $\alpha$ for $\D$ can be characterised by numbers $x_{s,\m} \in [0,1]$ where $s \in S$ and $\m \in \A$ such that $x_{s,\m} = \alpha(s)(\m)$. We write $\bar{x}$ for the collection $(x_{s,\m})_{s\in S, \m \in \A}$. 
	
	According to Proposition~\ref{prop:language-equivalence}, in the LMC $\mathcal{D}(\alpha)$, we have $\mu \equiv \nu$  if and only if the following decision problem, which is a closed formula in the existential theory of the reals, has answer ``yes'':
	
	$\exists \bar{x}$, matrices $B(a) \in \mathbb{R}^{S \times S}$ for all $a \in L$ and a matrix $F \in \mathbb{R}^{S \times S}$ such that 
	\begin{itemize}
		\item[-]
		$\sum_{m \in \A(s)}x_{s, m} = 1$ for all $s \in S$;
		\item[-]
		the first row of $F$ is $\mu - \nu$;
		\item[-]
		$F \mathbf{1} = \mathbf{0}$;
		\item[-]
		$F  M_{\alpha}(a) = B(a) F$ for all $a \in L$. \qedhere
	\end{itemize}
\end{proof}
 
 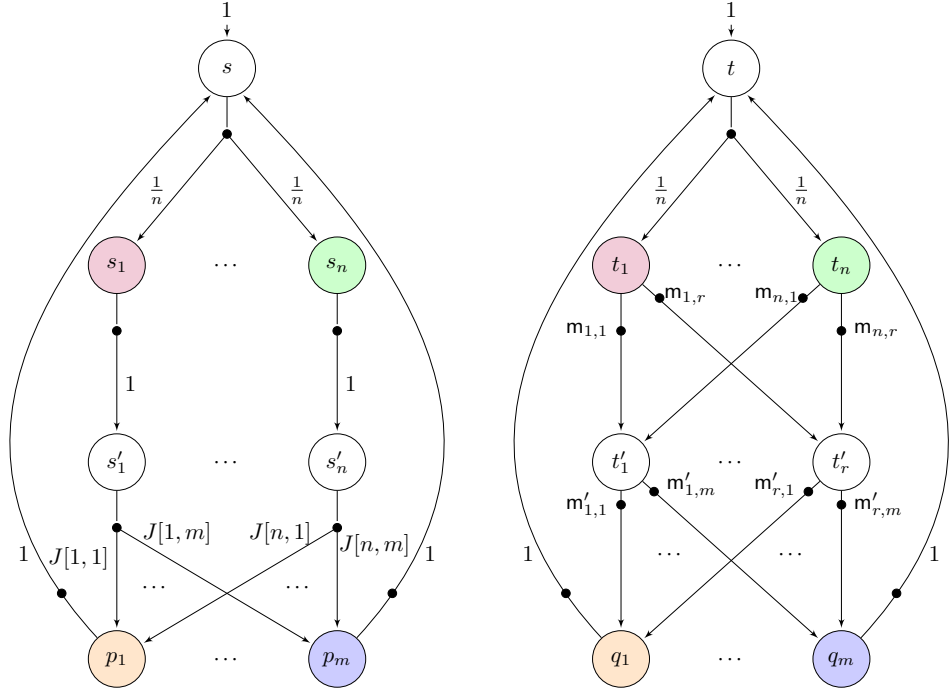
\begin{figure}[h]
	\begin{minipage}{0.45\linewidth}
		\centering
		\begin{adjustwidth*}{}{0em} \resizebox{1.3\columnwidth}{!}{\centering
\tikzstyle{BoxStyle} = [draw, circle, fill=black, scale=0.4,minimum width = 1pt, minimum height = 1pt]
\rotatebox{270}{
\begin{tikzpicture}[xscale=.6,scale=0.85,>=latex',shorten >=1pt,node distance=3cm,on grid,auto]
  \node[state,initial,initial text={\rotatebox{90}{$1$}}] (qin) at (-8,2) {\rotatebox{90}{$s$}};
  \node[BoxStyle] (bqin) at (-6,2){};
  
 \node[state,fill=green!20] (qin1) at (-2,4) {\rotatebox{90}{$s_{n}$}};
 \node[BoxStyle] (bqin1) at (0,4){};

 \node[state,fill=purple!20] (qinn) at (-2,0) {\rotatebox{90}{$s_{1}$}};
 \node[BoxStyle] (bqinn) at (0,0){};
 \node[label] at (-2,2.1) {$\vdots$};
 
 \node[state] (q1) at (4,4) {\rotatebox{90}{$s_{n}'$}};
 \node[BoxStyle] (bq1) at (6,4){};
 \node[label] at (4,2.1) {$\vdots$};
 \node[state] (qn) at (4,0) {\rotatebox{90}{$s_{1}'$}};
 \node[BoxStyle] (bqn) at (6,0){};
 \node[state,fill=blue!20] (qfi1) at (10,4) {\rotatebox{90}{$p_{m}$}};
 \node[state,fill=orange!20] (qfin) at (10,0) {\rotatebox{90}{$p_{1}$}};
 \node[label] at (10,2.1) {$\vdots$};
 %\node[BoxStyle] (aq) at (13,2){};
 \node[BoxStyle] at (8,5){};
 \node[BoxStyle] at (8,-1){};

 %\path[->] (q0) edge node [ near start] {$yyy$} node [very near end] {$xxx$} (q1);
 \path[-] (qin) edge  node [midway, right] {} (bqin);
 \path[->] (bqin) edge  node [midway, above, yshift=1mm] {\rotatebox{90}{$\frac{1}{n}$}} (qin1);
 \path[->] (bqin) edge  node [midway, below, yshift=-1mm] {\rotatebox{90}{$\frac{1}{n}$}} (qinn);
   
 \path[-] (qin1) edge  node [midway, right] {} (bqin1);
 \path[->] (bqin1) edge  node [midway, above] {\rotatebox{90}{$1$}} (q1);

 \path[-] (qinn) edge  node [midway, right] {} (bqinn);
 \path[->] (bqinn) edge  node [midway, above] {\rotatebox{90}{$1$}} (qn);
 %\path[->] (bqin) edge  node [midway, right] {$\frac{1}{n},a_n$} (qn);

 \path[-] (q1) edge  node [midway, right] {} (bq1);
 \path[->] (bq1) edge [bend left=0] node [very near start, above, yshift=-1mm] {\rotatebox{90}{$J[n,m]$}} (qfi1);%, xshift=2mm
 \node[label] at (7.8,3.4) {$\vdots$};
 \path[->] (bq1) edge [bend right=0] node [very near start, below, xshift=-2mm, yshift=2mm] {\rotatebox{90}{$J[n,1]$}} (qfin);%, xshift=-2mm, yshift=-1mm

 \path[-] (qn) edge  node [midway, right] {} (bqn);
 \path[->] (bqn) edge [bend left=0] node [very near start, above,xshift=-2mm, yshift=-2mm] {\rotatebox{90}{$J[1,m]$}} (qfi1);
 \node[label] at (7.8,0.8) {$\vdots$};
 \path[->] (bqn) edge [bend right=0] node [near start, below] {\rotatebox{90}{$J[1,1]$}} (qfin);

 %\path[->] (qfi) edge [loop right] node [right, midway] {$c,1$}  (qfi);
 \draw [->,postaction={decorate}] (qfi1) to [out=150, in=30] node [pos=0.15, above] {\rotatebox{90}{$1$}} (qin);
 %\node[BoxStyle] () at (6,6.5){};
  
 \draw [->,postaction={decorate}] (qfin) to [out=-150, in=-30] node [pos=0.15, below] {\rotatebox{90}{$1$}} (qin);
  % \node[BoxStyle] () at (6,-2.5){};
   
  % \path[-] (qfi1) edge  node [midway, right] {} (aq);
  % \path[-] (qfin) edge  node [midway, right] {} (aq);
   
\end{tikzpicture}
}}
		\end{adjustwidth*}
	\end{minipage}
	\hspace{0.1cm}
	\begin{minipage}{0.45\linewidth}
		\centering
		\begin{adjustwidth*}{}{0em}
			\resizebox{1.3\columnwidth}{!}{\centering

\tikzstyle{BoxStyle} = [draw, circle, fill=black, scale=0.4,minimum width = 1pt, minimum height = 1pt]
\rotatebox{270}{
\begin{tikzpicture}[xscale=.6,scale=0.85,>=latex',shorten >=1pt,node distance=3cm,on grid,auto]
 \node[state,initial,initial text={\rotatebox{90}{$1$}}] (pin) at (-8,2) {\rotatebox{90}{$t$}};
 \node[BoxStyle] (bpin) at (-6,2){};

 \node[state,fill=green!20] (p1) at (-2,4) {\rotatebox{90}{$t_{n}$}};
 \node[label] at (-2,2.1) {$\vdots$};
 \node[state,fill=purple!20] (pn) at (-2,0) {\rotatebox{90}{$t_{1}$}};

 \node[BoxStyle,label={above}:  \rotatebox{90}{$\m_{n,r}$}] at (0,4){};
 \node[BoxStyle,label={left,xshift=3mm,yshift=-4mm}: \rotatebox{90}{$\m_{n,1}$}] at (-1,3.3){};

 \node[BoxStyle,label={below}:\rotatebox{90}{$\m_{1,1}$}] at (0,0){};
 \node[BoxStyle,label={left, xshift=3mm,yshift=4mm}:\rotatebox{90}{$\m_{1,r}$}] at (-1,.7){};

 \node[state] (l1) at (4,4) {\rotatebox{90}{$t_{r}'$}};
 \node[label] at (4,2.1) {$\vdots$};
 \node[state] (lk) at (4,0) {\rotatebox{90}{$t_{1}'$}};

 \node[BoxStyle, label=above:\rotatebox{90}{$\m'_{r,m}$}] at (5.3,4){};
 \node[BoxStyle,label={below,xshift=-1mm}:\rotatebox{90}{$\m'_{r,1}$}] at (4.9,3.4){};

 \node[BoxStyle,label={below}:\rotatebox{90}{$\m'_{1,1}$}] at (5.3,0){};
 \node[BoxStyle,label={above,xshift=-1mm}:\rotatebox{90}{$\m'_{1,m}$}] at (4.9,.6){};

\node[state,fill=blue!20] (pfi1) at (10,4) {\rotatebox{90}{$q_{m}$}};
\node[state,fill=orange!20] (pfin) at (10,0) {\rotatebox{90}{$q_{1}$}};
\node[label] at (10,2.1) {$\vdots$};
%\node[BoxStyle,label={right}:$\m_{\fin}$] (aq) at (13,2){};
\node[BoxStyle] at (8,5){};
\node[BoxStyle] at (8,-1){};

\path[-] (pin) edge  node [midway, right] {} (bpin);
\path[->] (bpin) edge  node [midway, above, yshift=1mm] {\rotatebox{90}{$\frac{1}{n}$}} (p1);
\path[->] (bpin) edge  node [midway, below, yshift=-1mm] {\rotatebox{90}{$\frac{1}{n}$}} (pn);

\path[->] (p1) edge node  [very near end] {} (l1);
\path[->] (p1) edge node  [very near end,right] {} (lk);
 
\path[->] (pn) edge node [very near end,right] {} (l1);
\path[->] (pn) edge node [very near end,above] {} (lk);

 \path[->] (l1) edge [bend left=0]  node [near end, above] {} (pfi1);
 \node[label] at (6.8,3.2) {$\vdots$};
 \path[->] (l1) edge [bend right=0] node [near end, above] {} (pfin);

 \path[->] (lk) edge [bend left=0] node [near end, below] {} (pfi1);
 \node[label] at (6.8,1) {$\vdots$};
 \path[->] (lk) edge [bend right=0] node[yshift=-1mm] [near end, below] {} (pfin);

 \draw [->,postaction={decorate}] (pfi1) to [out=150, in=30] node [pos=0.15, above] {\rotatebox{90}{$1$}} (pin);
  \draw [->,postaction={decorate}] (pfin) to [out=-150, in=-30] node [pos=0.15, below] {\rotatebox{90}{$1$}} (pin);

\end{tikzpicture}
}}
		\end{adjustwidth*}
	\end{minipage}
	\caption{The MDP~$\mathcal{D}$ in the reduction for $\ETR$-hardness of $\TVeqZERO$ (or $\TVneqONE$). The labels of the states are as follows: $\ell(s_i) = \ell(t_i) = a_i$ for $1 \le i \le n$, $\ell(p_j) = \ell(q_j) = b_j$ for $1 \le j \le m$ and all remaining states have label $c$.}\label{fig:reductionfromNMFMDP}
\end{figure}

To show that the problem $\TVeqZERO$ is hard for $\ETR$, we present the reduction from the nonnegative matrix factorization (NMF) problem. Given the instance of the NMF, a nonnegative matrix $J \in \mathbb{Q}^{n \times m}$ and a number $r \in \mathbb{N}$, we construct an MDP $\mathcal{D}$; see \figurename~\ref{fig:reductionfromNMFMDP}. Similar to \cite[Theorem~4.5]{FKS2020}, we assume, without loss of generality, that~$J$ is a stochastic matrix. The left part is an LMC. The transition probability from $s_i'$ to $p_j$ in the LMC encodes the entry $J[i, j]$. 

The other part is an MDP; see the right of \figurename~\ref{fig:reductionfromNMFMDP}.
The initial state $t$ transitions to the successors $t_1, \cdots, t_n$ with equal probabilities. In each~$t_i$ where~$1\leq i\leq n$, there are $r$ actions~$\m_{i,1}, \m_{i,2},\dots,\m_{i,r}$
where~$\varphi(t_i, \m_{i,k})=\delta_{t_k'}$ for $1\leq k \leq r$.
In each~$t_k'$, there are $m$ actions~$\m'_{k,1}, \m'_{k,2},\dots,\m'_{k,m}$ where
$\varphi(t_k', \m'_{k,j})=\delta_{q_{j}}$ where $1\leq j\leq m$.
In state~$q_{j}$, there is only one action which transitions back to state $t$ with probability one.

The probabilities of choosing the action~$\m_{i,k}$ in~$s_{i}$ and
choosing~$\m'_{k,j}$ in~$s_k'$ simulate the entries of~$A[i,k]$ and~$W[k,j]$.

The distribution $\mu$ and $\nu$ are the Dirac distribution on $s$ and $t$, respectively. The labels of the states are as follows: $\ell(s_i) = \ell(t_i) = a_i$ for $1 \le i \le n$, $\ell(p_j) = \ell(q_j) = b_j$ for $1 \le j \le m$ and all remaining states have label $c$. The construction is very similar to the one in \cite[Theorem~4.5]{FKS2020}.

The following proposition is technical and is used in proving \cref{theorem:trace-equivalence-ETR-hardness} and \cref{theorem:tvdistance-lt-one-ETR-hardness}.

\begin{proposition} \label{proposition:NMF-yes-iff-tvdistance-zero}
	The NMF instance is a yes-instance if and only if there is a memoryless strategy $\alpha$ such that $d_{\tv}(\mu,\nu) = 0$ in $\D(\alpha)$.
\end{proposition}

\begin{proof}
	($\impliedby$)
	Assume there is a memoryless strategy $\alpha$ such that in the induced Markov chain $d_{\tv}(\mu,\nu) = 0$, that is, we have $\mathrm{Pr}_{\mu}(\Run(w))=\mathrm{Pr}_{\nu}(\Run(w))$ for all words $w \in L^{*}$ . For all $1\leq i \leq n$,  $1 \leq k\leq r$ and $1\leq j\leq m$, let
	\[A[i,k]=\alpha(t_{i})(\m_{i,k}) \quad \quad \text{ and } \quad \quad W[k,j]=\alpha(t_k')(\m'_{k,j}).\]
	In the LMC $\D(\alpha)$, for all $1\leq i\leq n$ and all $1\leq j\leq m$, we have
	\begin{align*}
		\mathrm{Pr}_{\mu}(\Run(ca_icb_j)) =& \frac{1}{n} J[i,j]  \quad  \text{ and } \\
		\mathrm{Pr}_{\nu}(\Run(ca_icb_j)) =& \frac{1}{n} \sum_{k=1}^{r}\alpha(t_i)(\m_{i,k})\cdot \alpha(t_k')(\m'_{k,j}) = \frac{1}{n} \sum_{k=1}^{r} A[i, k] \cdot W[k, j].
	\end{align*}
	
	For all $i,j$ we have $\mathrm{Pr}_{\mu}(\Run(ca_icb_j)) = \mathrm{Pr}_{\nu}(\Run(ca_icb_j))$. Thus, we have
	$\sum_{k=1}^{r} A[i,k]\cdot W[k,j] =J[i,j]$ for all $i,j$.
	
	($\implies$)
	Assume the NMF instance is a yes-instance, that is, $\sum_{k=1}^{r} A[i,k]\cdot W[k,j] =J[i,j]$ for all $i, j$. We construct a memoryless strategy $\alpha$ such that $d_{\tv}(\mu,\nu) = 0$ in $\D(\alpha)$. For all state~$s' \in S$ and $\m \in \A$, the strategy~$\alpha$ is defined by
	\[
	\alpha(s')(\m) =
	\left\{
	\begin{array}{ll}
	A[i,k]   & \mbox{if } s'=t_{i} \text{ and } \m = \m_{i,k} \text{ where } 1 \leq i\leq n \text{ and } 1 \leq k\leq r\\
	&\\
	W[k,j]  & \mbox{if } s'=t_{k}'  \text{ and } \m =\m'_{k,j} \text{ where } 1\leq k\leq r \text{ and } 1 \leq j \leq m\\
	&\\
	1  & \mbox{if $\m$ is the only action available to $s'$}\\
	&\\
	0  & \mbox{otherwise. } \\   
	\end{array}
	\right.
	\]
	
	Let $k \in \nat$. Define $w_k$ to be the word $c a_{i_k} c b_{j_k}$ where $1 \le i_{k} \le n$ and $1 \le j_{k} \le m$. To show that $\mathrm{Pr}_{\mu}(\Run(w)) = \mathrm{Pr}_{\nu}(\Run(w))$ for all $w \in L^{*}$, it suffices to show for all $k \in \nat$,  we have: 
	\begin{itemize}
		\item[-] $\Pr_{\mu}(\Run(w_1 \cdots w_k))= \Pr_{\nu}(\Run(w_1 \cdots w_k))  = \frac{1}{n^{k}}\prod_{k' = 1}^{k}J[i_{k'}, j_{k'}]$;
		\item[-]  $\Pr_{\mu}(\Run(w_1 \cdots w_kc)) =  \Pr_{\nu}(\Run(w_1 \cdots w_kc)) $;
		\item[-] $\Pr_{\mu}(\Run(w_1 \cdots w_kca_{i_{k+1}})) = \Pr_{\nu}(\Run(w_1 \cdots w_kca_{i_{k+1}})) $;
		\item[-] $\Pr_{\mu}(\Run(w_1 \cdots w_kca_{i_{k+1}}c))  = \Pr_{\nu}(\Run(w_1 \cdots w_kca_{i_{k+1}}c))$.
	\end{itemize}
	
	We prove the statement by induction on $k$. The base case is $k = 0$. We have $\Pr_{\mu}(\Run(\varepsilon)) = \Pr_{\nu}(\Run(\varepsilon)) = \Pr_{\mu}(\Run(c)) = \Pr_{\nu}(\Run(c)) = 1$ and  $\Pr_{\mu}(\Run(ca_{i_1})) = \Pr_{\nu}(\Run(ca_{i_1})) = \Pr_{\mu}(\Run(ca_{i_1}c)) = \Pr_{\nu}(\Run(ca_{i_1}c))  =\frac{1}{n}$. 
	
	For the induction step, assume the statement holds for all $k' \le k$. By the induction hypothesis, we have: 
	\begin{align}
	\mu M_\alpha(w_1 \cdots w_k) =& \big(\frac{1}{n^{k}}\textstyle\prod_{k' = 1}^{k}J[i_{k'}, j_{k'}]\big) \delta_{s}   = \big(\frac{1}{n^{k}}\prod_{k' = 1}^{k}J[i_{k'}, j_{k'}]\big) \mu \text{ and } \label{equation:Dirac-distribution-on-s}\\
	\nu M_\alpha(w_1 \cdots w_k) =& \big(\frac{1}{n^{k}}\textstyle\prod_{k' = 1}^{k}J[i_{k'}, j_{k'}]\big) \delta_{t} = \big(\frac{1}{n^{k}}\prod_{k' = 1}^{k}J[i_{k'}, j_{k'}]\big) \nu \label{equation:Dirac-distribution-on-t}
	\end{align}
	
	First, we show that 
	\begin{equation}\label{equation:Dirac-distribution-on-s-t-induction}
	\textstyle\Pr_{\mu}(\Run(w_1 \cdots w_{k+1}))  = \textstyle\Pr_{\nu}(\Run(w_1 \cdots w_{k+1})) = \frac{1}{n^{k+1}}\textstyle\prod_{k' = 1}^{k+1}J[i_{k'}, j_{k'}]. 
	\end{equation} We have
	\begin{align*}
		&\textstyle\Pr_{\mu}(\Run(w_1 \cdots w_{k+1})) \\
		%&=& \textstyle\Pr_{\mu}(\Run(w_1 \cdots w_{k}ca_{i_{k+1}}cb_{j_{k+1}})) \commenteq{$w_{k+1} = 
		=& |\mu M_{\alpha} (w_1 \cdots w_{k+1})| \\ 
		=& |\mu M_{\alpha} (w_1 \cdots w_{k}) M_{\alpha} (w_{k+1}) | \\ 
		=& |\big(\frac{1}{n^{k}}\textstyle\prod_{k' = 1}^{k}J[i_{k'}, j_{k'}] \big) \mu M_{\alpha} (w_{k+1})| \commenteq{\eqref{equation:Dirac-distribution-on-s}}\\
		=& \big(\frac{1}{n^{k}}\textstyle\prod_{k' = 1}^{k}J[i_{k'}, j_{k'}] \big)  |\mu M_{\alpha} (w_{k+1})| \\
		=& \big(\frac{1}{n^{k}}\textstyle\prod_{k' = 1}^{k}J[i_{k'}, j_{k'}] \big)  \frac{1}{n} J[i_{k+1}, j_{k+1}]\commenteq{induction hypothesis}\\
		=& \frac{1}{n^{k+1}}\textstyle\prod_{k' = 1}^{k+1}J[i_{k'}, j_{k'}]  
	\end{align*}
	
	Similarly, 
	\begin{align*}
		&\textstyle\Pr_{\nu}(\Run(w_1 \cdots w_{k+1})) \\
		=& |\nu M_{\alpha} (w_1 \cdots w_{k+1})| \\ 
		=& |\nu M_{\alpha} (w_1 \cdots w_{k}) M_{\alpha} (w_{k+1}) | \\ 
		=& |\big(\frac{1}{n^{k}}\textstyle\prod_{k' = 1}^{k}J[i_{k'}, j_{k'}] \big) \nu M_{\alpha} (w_{k+1})| \commenteq{\eqref{equation:Dirac-distribution-on-t}}\\
		=& \big(\frac{1}{n^{k}}\textstyle\prod_{k' = 1}^{k}J[i_{k'}, j_{k'}] \big)  |\nu M_{\alpha} (w_{k+1})| \\
		=& \big(\frac{1}{n^{k}}\textstyle\prod_{k' = 1}^{k}J[i_{k'}, j_{k'}] \big)  \frac{1}{n} J[i_{k+1}, j_{k+1}]\commenteq{induction hypothesis}\\
		=& \frac{1}{n^{k+1}}\textstyle\prod_{k' = 1}^{k+1}J[i_{k'}, j_{k'}]. 
	\end{align*}
	
	By equation~\eqref{equation:Dirac-distribution-on-s-t-induction}, we have 
	\begin{align}
	\mu M_\alpha(w_1 \cdots w_{k+1}) =& \big(\frac{1}{n^{k+1}}\textstyle\prod_{k' = 1}^{k+1}J[i_{k'}, j_{k'}]\big) \delta_{s}   = \big(\frac{1}{n^{k+1}}\textstyle\prod_{k' = 1}^{k+1}J[i_{k'}, j_{k'}]\big) \mu \text{ and } \label{equation:Dirac-distribution-on-s-induction}\\
	\nu M_\alpha(w_1 \cdots w_{k+1}) =& \big(\frac{1}{n^{k+1}}\textstyle\prod_{k' = 1}^{k+1}J[i_{k'}, j_{k'}]\big) \delta_{t} = \big(\frac{1}{n^{k}}\textstyle\prod_{k' = 1}^{k}J[i_{k'}, j_{k'}]\big) \nu \label{equation:Dirac-distribution-on-t-induction}
	\end{align}
	
	Thus, 
	\begin{align*}
		&\textstyle\Pr_{\mu}(\Run(w_1 \cdots w_{k+1}c)) \\
		=& |\mu M_{\alpha} (w_1 \cdots w_{k+1}c)| \\ 
		=& |\mu M_{\alpha} (w_1 \cdots w_{k+1}) M_{\alpha} (c) | \\ 
		=& |\big(\frac{1}{n^{k+1}}\textstyle\prod_{k' = 1}^{k+1}J[i_{k'}, j_{k'}] \big) \mu M_{\alpha} (c)| \commenteq{\eqref{equation:Dirac-distribution-on-s-induction}}\\
		=& \big(\frac{1}{n^{k+1}}\textstyle\prod_{k' = 1}^{k+1}J[i_{k'}, j_{k'}] \big)  |\mu M_{\alpha} (c)| \\
		=& \frac{1}{n^{k+1}}\textstyle\prod_{k' = 1}^{k+1}J[i_{k'}, j_{k'}]  
	\end{align*}
	
	Similarly, $\textstyle\Pr_{\nu}(\Run(w_1 \cdots w_{k+1}c)) = \frac{1}{n^{k+1}}\textstyle\prod_{k' = 1}^{k+1}J[i_{k'}, j_{k'}]$. We also have  
	\begin{align*}
		&\textstyle\Pr_{\mu}(\Run(w_1 \cdots w_{k+1}ca_{i_{k+2}})) \\
		=& |\mu M_{\alpha} (w_1 \cdots w_{k+1}ca_{i_{k+2}})| \\ 
		=& |\mu M_{\alpha} (w_1 \cdots w_{k+1}) M_{\alpha} (ca_{i_{k+2}}) | \\ 
		=& |\big(\frac{1}{n^{k+1}}\textstyle\prod_{k' = 1}^{k+1}J[i_{k'}, j_{k'}] \big) \mu M_{\alpha} (ca_{i_{k+2}})| \commenteq{\eqref{equation:Dirac-distribution-on-s-induction}}\\
		=& \big(\frac{1}{n^{k+1}}\textstyle\prod_{k' = 1}^{k+1}J[i_{k'}, j_{k'}] \big)  |\mu M_{\alpha} (ca_{i_{k+2}})| \\
		=& \frac{1}{n^{k+2}}\textstyle\prod_{k' = 1}^{k+1}J[i_{k'}, j_{k'}]  
	\end{align*}
	
	Similarly, $\textstyle\Pr_{\nu}(\Run(w_1 \cdots w_{k+1}ca_{i_{k+2}})) =  \Pr_{\mu}(\Run(w_1 \cdots w_{k+1}ca_{i_{k+2}}c)) = \textstyle\Pr_{\nu}(\Run(w_1 \cdots w_{k+1}ca_{i_{k+2}}c)) = \frac{1}{n^{k+2}}\textstyle\prod_{k' = 1}^{k+1}J[i_{k'}, j_{k'}]$. \qedhere
\end{proof}

\begin{theorem}
	\label{theorem:trace-equivalence-ETR-hardness}
		The NMF problem is polynomial-time reducible to the problem $\TVeqZERO$, hence $\TVeqZERO$ is $\ETR$-hard.
\end{theorem}
\begin{proof}
	Proposition~\ref{proposition:NMF-yes-iff-tvdistance-zero} shows that the NMF problem is polynomial-reducible to the problem $\TVeqZERO$.  Since the NMF problem is $\ETR$-complete \cite{Shitov2016}, we have that the problem $\TVeqZERO$ is $\ETR$-hard.
\end{proof}

\theoremTVDistanceZeroUBLB*
\begin{proof}
	It follows from \cref{theorem:trace-equivalence-reduce-to-ETR} and \cref{theorem:trace-equivalence-ETR-hardness}.
\end{proof}

%TV<1
\subsection{Proofs of $\TVneqONE$}\label{appendix:tvDistanceNeqOne}
 Next, we show that the problem $\TVneqONE$ is $\ETR$-complete. Recall that $\TVneqONE$ is the problem asking whether there is a memoryless strategy $\alpha$ for $\D$ such that the total variation distance of the two initial distributions is less than one in the induced labelled Markov chain $\D(\alpha)$, i.e., $d_{\tv}(\mu, \nu) \ls 1$.

With \cref{proposition:tvdistance-one-theorem} at hand, we obtain:

\begin{theorem} \label{theorem:tvdistance-lt-one-reduce-to-ETR}
	The problem $\TVneqONE$ is in $\ETR$.%the existential theory of the reals.
\end{theorem}

\begin{proof}
	A memoryless strategy $\alpha$ for $\D$ can be characterised by numbers $x_{s,\m} \in [0,1]$ where $s \in S$ and $\m \in \A$ such that $x_{s,\m} = \alpha(s)(\m)$. We write $\bar{x}$ for the collection $(x_{s,\m})_{s\in S, \m \in \A}$. 
	
	From \cref{proposition:tvdistance-one-theorem}, to check whether there is a memoryless strategy $\alpha$ such that $d_{\tv}(\mu , \nu) \ls 1$, it suffices to check if there are subdistributions $\mu_1$ and $\mu_2$ that satisfy Equation~\eqref{equation:conditions-vardistance-less-than-one}. Thus, we can nondeterministically guess $r_1$ and support of $\mu_2$ such that $\support(\mu_2) \subseteq R_{r_1}^{\mu,\nu}$, and then check the following decision problem, which is a closed formula in the existential theory of the reals: 
	
	 $\exists \bar{x}$, matrices $B(a) \in \mathbb{R}^{S \times S}$ for all $a \in L$, a matrix $F \in \mathbb{R}^{S \times S}$, subdistributions $\mu_1$ and $\mu_2$ such that 
	%compute $R^{\mu, \nu}$ with a guessed support of $\alpha$, guess a state $r_1 \in S$ and check feasibility for (\ref{equation:conditions-vardistance-less-than-one}). By Proposition~\ref{prop:language-equivalence}, the conditions in (\ref{equation:conditions-vardistance-less-than-one}) can be characterized as the following decision problem: 
	
	%$\exists \bar{x}$, matrices $B(a) \in \mathbb{R}^{S \times S}$ for all $a \in L$, a matrix $F \in \mathbb{R}^{S \times S}$, subdistributions $\mu_1$ and $\mu_2$ such that 
	\begin{itemize}
		\item[-]
		$\sum_{\m \in \A(s)}x_{s, \m} = 1$ for all $s \in S$;
		\item[-]
		the first row of $F$ is $\mu_1 - \mu_2$;
		\item[-]
		$F \mathbf{1} = \mathbf{0}$;
		\item[-]
		$F  M(a) = B(a) F$ for all $a \in L$;
		\item[-]
		$r_1 \in \support(\mu_1)$;
		\item[-]
		$\support(\mu_2)  \subseteq  R^{\mu, \nu}_{r_1}$.
	\end{itemize}
	
	It follows that the problem is in $\ETR$ since $\ETR$ is closed under {\sf NP}-reductions~\cite{CateKO13}.
\end{proof}

% explain what is NMF
To show that the problem $\TVneqONE$ is hard for $\ETR$, we present the reduction from the nonnegative maitrx factorization (NMF) problem. We construct the same MDP $\D$ as shown in Figure~\ref{fig:reductionfromNMFMDP}. The reduction is similar to \cite[Theorem~4.5]{FKS2020}.

The proposition below is technical and is only used in the proof of \cref{theorem:tvdistance-lt-one-ETR-hardness}.
\begin{proposition} \label{prop:NMF-no-to-inequivalent-distributions}
	If the NMF instance is a no-instance then for all memoryless strategy $\alpha$ and all (sub)distributions $\mu_1$ over the left part of $\D$ and all (sub)distributions $\mu_2$ over the right part, we have $d_{\tv}(\mu_1,\mu_2) \gr 0$ in the LMC $\D(\alpha)$.
\end{proposition}
\begin{proof}
	Let $\mu_1$ and $\mu_2$ be two (sub)distributions where $\mu_1$ is over the left part of $\D$ and $\mu_2$ is over the right part. Let the NMF instance be a no-instance. Let $\alpha$ be any memoryless strategy. 
	
	By the construction of the MDP $\D$ (see Figure~\ref{fig:reductionfromNMFMDP}), there must exist a word $w' \in L^{*}$ such that $\mu_1M_{\alpha}(w')$ is a Dirac distribution on state $s$. Let $\mu_1' = \mu_1M_{\alpha}(w')$ and $\mu_2'= \mu_2M_{\alpha}(w')$. We have that $\mu_1' =  \mu_1'(s) \delta_s$. We distinguish the following three cases. %We have that either $\mu_1'(s)  \not= \mu_2'(t)$ or $\mu_1'(s) = \mu_2'(t) \gr 0$. 
	\begin{enumerate}[label=(\alph*)]
		\item Assume $|\mu_1'| \not= |\mu_2'|$.  
		Let  $E' =  L^{\omega}$. We have
		\[
		\begin{array}{rcl}
		d_{\tv}(\mu_1', \mu_2')	& = & \sup_{E \in \mathcal{F}} |\textstyle\Pr_{\mu_1'}(E) - \textstyle\Pr_{\mu_2'}(E)| \\
		& \ge & |\Pr_{\mu_1'}(E') - \Pr_{\mu_2'}(E')|  \commenteq{$E' \in \mathcal{F}$}\\
		&=& \big||\mu_1'| - |\mu_2'|\big|\\
		& \gr & 0 \commenteq{$|\mu_1'| \not= |\mu_2'|$}\\
		\end{array}
		\]
		\item Assume $|\mu_1'| = |\mu_2'|$ and $\mu_1'(s) \not= \mu_2'(t)$. 
		Let  $E' = \Run(ca_i) \in \mathcal{F}$. We have
		\[
		\begin{array}{rcl}
		d_{\tv}(\mu_1', \mu_2')
		& = & \sup_{E \in \mathcal{F}} |\textstyle\Pr_{\mu_1'}(E) - \textstyle\Pr_{\mu_2'}(E)| \\
		& \ge & |\Pr_{\mu_1'}(E') - \Pr_{\mu_2'}(E')|  \commenteq{$E' \in \mathcal{F}$}\\
		& = & |\Pr_{\mu_1'}(\Run(ca_i)) - \Pr_{\mu_2'}(\Run(ca_i))| \commenteq{$E' = \Run(ca_i)$}\\
		& = & \big||\mu_1'M_{\alpha}(ca_i)| - |\mu_2'M_{\alpha}(ca_i)|\big|\\
		& = & |\frac{1}{n}\mu_1'(s) - \frac{1}{n}\mu_2'(t)| \\
		& \gr & 0 \commenteq{$\mu_1'(s) \not= \mu_2'(t)$}\\
		\end{array}
		\]
		%Since $|\mathrm{Pr}_{\mu_1}(\Run(w))| \not= |\mathrm{Pr}_{\mu_2}(\Run(w))|$, we have that $\mu_1 \not\equiv \mu_2$.
		\item Assume $|\mu_1'| = |\mu_2'|$ and $\mu_1'(s) = \mu_2'(t) \gr 0$.
		
		Since $\mu_1' =  \mu_1'(s) \delta_s$, $|\mu_1'| = |\mu_2'|$ and $\mu_2'(t) = \mu_2'(s) \gr 0$, we have that $\mu_2' =  \mu_2'(t) \delta_t$.
		
		By Proposition~\ref{proposition:NMF-yes-iff-tvdistance-zero}, we have that if the NMF instance is a no-instance, then $d_{\tv}(\mu, \nu) \gr 0$, that is, there exists a word $w \in L^{*}$ such that $\Pr_{\mu}(\Run(w)) \not= \Pr_{\nu}(\Run(w))$. 
		This word $w$ is of the form $ca_{i_1}cb_{j_1}ca_{i_2}\cdots$, since it is emitted by running the MDP $\D$ from state $s$.
		
		Let  $E' = \Run(w) \in \mathcal{F}$. We have
		\[
		\begin{array}{rcl}
		d_{\tv}(\mu_1', \mu_2')
		& = & \sup_{E \in \mathcal{F}} |\textstyle\Pr_{\mu_1'}(E) - \textstyle\Pr_{\mu_2'}(E)| \\
		& \ge & |\Pr_{\mu_1'}(E') - \Pr_{\mu_2'}(E')|  \commenteq{$E' \in \mathcal{F}$}\\
		& = & |\Pr_{\mu_1'}(\Run(w)) - \Pr_{\mu_2'}(\Run(w))| \commenteq{$E' = \Run(w)$}\\
		& = & \big| |\mu_1'M_{\alpha}(w)| - |\mu_2'M_{\alpha}(w)| \big| \\
		& = & \big| |\mu_1'(s)\delta_s M_{\alpha}(w)| - |\mu_2'(t)\delta_t M_{\alpha}(w)| \big| \commenteq{$\mu_1' =  \mu_1'(s) \delta_s$ and $\mu_2' =  \mu_2'(t) \delta_t$} \\ %$w = ca_{i_1}cb_{j_1}ca_{i_2}\cdots$
		& = & |\Pr_{\mu_1'(s)\delta_s}(\Run(w)) - \Pr_{\mu_2'(t)\delta_t}(\Run(w))| \\
		& = & |\Pr_{\mu_1'(s)\mu}(\Run(w)) - \Pr_{\mu_2'(t)\nu}(\Run(w))| \commenteq{$\mu = \delta_s$ and $\nu = \delta_t$}\\
		& = & |\mu_1'(s) \Pr_{\mu}(\Run(w)) - \mu_2'(t)\Pr_{\nu}(\Run(w))| \\
		& = & \mu_1'(s) |\Pr_{\mu}(\Run(w)) - \Pr_{\nu}(\Run(w))| \commenteq{$\mu_1'(s) = \mu_2'(t)$}\\
		& \gr&  0 \commenteq{$\mu_1'(s) \gr 0$ and $\Pr_{\mu}(\Run(w)) \not= \Pr_{\nu}(\Run(w))$}
		\end{array}
		\]
	\end{enumerate}
	
	Following the three cases, we have $d_{\tv}(\mu_1',\mu_2') \gr 0$, that is, there exists a word $w \in L^{*}$ such that $\Pr_{\mu_1'}(\Run(w)) \not= \Pr_{\mu_2'}(\Run(w))$. Consider the word $w'w$, we have
	\[
	\begin{array}{rcl}
	\Pr_{\mu_1}(\Run(w'w)) 
	&=& |\mu_1 M_\alpha(w'w)| \\
	&=& |\mu_1 M_{\alpha}(w') M_\alpha(w)|\\
	&=& \Pr_{\mu_1M_{\alpha}(w')}(\Run(w) \\
	&=& \Pr_{\mu_1'}(\Run(w)) \commenteq{$\mu_1' = \mu_1M_{\alpha}(w')$} \\
	&\not=& \Pr_{\mu_2'}(\Run(w)) \\
	&=& \Pr_{\mu_2M_{\alpha}(w')}(\Run(w)) \commenteq{$\mu_2'= \mu_2M_{\alpha}(w')$}\\
	&=& |\mu_2 M_{\alpha}(w') M_{\alpha}(w)|\\
	&=& |\mu_2 M_{\alpha}(w'w)| \\
	&=& \Pr_{\mu_2}(\Run(w'w)).
	\end{array}
	\]
	
	Thus, we have $d_{\tv}(\mu_1, \mu_2) \gr 0$.
\end{proof}

\begin{theorem}\label{theorem:tvdistance-lt-one-ETR-hardness}
	The NMF problem is polynomial-time reducible to the problem $\TVneqONE$, hence $\TVneqONE$ is $\ETR$-hard.
\end{theorem}

\begin{proof}
	According to \cref{proposition:NMF-yes-iff-tvdistance-zero}, we have that if the NMF instance is a yes-instance then there is a memoryless strategy such that $d_{\tv}(\mu, \nu) = 0$ in the induced LMC, which implies $d_{\tv}(\mu, \nu) \ls 1$.
	
	It remains to show that if there is a memoryless strategy such that $d_{\tv}(\mu, \nu) \ls 1$, then the NMF instance is a yes-instance. We show the contrapositive, that is,  if the NMF instance is a no-instance, then for all memoryless strategy $d_{\tv}(\mu, \nu) = 1$ in the induced LMC. For all $w \in L^{*}$ and memoryless strategy $\alpha$, we have that if $|\mu M_{\alpha}(w)| \gr 0$ and $|\nu M_{\alpha}(w)| \gr 0$ then $\mu M_{\alpha}(w)$ and $\nu M_{\alpha}(w)$ are subdistributions over the left and right part of $\D$, respectively. It follows from Proposition~\ref{prop:NMF-no-to-inequivalent-distributions} that  $d_{\tv}(\mu_1, \mu_2) \gr 0$ for all subdistributions $\mu_1, \mu_2$ over the left and right part of $\D$, respectively. By Proposition~\ref{proposition:tvdistance-one}, we have that $d_{\tv}(\mu,\nu) = 1$ in all LMC~$\D(\alpha)$. 
\end{proof}

\theoremTVDistanceNeqOneUBLB*
\begin{proof}
	It follows from \cref{theorem:tvdistance-lt-one-reduce-to-ETR} and \cref{theorem:tvdistance-lt-one-ETR-hardness}.
\end{proof}

%PB=0 
\subsection{Proofs of $\PBeqZERO$}\label{appendix:pbDistanceZERO}
Next, we show that the problem $\PBeqZERO$ is {\sf NP}-complete. Recall that $\PBeqZERO$ is the problem asking whether there is a memoryless strategy $\alpha$ for $\D$ such that the probabilistic bisimilarity distance of the two initial states is zero in the induced labelled Markov chain $\D(\alpha)$, i.e., $d_{\pb}(s, t) = 0$.

\begin{theorem} {\normalfont\cite[Theorem~1]{DGJP1999}}\label{theorem:prob-bisimilar-iff-pbdistance-zero}
	For all $s,t \in S$, $s \sim t$ if and only if $d_{\pb}(s,t) = 0$.
\end{theorem}

\begin{theorem}\label{theorem:pbdistance-zero-NP}
	The problem $\PBeqZERO$ is in $\sf NP$.
\end{theorem}

\begin{proof}
	According to \cref{theorem:prob-bisimilar-iff-pbdistance-zero} and the definition of probabilistic bisimilarity, there exists a memoryless strategy $\alpha$ such that $d_{\pb}(s, t) = 0$ in the induced LMC $\D(\alpha)$, if and only if the initial states $s$ and $t$ are probabilistic bisimilar, i.e., $s$ and $t$ are in the same probabilistic bisimulation induced equivalence class. 
	
	We can nondeterministically guess a partition $E_1, \cdots, E_n$ of $S$ such that each subset $E_i$ is a probabilistic bisimulation induced equivalence class and state $s, t$ are in the same equivalence class, that is, $\bigcup_{E_i} = S$, $E_i \cap E_j = \emptyset$ for any $i \not= j$, and $s, t \in E_i$ for some $i$.  
	
	A memoryless strategy $\alpha$ for $\D$ can be characterised by numbers $x_{s,m} \in [0,1]$ where $s \in S$ and $\m \in \A$ such that $x_{s,\m} = \alpha(s)(\m)$. We write $\bar{x}$ for the collection $(x_{s,\m})_{s\in S, \m \in \A}$. To check $d_{\pb}(s, t) = 0$ in the induced LMC $\D(\alpha)$ amounts to a feasibility test of the following linear program:
	\begin{align*}
	\exists \bar{x} \text{ such that } &\textstyle\sum_{\m\in\A(s)} x_{s, \m} = 1 \text{ for all $s\in S$ and}\\
	&\tau(s')(E_j) = \tau(t')(E_j) \text{ for all $E_i, E_j$ and all $s', t' \in E_i$}, 
	\end{align*}
	and hence can be decided in polynomial time.
\end{proof}

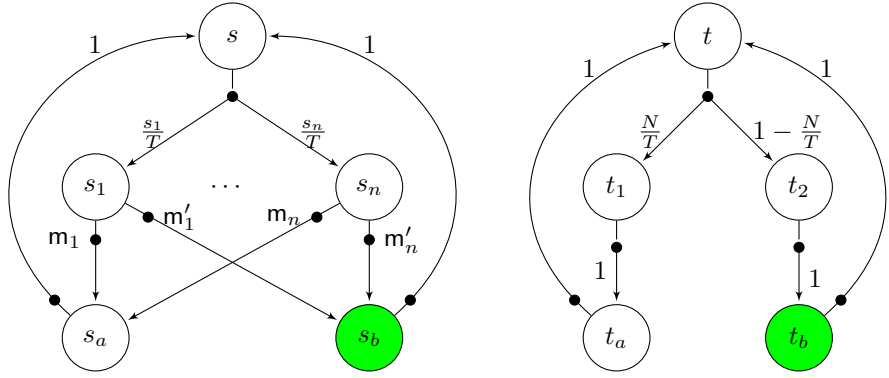
\begin{figure}[h]
	\begin{minipage}{0.45\linewidth}
		\centering
		\centering

\tikzstyle{BoxStyle} = [draw, circle, fill=black, scale=0.4,minimum width = 1pt, minimum height = 1pt]
\begin{tikzpicture}[xscale=.6,>=latex',shorten >=1pt,node distance=3cm,on grid,auto]
 
 \node[state] (us) at (3,4) {$s$};
 \node[BoxStyle] (st) at (3,3.2){};

 \node[state] (u1) at (0,2) {$s_{1}$};
 \node[label] at (2.9,2) {$\cdots$};
 \node[state] (un) at (6,2) {$s_{n}$};

 \node[BoxStyle,label=left:{{$\m_{1}$}}] at (0,1.3){};
 \node[BoxStyle,label=right:$\m_{1}'$] at (1.15,1.6){};
 \node[BoxStyle,label=left:{{$\m_{n}$}}] at (4.85,1.6){};
 \node[BoxStyle,label=right:$\m_{n}'$] at (6,1.3){};

 \node[state] (qb) at (0,0) {$s_a$};
 \node[state, fill=green] (qc) at (6,0) {$s_b$};
 
 \node[BoxStyle] at (-.9,0.5){};
 \node[BoxStyle] at (6.9,0.5){};

%\path[->] (qc) edge [loop right] node [right, midway] {$c,1$}  (qc);
%\path[->] (qb) edge [loop left] node [left, midway] {$b,1$}  (qb);
 \path[-] (us) edge node  [near end,left] {} (st);
 \path[->] (st) edge node  [midway,left] {$\frac{s_1}{T}$} (u1);
 \path[->] (st) edge node  [midway,right] {$\frac{s_n}{T}$} (un);

 \path[->] (u1) edge node  [very near end,left] {} (qb);
 \path[->] (u1) edge node  [very near end,above] {} (qc);
 \path[->] (un) edge node [very near end,above] {} (qb);
 \path[->] (un) edge node [very near end,right] {} (qc);

 %\path[->] (qc) edge  [out=15,in=-15,looseness=8] node [right] {$1$} (qc);
 %\path[->] (qb) edge [out=165,in=195,looseness=8] node [left]{$1$} (qb);

 \draw [->,postaction={decorate}] (qb) to [out=150, in=180, looseness=1.8] node [near end, above] {$1$} (us);
 \draw [->,postaction={decorate}] (qc) to [out=30, in=0, looseness=1.8] node [near end, above] {$1$} (us);
\end{tikzpicture}
	\end{minipage}
	\hspace{0.5cm}
	\begin{minipage}{0.45\linewidth}
		\centering
		\centering
\tikzstyle{BoxStyle} = [draw, circle, fill=black, scale=0.4,minimum width = 1pt, minimum height = 1pt]

\begin{tikzpicture}[xscale=.6,>=latex',shorten >=1pt,node distance=3cm,on grid,auto]

 \node[state] (t) at (2,4) {$t$};
 \node[BoxStyle] (tt) at (2,3.2){};

 \node[state] (t1) at (0,2) {$t_1$};
 \node[state] (t2) at (4,2) {$t_2$};
 
 \node[BoxStyle] (tat) at (0,1.2){};
 \node[BoxStyle] (tbt) at (4,1.2){};

 \node[state] (qb) at (0,0) {$t_a$};
 \node[state, fill=green] (qc) at (4,0) {$t_b$};

 \node[BoxStyle]at (-.92,0.5){};
 \node[BoxStyle] at (4.92,0.5){};
 
 \path[-] (t) edge node [midway, right] {} (tt);
 \path[->] (tt) edge node [midway, left] {$\frac{N}{T}$} (t1);
 \path[->] (tt) edge node [midway, right] {$1-\frac{N}{T}$} (t2);
 \path[-] (t1) edge node [midway, left] {} (tat);
 \path[->] (tat) edge node [midway, left, yshift=0.1cm] {$1$} (qb);

  \path[-] (t2) edge node [midway, left] {} (tbt);
 \path[->] (tbt) edge node [midway, right] {$1$} (qc);

 \draw [->,postaction={decorate}] (qb) to [out=150, in=190, looseness=1.8] node [near end, above] {$1$} (t);
 \draw [->,postaction={decorate}] (qc) to [out=30, in=-10, looseness=1.8] node [near end, above] {$1$} (t);

\end{tikzpicture}

 
	\end{minipage}
	\caption{The MDP~$\mathcal{D}$ in the reduction for {\sf NP}-hardness of $\PBeqZERO$ (or $\PBneqONE$). All states have the same label $a$ except $s_b$ and $t_b$ which have label $b$.  
	}\label{fig:reductionfromSubsset}
\end{figure}

Given a set $S =\{s_1, \cdots, s_n\}$ and $N \in \nat$, Subset Sum asks whether there exists a set $P \subseteq S$ such that $\sum_{s_i \in P} s_i = N$. 

\begin{theorem}\label{theorem:pbdistance-zero-NP-hard}
	The Subset Sum problem is polynomial-time many-one reduction to $\PBeqZERO$, hence $\PBeqZERO$ is {\sf NP}-hard. 
\end{theorem}

\begin{proof}
	Given an instance of Subset Sum $<S, N>$ where $S =\{s_1, \cdots, s_n\}$ and $N \in \nat$, we construct an MDP $\D$; see \cref{fig:reductionfromSubsset}. Let $T = \sum_{s_i \in S} s_i$. In the MDP $\D$, state $s$ transitions to state $s_i$ with probability $s_i / T$ for all $1 \le  i \le n$. Each state $s_i$ has two available actions, each transitions to $s_a$ and $s_b$ by taking the action $\m_i$ and $\m_i'$, respectively. State $t$ transitions to $t_1$ and $t_2$ with probability $N / T$ and $1 - N / T$, respectively. All the remaining states have only one available action transitioning to the successor state with probability one. States $s_b$ and $t_b$ have label $b$ and all other states have label $a$.  
	
	Next, we show that
	$$<S, N> \in {\mbox{Subset Sum}} \iff \exists \, \mbox{memoryless strategy $\alpha$ such that} \; d_{\pb}(s, t) = 0 \, \mbox{in} \; \D(\alpha).$$ 
	
	Intuitively, making $s_i$ probabilistic bisimilar with $t_1$ simulates the membership of $s_i$ in $P$. Conversely, making $s_i$ probabilistic bisimilar with $t_2$ simulates the membership of $s_i$ in $S \setminus P$. 
	
	($\implies$)
	Let $P \subseteq S$ be the set such that $\sum_{s_i \in P} s_i = N$. Let $\alpha$ be an MD strategy such that if $s_i \in P $ then $\alpha(s_i) = \m_i$ and $\alpha(s_i) = \m_i'$ otherwise. Consider the following partition of states of $\D$, 
	$$E_1 = \{s, t\}, E_2 = P \cup \{t_1\}, E_3 = (S\setminus P) \cup \{t_2\}, E_4 = \{s_a, t_a\} \text{ and } E_5 = \{s_b, t_b\}. $$
	
	Then, 
	\begin{align*}
		\tau(s)(E_2) &= \tau(s)(P)  = \sum_{s_i \in P} \frac{s_i}{T}\\    
		&  = \frac{N}{T}  && \commenteq{$\sum_{s_i \in P} s_i = N$} \\
		&= \tau(t)(t_1) = \tau(t)(E_2)  \quad \text{and}
	\end{align*}
	
	\begin{align*}
		\tau(s)(E_3) &= \tau(s)(S \setminus P)  = \sum_{s_i \in S \setminus P} \frac{s_i}{T} = 1 - \sum_{s_i \in P} \frac{s_i}{T} \\
		& = 1 -\frac{N}{T}  && \hspace{-15ex}\commenteq{$\sum_{s_i \in P} s_i = N$} \\
		&= \tau(t)(t_2) = \tau(t)(E_3). 
	\end{align*}
	
	Similarly, we can verify that 
	$\text{for all } E_i, E_j \text{ and all } s', t' \in E_i : \tau(s')(E_j) = \tau(t')(E_j).$
	
	By the definition of probabilistic bisimulation, each set $E_i$ is a probabilistic bisimulation induced equivalence class. Since $s$ and $t$ are in the same equivalence class, we have $s \sim t$, and hence $d_{\pb}(s, t) = 0$ by \cref{theorem:prob-bisimilar-iff-pbdistance-zero}.
	
	($\impliedby$)
	Assume there is a memoryless strategy $\alpha$ such that $d_{\pb}(s, t) = 0$ in the LMC $\D(\alpha)$. By \cref{theorem:prob-bisimilar-iff-pbdistance-zero}, $s$ and $t$ are probabilistic bisimilar in $\D(\alpha)$. Let $P$ be the set of successor states of $s$ that are probabilistic bisimilar to $t_1$. 	Then, 
	\begin{align*}
		\tau(s)(P) = \sum_{s_i \in P} \frac{s_i}{T} \quad \text{ and } \quad  \tau(t)(t_1)  =\frac{N}{T} . 
	\end{align*}
	
	Since $\ell(t_a) \not= \ell(t_b)$, by definition of probabilistic bisimilarity, we have $t_a \not\sim t_b$, and hence $t_a$ and $t_b$ are not in the same $\sim$-equivalence class. Since $\tau(t_1)(t_a)= \tau(t_2)(t_b) = 1$ in the LMC $\D(\alpha)$, again by  definition of probabilistic bisimilarity,  we have $t_1 \not\sim t_2$ and $t_1$ and $t_2$ are not in the same $\sim$-equivalence class, and thus $t_2$ is not in the same $\sim$-equivalence class as the states in $P$. Since $s \sim t$, we have $\sum_{s_i \in P} \frac{s_i}{T} = \frac{N}{T}$, and hence, $\sum_{s_i \in P} = N$.
\end{proof}

\theoremPBDistanceZeroUBLB*
\begin{proof}
	It follows from \cref{theorem:pbdistance-zero-NP} and \cref{theorem:pbdistance-zero-NP-hard}.
\end{proof}

%PB<1
\subsection{Proofs of $\PBneqONE$}\label{appendix:pbDistanceNeqOne}
We show in this section that the problem $\PBneqONE$ is {\sf NP}-complete. Recall that $\PBneqONE$ is the problem asking whether there is a memoryless strategy $\alpha$ for $\D$ such that the probabilistic bisimilarity distance of the two initial states is less than one in the induced labelled Markov chain $\D(\alpha)$, i.e., $d_{\pb}(s, t) \ls 1$.

%\begin{figure}[h]
%	\centering
%	\input{figures/exampleNoMDstrategyPBneqOne.tex}
%	\caption{The MDP~$\mathcal{D}$ that no MD strategy witnesses $d_{pb}(s,t) \ls 1$. All states have the same label except $s_b$ and $t_b$ where $s_b$ and $t_b$ have the same label.}\label{fig:noMDstrategyPBneqOne}
%\end{figure}	

%For an MDP, we might have some memoryless strategies $\alpha$ such that $d_{pb}(s,t) \ls 1$  in $\D(\alpha)$ but none of such strategies are MD ones. The MDP in \cref{fig:noMDstrategyPBneqOne} is such an example. It is easy to check that the strategy $\alpha$ such that $\alpha(s)(m_1) = \alpha(s)(m_2) = \frac{1}{2}$, where $m_1$ and $m_2$ are the two available actions of $s$, is the only strategy that makes $d_{pb}(s,t) = 0$ and $d_{pb}(s,t) \ls 1$. Thus, to show the {\sf NP} upper bound, we cannot simply guess an MD strategy for the MDP. 

\theoremPBDistanceNeqOneUBLB*

\begin{proof}
	We first show that this problem is in {\sf NP}. We nondeterministically guess a partition $E_1, \cdots, E_n$ of the states of $\D$ and two states $u, v$ in the same subset $E_i$ for some $i$. We also nondeterministically guess the graph $G$ of \cref{definition:pb-distance-graph} for the LMC $\D(\alpha)$ induced by some strategy $\alpha$. By \cref{proposition:pbdistance-neq-one-graph-reachability}, if $(s, t)$ is reachable from some state pair $(u, v)$ in the graph $G$ and $u \sim v$ then $d_{\pb}(s, t) \ls 1$. The condition that $(s, t)$ is reachable from $(u, v)$ in the graph $G$ can be checked in polynomial time using e.g. breadth-first search. To check $u \sim v$, it suffices to check each subset $E_i$ is a probabilistic bisimulation induced equivalence class, which amounts to a feasibility test of the linear program:
	\begin{align*}
	\exists \bar{x} \text{ such that } &\textstyle\sum_{\m\in\A(s)} x_{s, \m} = 1 \text{ for all $s\in S$ and}\\
	&\tau(s')(E_j) = \tau(t')(E_j) \text{ for all $E_i, E_j$ and all $s', t' \in E_i$}, 
	\end{align*}
	and hence can be decided in polynomial time.  
	
	Next, we establish {\sf NP}-hardness of the problem. Similar to \cref{theorem:pbdistance-zero-NP-hard}, we provide a polynomial-time many-one reduction from Subset Sum. Given an instance of $<S, N>$ of Subset Sum, we construct the same MDP $\D$ as shown in \cref{fig:reductionfromSubsset}.
	
	Next, we show that
	$$<S, N> \in {\mbox{Subset Sum}} \iff \exists \, \mbox{memoryless strategy $\alpha$ such that} \; d_{\pb}(s, t) \ls 1 \, \mbox{in} \; \D(\alpha).$$ 
	
	($\implies$)
	Let $P \subseteq S$ be the set such that $\sum_{s_i \in P} s_i = N$. By \cref{theorem:pbdistance-zero-NP-hard}, there exists a memoryless strategy $\alpha$ such that $d_{\pb}(s, t) = 0$, and hence $d_{\pb}(s, t) \ls 1$.
	
	($\impliedby$)
	We prove its contrapositive, that is, if the instance is a no-instance then for all memoryless strategy $\alpha$ we have $d_{\pb}(s, t) = 1$ in $\D(\alpha)$.
	
	Assume the instance is a no-instance. By \cref{theorem:pbdistance-zero-NP-hard}, we have $d_{\pb}(s, t) \gr 0$ in $\D(\alpha)$ for all memoryless strategy $\alpha$, i.e. $s \not\sim t$. Let $\alpha$ be an arbitrary memoryless strategy. By the construction of the MDP, we have $s_a \not\sim t_a$ and $s_b \not\sim t_b$. Since $\ell(s_a) \not= \ell(t_b)$ and $\ell(s_b) \not= \ell(t_a)$, we also have $s_a \not\sim t_b$ and $s_b \not\sim t_a$. Thus, in the LMC $\D(\alpha)$, $s_i$ for all $1 \le i \le n$ is not probabilistic bisimilar to $t_1$ or $t_2$. In the graph of \cref{definition:pb-distance-graph}, the following vertices could reach $(s, t)$: $(s_i, t_1)$ or $(s_i, t_2)$ for all $1 \le i \le n$, $(s_a, t_a)$ and $(s_b, t_b)$. However, none of them are probabilistic bisimilar.  By \cref{proposition:pbdistance-neq-one-graph-reachability}, we have $d_{\pb}(s, t) = 1$ in the LMC $\D(\alpha)$.
\end{proof}

\end{document}